\documentclass[11pt]{article}
\usepackage[utf8]{inputenc}
\usepackage[T1]{fontenc}
\usepackage[english]{babel}
\usepackage[]{algorithm2e}
\usepackage{mwe}
\usepackage{amsmath,mathtools}
\usepackage[dvipsnames]{xcolor}

\usepackage{amsthm,bm}
\usepackage{lineno}

\usepackage{tikz,pgfplots}
\pgfplotsset{compat=1.5}
\usepackage{subfigure}
\usetikzlibrary{intersections}
\usetikzlibrary{patterns}
\usepackage{multirow}
\usepackage{tkz-tab}

\usepackage[shortlabels]{enumitem}
\usepackage[colorlinks=true, linkcolor=blue, citecolor=blue, hypertexnames=false]{hyperref}  
\usepackage{nameref}
\usepackage{cleveref}

\usepackage{color}              %
\usepackage{color-edits}
\addauthor{Omar}{magenta}
\addauthor{Humberto}{blue}
\addauthor{Ilan}{red}

\usepackage[right=1.0in, top=1in, bottom=1.0in, left=1.0in]{geometry}

\usepackage{lmodern}

\usepackage{graphicx}
\usepackage{lscape}
\usepackage[final]{pdfpages}
\usepackage{amsthm}
\usepackage[authoryear,round]{natbib}
\usepackage{amsfonts}
\usepackage{mathtools}
\usepackage{bbm}
\usepackage{dsfont}

\usepackage{minitoc}

\usepackage{bbm}
\newtheorem{lemma}{Lemma}
\newtheorem{proposition}{Proposition}
\newtheorem{corollary}{Corollary}
\newtheorem{theorem}{Theorem}
\newtheorem*{theorem*}{Theorem}
\newtheorem{assumption}{Assumption}

\newtheorem{definition}{Definition}

\newtheorem{remark}{Remark}

\usepackage{authblk}

\usepackage{setspace}
\newcommand{\aF}[1][x]{\alpha^*_F(#1)}
\newcommand{\bF}[1][x]{\beta^*_F(#1)}

\newcommand{\aFs}[1][x]{\tilde{\alpha}(#1)}
\newcommand{\bFs}[1][x]{\tilde{\beta}(#1)}

\newcommand{\Rev}{\mathrm{Rev}}

\DeclareMathOperator*{\argmax}{arg max}

\date{Current version: October 02, 2026\\Original version: February 10, 2025}

\title{Randomly Wrong Signals:\\ Bayesian Auction Design with ML Predictions\footnote{A previous version of this paper was entitled ``Auction Design using Value Prediction with Hallucinations.''}}

\begin{document}

\author[1]{Ilan Lobel}
\author[2]{Humberto Moreira}
\author[3]{Omar Mouchtaki}
\affil[1]{NYU Stern School of Business, \texttt{ilobel@stern.nyu.edu}}
\affil[2]{FGV/EPGE Escola Brasileira de Economia e Finança, \texttt{humberto.moreira@fgv.br}}
\affil[3]{NYU Stern School of Business, \texttt{om2166@stern.nyu.edu}}

\maketitle

\vspace{-.1in}

\begin{abstract}

Machine learning can help sellers predict buyers' willingness to pay, but its predictions may be plausible and yet entirely unrelated to a buyer's value. We introduce a Bayesian auction model that distinguishes such hallucinations from ordinary measurement noise.
We characterize revenue-maximizing auctions with publicly disclosed predictions in the benchmark where informative predictions are exact. For a single buyer, this characterization explains when the seller should follow a prediction and when to hedge against its failure by ignoring it or capping the price.
These insights guide the design of simple auctions that account for both hallucinations and measurement noise. Numerical experiments show that the resulting auctions achieve near-optimal revenue.


\end{abstract}

\doparttoc %
\faketableofcontents %

\section{Introduction}\label{sec:intro}

In this work, we study a Bayesian auction design problem where a seller aims to design a revenue-maximizing mechanism to sell an indivisible good to $n$ buyers. In the classical version of the problem \citep{myerson1981optimal}, each buyer's private value is independently drawn from a prior distribution, which is common knowledge to all agents. A classical feature of this problem is information asymmetry: while buyers know their own private values, the seller has no direct access to this information beyond the prior. However, in many practical applications of mechanism design, such as in advertising auctions, the seller often does possess additional information about buyers' private values. In particular, sellers can train machine learning models to predict buyers' valuations. To do so, they can often rely on a wealth of data: past interactions with the same buyer, contextual information, and even bids by similar buyers.

A prediction based on contextual information can be inaccurate for different reasons. It may use the appropriate evidence but measure or infer the buyer's value imprecisely in which case the signal would be a noisy version of the value of the buyer. It may also produce an answer whose connection to that buyer's evidence has failed, while remaining plausible for the relevant population. In this work, we refer to this loss of association between the prediction and the buyer's value as a \emph{hallucination}. Research on contextual generation documents related failures: models may favor memorized answers over supplied evidence, and generating plausible content need not ensure that the content is supported by the current source \citep{longpre2021entity,van2022mutual,shi2024trusting}. Motivated by this distinction, we separate local measurement noise from the reliability of the signal's association with the buyer and propose a model which encompasses both by using a structured specialization of Huber's contamination framework \citep{huber1964robust}.

Specifically, we study a Bayesian auction design problem in which the seller observes a machine learning prediction, or signal, for each buyer that may be randomly wrong. Each buyer's private value is independently drawn from a known prior distribution. Conditional on being informative, the signal is a noisy version of the buyer's private value. With some probability, however, the signal is a hallucination drawn independently from the same prior. The failure is probabilistic and is not chosen by an adversary. As such, while our paper is inspired by the recent literature on learning-augmented algorithms, our approach is quite different from prior work (see \Cref{sec:lit_review}). We use a classical Bayesian framework in which the seller knows the prior distribution and the reliability of the prediction technology.

\subsection{Contributions}

\paragraph{Optimal signal-revealing mechanism}
We first establish that, in our model, allowing arbitrary mechanisms enables the seller to extract the full surplus via highly non-credible designs that condition on unverifiable private signals. Motivated by the resulting implementability concerns, we therefore focus on optimal \emph{signal-revealing} direct mechanisms, in which the seller publicly discloses the signals.
Under this signal-revealing model, the optimal mechanism can be computed by a classical ironing procedure. While this approach recovers the mechanism numerically, it falls short of revealing its structural properties. We therefore study an extreme regime in which the informative signal carries no measurement error, so that the signal is either a hallucination or exactly equal to the value. Our first technical contribution is a characterization of the optimal signal-revealing mechanism in this setting. The seminal characterization of \citet{myerson1981optimal} does not apply in the noise-free setting, because the posterior distribution induced by the signal contains an atom and therefore does not admit a continuous density. We instead build on the more general formalism developed by \citet{monteiro2010optimal} for optimal auctions under arbitrary distributions. While their approach characterizes ironed virtual values via infinitely many semi-infinite linear optimization problems, we show in \Cref{thm:main} that in our setting this characterization admits a closed-form solution. Even when the prior distribution is regular, the posterior necessarily requires ironing due to the point mass at the signal. Our characterization reveals an almost decomposable structure: the ironing for values below and above the realized signal can be performed nearly independently.

\paragraph{A simple, near-closed-form price for a single buyer.}
Specializing \Cref{thm:main} to a single buyer, we obtain a clean characterization of the optimal mechanism as a posted price with four regimes: the seller ignores low signals, follows intermediate signals, caps high signals, and follows again at the very top (\Cref{prop:optimal_price_corrected}). This structure closely mimics the optimal price that one would obtain numerically by applying Myerson's ironing procedure in the noisy model, but in a form that is transparent and interpretable: under measurement noise, the optimal price is a smoothed version of the same ignore--follow--cap structure (\Cref{fig:intro-single-buyer}) and retains similar qualitative properties.

\begin{figure}[htbp]
\centering
\begin{tikzpicture}[scale = 0.7]
\begin{axis}[
    width=10cm,
    height=8cm,
    xlabel={Observed signal},
    ylabel={Optimal price},
    table/col sep=comma,
    grid=both,
    legend style={font=\footnotesize},
    legend pos=south east,
    ymin=0, ymax=1
]
\addplot [red,  unbounded coords=jump, line width = .5mm] table [x=s, y={p_hall_gamma_0.8}]        {Data/one_buyer_prices_intro_uniform_gamma0.8_sigma0.1.csv};
\addlegendentry{Noise-free}

\addplot [gray, unbounded coords=jump, line width = .5mm] table [x=s, y={p_hall_noise_gamma_0.8}]  {Data/one_buyer_prices_intro_uniform_gamma0.8_sigma0.1.csv};
\addlegendentry{Noisy model}
\end{axis}
\end{tikzpicture}
\caption{Optimal single-buyer posted price for uniform prior.}
\label{fig:intro-single-buyer}
\end{figure}

We then use the noise-free mechanism as a design device for the noisy model. Applying it verbatim can incur a large revenue loss: pricing at the signal in the follow regime overprices whenever the realized noise pushes the signal above the value, and revenue collapses (\Cref{prop:brittleness_noise}).
However, a simple correction restores performance: deflating the follow price by a buffer $\tau$ yields a family of mechanisms, and the theoretical choice $\tau=\sigma\sqrt{2\log(b/\sigma)}$ brings revenue within a vanishing loss of the noisy optimum (\Cref{prop:sensitivity-noise}). The same buffer parameter can be tailored numerically to the prior and noise level.
Finally, we show that our structural findings are robust to not knowing the hallucination probability: in \Cref{sec:gamma-robust} we show that the posted price minimizing worst-case regret over all reliability levels $\gamma \in [0,1]$ retains the ignore--follow--cap structure, with a cap that now increases with the signal.

\paragraph{Designing simple auctions with personalized reserves.}
In multi-bidder settings, the optimal signal-revealing auction characterized in \Cref{thm:main} can be complex because the realized signals create asymmetries across bidders. We therefore focus on a standard workhorse format in practice, especially in online advertising: a second-price auction with personalized reserve prices. We implement reserves in the \emph{eager} manner, meaning that bidders whose bids fall below their reserves are first discarded, and the second-price
auction is then run among the remaining bidders \citep{paes2016field}.

A natural approach is to take the bidder-specific reserve prices suggested by the optimal mechanism in \Cref{thm:main} and plug them into this eager format. Somewhat surprisingly, this direct translation can reduce revenue, and can perform worse than the naive eager benchmark that simply sets each bidder's reserve equal to their realized signal.
Guided by our characterization in the one-buyer setting, we design modified personalized reserves that adjust the optimal reserve prices to account for eager filtering and to hedge against hallucinated outliers. We prove in \Cref{thm:signal_worse_than_cap} that the resulting eager mechanism attains higher expected revenue than the signal-as-reserve eager benchmark. Numerically, our proposed eager auction is consistently near-optimal, typically achieving at least 97\% of the optimal signal-revealing revenue while improving the signal-as-reserve eager benchmark by about 15\% to 20\% for heavier-tailed distributions.

\subsection{Literature Review}\label{sec:lit_review}
Mechanism design and auction theory have been very active areas of research since at least the 1960s, including the celebrated Vickrey-Clarke-Groves framework for welfare maximization \citep{vickrey1961counterspeculation,clarke1971multipart,groves1973incentives}. 
\citet{myerson1981optimal} laid the foundation for the literature on revenue maximization, proving many of the results that we build on: revelation principle, the role of the virtual value and the ironing procedure. We also build closely on \citet{monteiro2010optimal}, who developed techniques for ironing virtual values in settings where the priors do not have densities. 
For general distributions, the complexity of the revenue-maximizing auction derived in \citet{myerson1981optimal} has motivated extensive research into simple and more practical mechanisms that are easier to implement while remaining near-optimal \citep{hartline2009simple,roughgarden2019approximately}. 
Our work contributes to this literature by modeling a practical setting in which the seller relies on machine learning algorithms that provide hallucination-prone predictions and by studying the design of optimal mechanisms that are robust to such predictive errors.

Recent research has begun exploring the use of modern generative models such as large language models (LLMs) to help bidders uncover or refine their valuations from rich, unstructured information. \cite{sun2024large} develops the Semantic-enhanced Personalized Valuation in Auction (SPVA) framework, using semantic embeddings of product and preference data to support personalized valuation, while \cite{waseem2023artificial} reports the development of an LLM-based assistant for bid evaluation in procurement. \cite{huang2025accelerated} develops LLM-based \emph{proxy} designs for preference elicitation in combinatorial auctions, using natural-language interaction to reduce the communication burden.
Beyond bidder-side assistance, a small but growing line of work studies how LLMs can be incorporated directly into auction design. The concurrent work \cite{bergemann2025data} examines welfare-maximizing mechanisms that use side information about the state to structure payments (while restricting allocation to depend only on bidder messages), and \cite{sun2025role} studies prediction-based prescreening rules that affect participation and revenue across standard auction formats. 
Our work complements this emerging literature by providing a stylized and tractable model of random prediction error motivated by  modern generative systems.

Our work also relates to the literature on learning-augmented algorithms, also known as algorithms with predictions/advice in which a decision-maker has access to some prediction with unknown accuracy. The standard goal in this literature is to design algorithms that achieve a good trade-off between two performance metrics: consistency, which is the performance if the predictions are perfect, and robustness, which corresponds to the performance when the predictions are adversarial \citep{purohit2018improving,lykouris2021competitive}. This framework has been applied to the study of several problems in various fields. More recently, \citep{agrawal2022learning,balkanski2022strategyproof,gkatzelis2022improved,banerjee2022online,xu2022mechanism} studied learning-augmented algorithms in the context of strategic agent problems, including mechanism design.  In particular, \cite{xu2022mechanism,balkanski2023online,caragiannis2024randomized} and \citet{lu2024competitive} consider the auction design problem in which the private value of the agents is arbitrarily chosen and in which the seller observes a prediction with unknown precision. The goal of these works is to design a mechanism that performs well in both good consistency and robustness. In a similar vein, \citet{balcan2023bicriteria} propose a welfare-efficient mechanism for settings with ML signals that offers a minimum revenue guarantee. Our work is conceptually related to this literature as we also assume that the seller observes a prediction which can be used to infer the values of the buyers. A key modeling distinction is that we consider a fully Bayesian setting in which the values of the buyers are sampled from a known distribution and in which our model for prediction errors assumes that the ML algorithm is ``randomly'' wrong, as opposed to adversarially wrong.

Our work also relates to the design of data-driven mechanisms, which uses a finite set of samples, independently drawn from the buyer value distribution, to design mechanisms \citep{cole2014sample,gonczarowski2017efficient,guo2019settling}.
More recent work has addressed mechanisms that account for potentially corrupted samples \citep{cai2017learning,brustle2020multi,guo2021robust,besbes2022beyond}. Closely related is \citet{devanur2016sample}, who study auctions with side information and focus on the sample complexity of learning near-optimal mechanisms. In contrast, we abstract away from learning and instead focus on the behavior of machine learning predictions rather than the data-generation process. In particular, we analyze the impact of signals that may fail in an uninformative way, which fundamentally alters the structure of the optimal mechanism.

Our analysis focuses on {\it direct signal-revealing mechanisms}, in which the seller publicly discloses the signals she observes about buyers’ valuations. We adopt public disclosure as a benchmark to isolate the effects of hallucination-prone signals while ruling out revenue gains that rely on non-credible or opaque uses of information. %
This modeling choice contrasts with the informed principal framework of \citet{maskin&tirole1990informedprincipal}, who study environments in which the principal possesses private information and strategically designs contracts to influence agents’ beliefs. Their non-cooperative analysis emphasizes the sorting effects of contract proposals and shows that the principal generally does not lose by withholding information. %

Our work also relates to the broader literature on mechanism design with an informed principal, which studies settings in which the mechanism designer has private information. \citet{Cella2008} shows that when the principal’s information is kept private in a two-sided private information environment, correlation can be exploited to extract a larger share of the surplus by preserving agents’ uncertainty about the principal’s type. %
\citet{Severinov2008} shows that, under Crémer–McLean–type conditions and sufficient informational richness, an informed principal can implement efficient allocations and capture the entire social surplus, a result that is related to our Proposition \ref{prop:full_surplus}. Overall, whereas the informed principal literature highlights the strategic value of private information, we study environments in which credibility considerations lead the seller to disclose the signals. This disclosure is not costless, but it does eliminate impractical and non-credible mechanisms from the feasible set.

\section{The Model}\label{sec:model}

A seller has one indivisible good to sell to $n$ potential buyers. Buyer $i$ has a private value
$v_i \in [a_i,b_i]$, drawn independently from a commonly known distribution $F_i$ with strictly positive density $f_i$.
We impose the standard regularity assumption used in the classical Bayesian
auction model.

\begin{assumption}[Regularity]\label{ass:regular}
For every $i$, the virtual value function $ v - (1-F_i(v))/f_i(v)$ is non-decreasing on
$[a_i,b_i]$.
\end{assumption}

Regularity is a standard assumption in the mechanism design literature \citep{myerson1981optimal}: it is equivalent to the seller's revenue function $q \mapsto q \cdot F_i^{-1}(1-q)$ being concave in the quantile $q$. Many widely used distributions satisfy regularity, including the uniform, exponential, normal, Pareto with shape parameter at least one, and gamma and beta with appropriate parameter ranges. In particular, log-concave distributions are regular.

The seller observes a signal $s_i$ for each buyer $i$, produced by a machine learning system that conditions on available data and context. Following the distinction introduced in \Cref{sec:intro}, we model a \emph{randomly wrong} signal with two sources of error: local measurement noise when the prediction uses the appropriate evidence, and a loss of association between the prediction and the buyer's value, which we call a \emph{hallucination}. We represent the latter through a branch in which the signal is independent of the buyer's true value, although its output remains plausible for the relevant population. The failure is probabilistic, and the seller faces uncertainty about the signal's reliability.

Formally, for each buyer $i$, let $h_i \in \{0,1\}$ denote whether the signal is a hallucination, with $\mathbb{P}(h_i=1)=\gamma_i \in (0,1)$. Let $w_i$ be an independent draw from $F_i$ (we consider an extension where $w_i$ follows a different distribution in \Cref{sec:failure_myerson}) and let $\varepsilon_i$ be a Gaussian shock with mean zero and variance $\sigma_i^2$. The model is:
\begin{equation*}
v_i \sim F_i,\qquad h_i \sim \mathrm{Bernoulli}(\gamma_i),\qquad w_i \sim F_i, \qquad \varepsilon_i \sim \mathcal{N}(0,\sigma_i^2),
\end{equation*}
with $(v_i,h_i,w_i,\varepsilon_i)$ mutually independent, and
\begin{equation*}
s_i \;=\;
\begin{cases}
v_i + \varepsilon_i, & \text{if } h_i=0 \quad \text{(informative prediction)},\\
w_i, & \text{if } h_i=1 \quad \text{(hallucination)}.
\end{cases}
\end{equation*}

Throughout, we assume that the value distributions $(F_i)_{i=1}^n$, the hallucination probabilities $(\gamma_i)_{i=1}^n$, and the standard deviations $(\sigma_i)_{i=1}^n$ are common knowledge to both the seller and the buyers. The realized signals $(s_i)_{i=1}^n$ are observed by the seller. The seller does not observe $(v_i,w_i,h_i,\varepsilon_i)$. The realized value $v_i$ remains private to buyer~$i$.\label{lbl:common-knowledge}

We use $\bm{\gamma}$, $\bm{s}$, and $\bm{\sigma}$ to represent the vectors of hallucination probabilities, signals, and standard deviations, respectively. Given a signal, the seller can perform a Bayesian update to obtain what we call the posterior distribution of a buyer's value. We denote by $\bm{F}_{\bm{\gamma},\bm{s},\bm{\sigma}}$ the posterior distribution of the buyers' values and by $\bm{F}_{\bm{\gamma},\bm{s},\bm{\sigma},-i}$ the posterior distribution of the buyers excluding the $i^{th}$ buyer.

\begin{remark}[Relation to Huber contamination]\label{rem:huber}
Our signal specification follows the mixture structure of Huber's contamination model \citep{huber1964robust}. Allowing the hallucination to follow a distribution $G_i$, the conditional signal law is
\begin{equation*}
\mathcal{L}(s_i\mid v_i=v)=(1-\gamma_i) \cdot \mathcal{N}(v,\sigma_i^2)+\gamma_i \cdot G_i,
\end{equation*}
with $G_i=F_i$ in our benchmark. In the classical robust-estimation formulation, the contaminating distribution $G_i$ is unspecified, and the objective is to construct estimators that perform well across possible contaminations. Here, we specify the contamination distribution and study the Bayesian mechanism-design problem induced by the resulting signals. The restriction that $G_i$ does not depend on $v$ represents a loss of association between the prediction and the buyer's value.
\end{remark}

\paragraph{Discussion and modeling choice.}
The parameter $\sigma_i$ captures the precision of the informative prediction, while $\gamma_i$ captures the probability that the signal's association with the buyer fails. We keep the hallucination branch independent of the buyer's value to isolate this reliability risk. The same-prior assumption has a substantive interpretation: in the zero-noise benchmark, both informative and hallucinated signals have marginal distribution $F_i$, so a loss of association preserves the distribution of plausible values.

The case $\bm{\sigma}=\bm{0}$ is the main analytical benchmark in the paper. In that limit, a non-hallucinated signal equals the buyer's value exactly, and the posterior places an atom at the realized signal. This is an extreme case of the general noisy model, corresponding to nearly perfect conditioning in the informative branch.  When $\bm{\sigma}$ is strictly positive, the posterior is smooth and the standard Myerson program can be applied after the posterior has been computed for the realized signal. That approach is useful for numerical comparisons, while the zero-noise benchmark provides the sharp characterization that we use to understand the form of the optimal mechanism.

\paragraph{Objective.}
The question we aim to address in this paper is what is the seller's revenue-maximizing mechanism in the presence of this value prediction technology. A mechanism is defined by a pair $(\bm{x},\bm{p})$, where $\bm{x}$ (resp. $\bm{p}$) is an allocation (resp. payment) function which takes as input the vector of reported types $\bm{v}$ and the vector of observed signals $\bm{s}$ and outputs the vector of probability of allocation (resp. of payment) for each buyer. We assume that all agents have quasi-linear utilities. For a given vector of signals $\bm{s}$, we will explore the following problem:
\begin{subequations}
\label{eq:optimal_mechanism}
\begin{alignat}{2}
&\!\sup_{(\bm{x},\bm{p})} &\;& \mathbb{E}_{\bm{v} \sim \bm{F}_{\bm{\gamma},\bm{s},\bm{\sigma}}} \left[  \sum_{i=1}^n  p_i(\bm{v},\bm{s}) \right]    \\
&\text{s.t.} &      &  \mathbb{E}_{\bm{v_{-i}} \sim \bm{F}_{\bm{\gamma},\bm{s},\bm{\sigma},-i}} \left[ v_i \cdot x_i(v_i,\bm{v}_{-i},\bm{s}) - p_i(v_i,\bm{v}_{-i},\bm{s}) \right] \nonumber \\
 &  &  & \qquad \geq \mathbb{E}_{\bm{v_{-i}} \sim \bm{F}_{\bm{\gamma},\bm{s},\bm{\sigma},-i}} \left[ v_i \cdot x_i(v'_i,\bm{v}_{-i},\bm{s}) - p_i(v'_i,\bm{v}_{-i},\bm{s}) \right] \quad \text{for every $i, v_i, v'_i$,} \label{eq:IC} \\
 &  &  &\mathbb{E}_{\bm{v_{-i}} \sim \bm{F}_{\bm{\gamma},\bm{s},\bm{\sigma},-i}} \left[ v_i \cdot x_i(v_i,\bm{v}_{-i},\bm{s}) - p_i(v_i,\bm{v}_{-i},\bm{s}) \right] \geq 0 \quad \text{for every $i, v_i$,} \label{eq:IR}\\
& & & \sum_{i=1}^n x_i(\bm{v},\bm{s}) \leq 1 \quad \text{for every $\bm{v}$.}
\end{alignat}
\end{subequations}

Problem \eqref{eq:optimal_mechanism} is written for the general noisy model. In the next three sections, when $\bm{\sigma}=\bm{0}$ we write $\bm{F}_{\bm{\gamma},\bm{s}}$ for the corresponding posterior.

 \noindent \textbf{Signal-revealing direct mechanisms.} Problem \eqref{eq:optimal_mechanism} specifies the problem of finding the optimal signal-revealing direct mechanism. A direct mechanism is one where the seller chooses an incentive-compatible allocation and payment scheme, and asks the buyers to reveal their types. In standard mechanism design, restricting to direct mechanisms is without loss of optimality \citep{myerson1981optimal}. We define a signal-revealing mechanism to be one where the seller shares the signals alongside the allocation and payment rules.
 We note that, in general, the seller may achieve strictly higher revenue by employing mechanisms that are non-signal-revealing—that is, mechanisms in which the seller does not disclose their private signal to the buyer. In \Cref{sec:apx_full}, we demonstrate that allowing for such mechanisms enables the seller to extract the full surplus from the buyer by effectively punishing deviations when the reported type appears inconsistent with the seller’s private signal. This approach draws on the classic insight of ~\cite{cremer1988full}, who showed that correlated information between different buyers can be exploited to eliminate information rents.
Signal disclosure alone does not make signals verifiable. As we will see in \Cref{rem:own_signal_disclosure}, our optimal mechanism can nevertheless be implemented by announcing only the allocation and payment rules, and the seller cannot gain by misreporting the signals.
We also note that by assuming the mechanism is signal-revealing we made the formulation relatively straightforward: both the objective and the IC and IR constraints use the posterior distributions given signals rather than the priors.

\noindent \textbf{Timing of the game.}\label{lbl:timeline}
For clarity, we summarize below the timeline of the game considered throughout the paper.
\begin{enumerate}
    \item[0.] The seller commits to broadcasting the signal vector $\bm{s}$ to all buyers. 
    
    \item Each buyer $i$'s value $v_i$ is drawn independently from $F_i$.
    The seller does not observe the realized values.
    
    \item For each buyer $i$, the signal $s_i$ is generated according to the
    signal technology described above. The seller observes the signal vector
    $s=(s_1,\ldots,s_n)$.
    
    \item In the signal-revealing formulation, the seller discloses $s$,
    announces a direct mechanism $(x_s,p_s)$, and commits to implementing its
    allocation and payment rules.

    \item Each buyer $i$ reports a value $v_i$.

    \item The good is allocated according to $x_s(\bm{v})$, and payments are
    collected according to $p_s(\bm{v})$.
\end{enumerate}

\section{Bayesian Update and Applying Myerson}\label{sec:failure_myerson}

We now specialize the model to the zero-noise benchmark, $\sigma_i=0$ for every buyer $i$. In this benchmark, a non-hallucinated signal satisfies $s_i=v_i$, so conditioning on the signal creates a posterior with a point mass at $s_i$. We will use this singular case to develop a closed-form characterization in the next section.

In our setting, the seller obtains the signals $s_i$ prior to selecting the mechanism. After obtaining $s_i$, the seller's posterior belief about $v_i$ is given by:
\begin{equation}\label{eq:density-f} f_{\gamma_i,s_i}^i(v) = \gamma_i \cdot f_i(v) + (1-\gamma_i)\cdot \delta_{s_i}(v),\end{equation}
where $\delta_{s_i}(\cdot)$ is the Dirac delta function that places a unit of mass at $s_i$ and zero mass everywhere else. Equivalently, the cumulative distribution satisfies:
\begin{equation}\label{eq:cumulative-F}
F_{\gamma_i,s_i}^i(v) = \begin{cases}
\gamma_i \cdot F_i(v) \quad \text{for $v<s_i$},\\
\gamma_i \cdot F_i(v) + (1-\gamma_i) \quad \text{for $v \geq s_i$}.
\end{cases}
\end{equation}

The question we aim to address can thus be rephrased as what is the revenue-maximizing auction when the valuation of buyer $i$ is drawn according to $F_{\gamma_i,s_i}^i$. 

\vspace{.1in}
\noindent \textbf{On the distribution of hallucinations.} We will assume throughout the paper that the value $v_i$ and any potential hallucination $w_i$ are drawn from the same distribution. However, if we were to assume that the value were drawn from density $f_i$ and the hallucination from density $g_i$, where these distributions are absolutely continuous with respect to each other, we could obtain a similar formula via Bayesian updating. Let $h_i$ represent whether a hallucination occurred. The posterior density would then be given by:
\begin{eqnarray*}
    f_{\gamma_i,s_i}^i(v) &=&  P(h_i=1 \mid s_i) \cdot f_{\gamma_i,s_i}^i(v \mid h_i=1) +  P(h_i=0 \mid s_i) \cdot f_{\gamma_i,s_i}^i(v \mid h_i=0)\\ 
    &=& P(h_i=1 \mid s_i)\cdot f_i(v) + P(h_i=0 \mid s_i)\cdot \delta_{s_i}(v)\\ 
    &=& \frac{\gamma_i \cdot g_i(s_i)}{\gamma_i \cdot g_i(s_i) + (1-\gamma_i) \cdot f_i(s_i)} f_i(v) + \frac{(1-\gamma_i) \cdot f_{i}(s_i)}{\gamma_i \cdot g_i(s_i) + (1-\gamma_i) \cdot f_i(s_i)}\delta_{s_i}(v)\\
    &=& \widetilde {\gamma}^i_{s_i} \cdot f_{i}(v) + (1-\widetilde {\gamma}^i_{s_i}) \cdot \delta_{s_i}(v), 
\end{eqnarray*}
where
$\widetilde \gamma_{s_i} = \left(1+\frac{1-\gamma_i}{\gamma_i}\frac{f_i(s_i)}{g_i(s_i)}\right)^{-1}$. That is, our results from the rest of the paper would apply if we replace $\gamma_i$ with $\widetilde \gamma_{s_i}$.

\subsection{Applying Myerson}

\citet{myerson1981optimal} tells us that in a private values setting, the revenue-maximizing auction is given by calculating the virtual value of each agent (which might require ironing) and then allocating the item to the agent with the highest non-negative virtual value, or discarding the item if all of the virtual values are negative. Since virtual values are computed separately for each buyer, we will suppress the buyer index $i$ from the notation whenever possible to lighten the notational burden. 

For a given density $f$ and cumulative distribution $F$, the pre-ironing virtual value function is $\varphi_F(v) = v - (1-F(v))/f(v)$. For the density and cumulative distributions given by Eqs. \eqref{eq:density-f}
 and \eqref{eq:cumulative-F}, we have:
 \[ \varphi_{F_{\gamma,s}}(v) =  
 \begin{cases}
     v - \frac{1/\gamma - F(v)}{f(v)}, &\hbox{ for } v < s,\\
     v - \frac{1 - F(v)}{f(v)}, &\hbox{ for } v > s.\\
 \end{cases}\]
 We note that the virtual value function is not well-defined at $s$, but we will ignore this issue for now since that is a single point. The function $\varphi_{F_{\gamma,s}}$ does not need to be ironed after $s$ since $\varphi_{F_{\gamma,s}}(v) = \varphi_{F} (v)$ for $v > s$ and we have assumed $F$ is regular. Ironing could be necessary before $s$ depending on the choice of $F$. 

 Let's apply this to single-buyer, uniform over $[0,1]$ case. For this particular $F$, we obtain:
  \begin{equation}\label{eq:misleading-F} \varphi_{F_{\gamma,s}}(v) =  
 \begin{cases}
     2v - 1/\gamma, &\hbox{ for } v < s,\\
     2v - 1, &\hbox{ for } v > s.\\
 \end{cases}\end{equation}
 For this particular distribution, ironing is not necessary before $s$ since $2v - 1/\gamma$ is an increasing function of $v$. 
 Consider the special case $s=1/2-\epsilon$ and $\gamma = \epsilon$, for a small $\epsilon$. Eq. \eqref{eq:misleading-F} crosses zero at $v=1/2$, implying that the optimal price is $1/2$. However, this cannot be the correct optimal price. The revenue generated by this price is bounded above by $\epsilon$ since it requires $s$ to be a hallucination as a necessary condition for a sale to occur. Meanwhile, using the signal $1/2-\epsilon$ as the price would generate at least $(1-\epsilon)\cdot(1/2-\epsilon)$ in revenue. 

 It turns out that ignoring what occurred at $s$, where the density $f_{\gamma,s}$ is not well-defined, and applying Myerson's technique naively was a mistake. To obtain a correct optimal auction, we will need to use a more sophisticated characterization of optimal auctions that applies for distributions that do not admit densities.

\section{Characterization of the Optimal Auction}\label{sec:optimal_auction}

In this section, we first introduce a slight generalization of Myerson's ironing operation, which we will need to state our results. We then present our main technical theorem, and demonstrate what it implies for some simple distributions. %

\subsection{Truncated Myerson Ironing}\label{sec:ironing}

Consider a distribution $F$ supported on $[a,b]$ and which admits a positive density on its support.  In that case $F$ is strictly increasing on $[a,b]$ and therefore it admits an inverse function $F^{-1}$ strictly increasing on $[0,1]$. When the virtual value function of $F$ defined for every $x \in [a,b]$ as $\varphi_{F}(x)$  
is not monotonic non-decreasing, \citet{myerson1981optimal} proposes a general procedure called ironing to characterize the optimal auction. In what follows we introduce our slight generalization of Myerson's ironing operator. The only difference between the operator we introduce below and the one presented in \citet{myerson1981optimal} is that we also allow for the operation to be performed only in an interval of the quantile space rather than over the entire quantile space. Hence, we call this operation the truncated Myerson ironing. In \Cref{sec:main} we will need this notion of truncated ironing to prove that the ironing of the virtual value function in our setting decomposes into truncated ironing on subintervals of appropriately defined distributions.

For every quantile $q \in [0,1]$, let
\begin{equation}
\label{eq:J}
    J(q) = \int_0^q \varphi_{F}(F^{-1}(r)) dr.
\end{equation}
The function $J$ relates directly to revenue: posting a price $v$ yields revenue $J(1) - J(F(v))$, so maximizing revenue is equivalent to minimizing $J$. When $J$ is convex, the minimization step is simple as any local minimum is also a global minimum. Moreover, the virtual value $\varphi_{F}$ is related to the derivative of the revenue through the relation $J'(q)=\varphi_F(F^{-1}(q))$. Consequently, because $F^{-1}$ is increasing, $\varphi_F$ is monotone non-decreasing if and only if $J$ is convex. Thus, ironing (i.e., the procedure of constructing a monotonic surrogate for the virtual value function) can be expressed as a convexification of $J$, whose slopes determine the ironed virtual values.

Namely, for every $x \in [0,1]$, let $G_x:[0,x] \to \mathbb{R}$ be the convex hull of the restriction of the function $J$ on $[0,x]$, formally defined for every $q \in [0,x]$ as,
\begin{equation*}
    G_x(q) = \min_{ \substack{(\lambda,r_1,r_2) \in [0,1]\times[0,x]^2\\ \text{s.t. } \lambda \cdot r_1 + (1-\lambda) \cdot r_2 = q} } \lambda \cdot J(r_1) + (1-\lambda) \cdot J(r_2) 
\end{equation*}
Intuitively, we construct the function $G_x(q)$ by taking the ``lowest'' chords formed by two points on the curve of $J$.

By definition, $G_x$ is convex on $[0,x]$. Therefore, it is continuously differentiable on $[0,x]$ except at countably many points. For every $q \in [0,x]$, we define the function $g$ as,
\begin{equation*}
    g_x(q) = \begin{cases}
        G'_x(q) \quad \text{if $G$ is differentiable at $q$}\\
        \lim_{\tilde{q} \downarrow q} G'_x(\tilde{q}) \quad \text{otherwise.}
    \end{cases}
\end{equation*}
The convexity of $G_x$ implies that $g_x$ is monotone non-decreasing. For any $t \in [a,b]$ we define the truncated ironed virtual of $F$ on $[a,t]$ as the mapping,
\begin{equation*}
    \mathrm{IRON}_{[a,t]}[F] : \begin{cases}
        [a,t] \to \mathbb{R}\\
        v \mapsto g_{F(t)}(F(v)).
    \end{cases}
\end{equation*}

When $t = b$, we note that $\mathrm{IRON}_{[a,b]}[F]$ corresponds to the classical notion of ironing introduced in \citet{myerson1981optimal}. We emphasize that, when $t<b$, the mapping $\mathrm{IRON}_{[a,t]}[F]$ is in general different from the restriction of $\mathrm{IRON}_{[a,b]}[F]$ to $[a,t]$. This distinction arises because the convex envelope computed on $[0,1]$ may use supporting chords with endpoints above $F(t)$. Excluding those quantiles can therefore change the convex envelope, and hence its slope, even on $[0,F(t)]$.

\Cref{fig:ironing_operator} illustrates this distinction. Panel~(a) displays $J$ and the lower convex envelopes $G_{F(t)}$ computed on $[0,F(t)]$ for $t\in\{0.2,0.5,2\}$. Because the set of admissible supporting chords depends on $t$, the envelopes can differ on their common domains. Panel~(b) displays the corresponding slope functions mapped back into value space. Linear segments of the convex envelopes in Panel~(a) induce flat segments of the truncated ironed virtual values in Panel~(b). Thus, applying Myerson ironing on the full support and subsequently restricting the resulting function is not, in general, equivalent to truncated ironing.

 \begin{figure}[h!]
    \centering
    \subfigure[Convexification in quantile space]{
    \begin{tikzpicture}[scale=.60]
    \begin{axis}[
        width=12cm,
        height=10cm,
        xmin=-0.,xmax=1.0,
        ymin=-0.12,ymax=0.01,
        scaled y ticks={base 10:2},
        table/col sep=comma,
        xlabel={$q$},
        ylabel={$J(q)$},
        grid=both,
        legend pos=south west
    ]

    \addplot [blue, dashed, line width=.7mm] table[x=F,y=H] {Data/ironing_example_mix_truncated_normals.csv};
    \addlegendentry{Before ironing}
    
    \addplot [red, very thick] table[x=F,y=psi] {Data/ironing_example_mix_truncated_normals.csv};
    \addlegendentry{$\mathrm{IRON}_{[0,2]}$}

    \addplot [teal, very thick] table[x=F,y=psi_cut05] {Data/ironing_example_mix_truncated_normals.csv};
    \addlegendentry{$\mathrm{IRON}_{[0,0.5]}$}
    
    \addplot [black, very thick] table[x=F,y=psi_cut02] {Data/ironing_example_mix_truncated_normals.csv};
    \addlegendentry{$\mathrm{IRON}_{[0,0.2]}$}
    
    \end{axis}
    \end{tikzpicture}
    }
    \subfigure[Virtual value]{
    \begin{tikzpicture}[scale=.60]
    \begin{axis}[
        width=12cm,
        height=10cm,
        xmin=0,xmax=2.0,
        ymin=-2.5,ymax=2,
        table/col sep=comma,
        xlabel={$v$},
        ylabel={virtual value},
        grid=both,
        legend pos=south east
    ]
    
    \addplot [blue,  dashed, line width=.7mm] table[x=x,y={virtual_value_preiron}] {Data/ironing_example_mix_truncated_normals.csv};
    \addlegendentry{Before ironing}

    \addplot [red, very thick] table[x=x,y=virtual_value] {Data/ironing_example_mix_truncated_normals.csv};
    \addlegendentry{$\mathrm{IRON}_{[0,2]}$}

    \addplot [teal, very thick] table[x=x,y=virtual_value_cut05] {Data/ironing_example_mix_truncated_normals.csv};
    \addlegendentry{$\mathrm{IRON}_{[0,0.5]}$}

    \addplot [black, very thick] table[x=x,y=virtual_value_cut02] {Data/ironing_example_mix_truncated_normals.csv};
    \addlegendentry{$\mathrm{IRON}_{[0,0.2]}$}

    \end{axis}
    \end{tikzpicture}
    }
    \caption{ 
    The figure illustrates the truncated ironing procedure. The distribution $F$ used is a mixture of two truncated normals on $[0,2]$ with parameters $(0.1,0.04)$ and $(1.9,1.8)$ and respective weights $0.8$ and $0.2$.  (a) The figure shows the initial $J$ function (in blue) and its convex envelopes on $[0,F(t)]$ for $t \in \{0.2, 0.5, 2\}$.  (b) The figure shows the induced virtual value functions before ironing and after truncated ironing on $[0,t]$ for the same values of $t$.} 
    \label{fig:ironing_operator}
    \end{figure}

\subsection{Virtual Value Characterization}\label{sec:main}

If the distribution $F$ does not admit a density that is positive everywhere in the support, the classical Myerson ironing procedure is not applicable since it relies on the existence of the inverse $F^{-1}$. In this case, there exists a more general virtual value characterization developed by \citet{monteiro2010optimal} that is still applicable. That characterization is difficult to work with because it involves generalized convex hulls, rather than the standard convexification used by Myerson. We defer the presentation and discussion of how to use this complex machinery until Section \ref{sec:technical_work}. We are now ready to state the main  technical result of the paper, which states that if the value distributions are regular, then an ironing procedure that has the same complexity as Myerson does apply.

\begin{theorem}\label{thm:main}
Assume $\sigma =0.$ Let $F_i$ be distributions satisfying Assumption \ref{ass:regular} with a strictly positive continuous density.
Then, there exists a direct mechanism that is revenue-maximizing. In this mechanism, given reported values $\hat{v}_i$, the seller allocates the good to the buyer with the highest non-negative value of $\bar{\varphi}^i_{\gamma_i, s_i}(\hat{v}_i)$, where the function $\bar{\varphi}^i_{\gamma_i, s_i}(\hat{v}_i)$ is defined as:
\begin{equation}
\label{eq:ironed-vv} 
\bar{\varphi}^i_{\gamma_i, s_i}(v) = 
\begin{cases}
    \mathrm{IRON}_{[a_i, s_i]}[\gamma_i F_i](v), & \text{if } a_i \leq v < s_i, \\
    \varphi_{F_i}(T_i), & \text{if } s_i \leq v < T_i, \\
    \varphi_{F_i}(v), & \text{if } T_i \leq v \leq b_i,
\end{cases}
\end{equation}
for every $v \in [a_i, b_i]$. Furthermore, the winning bidder pays  the minimum amount they would need to bid to still win.
The constants $(T_i)_{i \in \{1, \ldots, n\}}$ are defined in \Cref{prop:from_F_to_feasible_Fs}, and the operator $\mathrm{IRON}$ is as specified in Section \ref{sec:ironing}.
\end{theorem}
We present the key technical arguments required to prove \Cref{thm:main} in \Cref{sec:technical_work}.

\begin{remark}[Implementation and truthful signal disclosure]\label{rem:own_signal_disclosure}
With a fixed priority for breaking ties, the mechanism in \Cref{thm:main} is dominant-strategy incentive compatible and ex-post individually rational: each buyer's allocation is nondecreasing in its report, winners pay the minimum bid needed to win, and losers pay zero. The seller can therefore announce and commit to the allocation and payment rules before bidding, without disclosing the signals, preserving truthful bidding and revenue.

Moreover, the mechanism associated with any other signal vector remains feasible for problem~\eqref{eq:optimal_mechanism} conditional on the true signals. Since the chosen mechanism maximizes expected revenue conditional on those signals, misreporting them cannot increase revenue in expectation. Truthful disclosure is thus weakly optimal within this family after observing the signals and before receiving bids.
\end{remark}

\Cref{thm:main} above states that $\bar\varphi^i_{\gamma_i,s_i}$ is the correct notion of ironed virtual value function given posterior beliefs $F_{\gamma_i,s_i}^i$. Before the signal $s_i$, the correct ironed virtual value is given by $\mathrm{IRON}_{[a_i, s_i]}[\gamma_i F_i]$, which might require ironing, but where ironing can be done using Myerson's classical approach with the domain truncated to $[a_i,s_i]$.
Immediately after the signal, we need to iron out a segment $[s_i,T_i]$ of the virtual value to account for the mass at $s_i$. After $T_i$, the original virtual value function $\varphi_{F_i}$ applies.
The theorem can be interpreted as a near-decomposition result. Ironing the section strictly before the signal yields $\mathrm{IRON}_{[a_i, s_i]}[\gamma_i F_i]$ while ironing the virtual value from $s_i$ (inclusive) onward yields the second and third clauses of Eq. \eqref{eq:ironed-vv}.
We call this a near-decomposition, not a full decomposition, because $T_i$ creates a link between the two sides, as the value of $T_i$ depends on the distribution before the signal.

\begin{remark}[Importance of the regularity assumption]
    We note that the key assumption that enables this near-decomposition is the regularity of $F_i$. We provide in \Cref{sec:apx_non_decomposition_example} an example to show that if $F_i$ is irregular, then \Cref{thm:main} may fail. 
\end{remark}

It is useful to see what Theorem \ref{thm:main} implies for some simple distributions. If $F$ is a uniform [0,1] distribution, then the virtual value is given by:
  \begin{equation*} \bar \varphi_{F_{\gamma,s}}(v) =   \begin{cases}
     2v - 1/\gamma, &\hbox{ for } v < s,\\
     2T - 1, &\hbox{ for } s \leq v < T,\\
     2v - 1, &\hbox{ for } v \geq  T.\\
 \end{cases}\end{equation*}
 If $F$ is an exponential distribution, then ironing might be required to the left of the signal. Note that the exponential distribution is not only a regular distribution, but satisfies the even stronger condition of monotone hazard rate. Despite this, the pre-signal distribution still sometimes requires ironing (see \Cref{fig:illustration_theorem1}).

    \begin{figure}[h]
    \centering
    \subfigure[Exponential prior $(\lambda = 1), \gamma = 0.95$]{
    \begin{tikzpicture}[scale = 0.60]
    \begin{axis}[
        width=10cm,
        height=10cm,
        xmin=0,xmax=6.5,
        ymin=-4,ymax=6,
        table/col sep=comma,
        xlabel={$v$},
        ylabel={virtual value},
        grid=both,
        legend pos=north west
    ]

    \addplot [blue, dashed,  thick,unbounded coords=jump] table[x=x,y={preiron_s=5}] {Data/virtual_value_gamma=095_exponential.csv};
    \addlegendentry{Unironed (s=5)}
    
    \addplot [red,  thick,unbounded coords=jump] table[x=x,y={s=5}] {Data/virtual_value_gamma=095_exponential.csv};
    \addlegendentry{Ironed (s=5)}
    \addplot[only marks, red, mark=*,forget plot] coordinates {(5, 4.952221754226520)};
    \end{axis}
    \end{tikzpicture}
    }
    \subfigure[Uniform prior, $\gamma = 0.75$]{
    \begin{tikzpicture}[scale = 0.60]
    \begin{axis}[
        width=10cm,
        height=10cm,
        xmin=0,xmax=1,
        ymin=-1.5,ymax=1.5,
        table/col sep=comma,
        xlabel={$v$},
        ylabel={virtual value},
        grid=both,
        legend pos=north west
    ]
    \addplot[domain=0:0.4,samples=50,thick,dashed,blue] {2*x - 1/0.75};  %
    \addplot[domain=0.4:1,samples=50,thick,dashed,blue,forget plot] {2*x - 1};        %
    \addlegendentry{Unironed (s=0.4)}

    \addplot [red,  thick,unbounded coords=jump] table[x=x,y={s=0.4}] {Data/virtual_value_gamma=075_uniform.csv};
    \addlegendentry{Ironed (s=0.4)}
    \addplot[only marks, red, mark=*,forget plot] coordinates {(0.41, 0.25170764)};

    \end{axis}
    \end{tikzpicture}
    }
    \caption{\textbf{Ironed virtual value for different priors.} In each plot the unironed virtual value corresponds to the naive evaluation $\varphi_{F_{\gamma,s}}$, wherever it is well defined (i.e., everywhere but at $s$). The ironed virtual value corresponds to the virtual value characterized in \Cref{thm:main}.}
     \label{fig:illustration_theorem1}
    \end{figure}
\section{The Single Buyer Case}\label{sec:single-buyer}

In this section, we first leverage \Cref{thm:main} to study the structure of the optimal mechanism for a single buyer in the zero-noise setting. We then leverage this structure to design near optimal mechanisms for the noisy setting of \Cref{sec:model}. Finally in \Cref{sec:gamma-robust} we investigate the setting in which the seller does not know the hallucination probability $\gamma$.

\subsection{Optimal Mechanism for One Buyer}
\label{sec:opt_1b}

An important implication of \Cref{thm:main} is the following characterization of the optimal mechanism for a single buyer. In this setting, the optimal mechanism is a posted price. 
For a price $p$ and a value distribution $G$, we write $\Rev(p,G) := p\cdot(1-G(p^-))$ for the expected revenue from posting $p$ to a buyer whose value follows $G$.

\begin{proposition}
\label{prop:optimal_price_corrected}
Assume $\sigma = 0.$ Let $n=1$ and assume $F$ has a log-concave density $f$ on $[0,b]$. Furthermore, assume that $f$ is twice continuously differentiable.
Define\footnote{We use the convention that the infimum (resp. supremum) of an empty subset of $[a,b]$ is $b$ (resp. $a$).}, 
\begin{equation*}
p^{\mathrm{ignore}} = \inf\{v \in [0,b] : \varphi_F(v) \geq 0\}, \qquad
\mbox{and}  \qquad
p^{\mathrm{cap}} = \inf\{v\in[0,b] : \varphi_{\gamma F}(v) \geq 0\}.
\end{equation*}
Then, the following posted price is optimal:
\begin{equation*}
    p^*(s)=
\begin{cases}
p^{\mathrm{ignore}}, 
    & s < L_\gamma, \\
s, 
    & L_\gamma \le s < M_\gamma, \\
p^{\mathrm{cap}},
    & M_\gamma \le s \leq U_\gamma, \\
s,
    & s > U_\gamma,
\end{cases}
\end{equation*}
where
\begin{align*}
    L_\gamma = \inf\{s \in [0,b] : T_s \ge p^{\mathrm{ignore}}\},& \qquad
    M_\gamma = \inf\{s \in [0,b] : \varphi_{\gamma F}(s) \geq 0\}, \\
    U_\gamma = \sup \{s \in [0,b] : s \cdot (1-&\gamma F(s)) \leq M_\gamma \cdot (1-\gamma F(M_\gamma)) \},
\end{align*}
and $T_s$ is the threshold in \Cref{thm:main}.
\end{proposition}

\Cref{prop:optimal_price_corrected} shows that, when signals may be hallucinations, the optimal posted price can exhibit \emph{four} signal regimes (while using only \emph{three} distinct price types), assuming densities are log-concave. For ease of exposition, we present the result when the lower bound on the support $a$ equals $0$. The four-piece characterization holds more generally and is presented in \Cref{sec:apx_single_buyer}.

For low signals ($s < L_\gamma$), the seller optimally \emph{ignores} the signal and posts the prior monopoly price $p^{\mathrm{ignore}}$. Intuitively, when $s$ is small, even in the best case in which the signal is correct, the maximal extractable revenue from following the signal is of order $s$. The seller therefore prefers to bet on the signal being a hallucination and instead rely on the prior-optimal price.
For intermediate signals ($L_\gamma \le s < M_\gamma$), the seller \emph{follows} the signal and posts $p^*(s)=s$. In this region, the signal is sufficiently informative that extracting approximately $s$ in the non-hallucinatory event is attractive, while pricing at $s$ is not yet excessively risky.
For moderately high signals ($M_\gamma \le s < U_\gamma$), the seller \emph{caps} the price and posts the constant $p^{\mathrm{cap}}$. Interestingly, one can show that $p^{\mathrm{ignore}} < p^{\mathrm{cap}} \le s$. Hence, in this regime, the seller prices more aggressively than in the absence of any signal, as it benefits from the fact that if the signal is correct then $s \ge p^{\mathrm{cap}}$ and the good is sold at price $p^{\mathrm{cap}}$, while avoiding the downside of setting an excessively aggressive price when the signal is a hallucination.
Finally, for very large signals ($s \ge U_\gamma$), the seller \emph{follows again} and posts $p^*(s)=s$, as the upside from extracting a very large value in the non-hallucinatory event dominates the additional risk from pricing at the signal.

We provide a visual representation of the relevant virtual value shapes for each of the four regions in \Cref{fig:virtual_single_buyer}.

    \begin{figure}[h!]
    \centering
    
    \subfigure[$s = 0.1$]{
    \begin{tikzpicture}
    \begin{axis}[
        width=0.45\textwidth,
        xlabel={$v$},
        ylabel={ironed virtual value},
        legend style={font=\tiny},
        legend pos=south east,
        table/col sep=comma,   %
        grid=both,
    ]
    
    \addplot [red, thick, unbounded coords=jump] table [x=v, y={iron_virtual_val_s_0.1}] {Data/virtual_values_beta1_2_gamma_0_77.csv};
    
    \addplot [blue, dashed, thick, unbounded coords=jump] table [x=v, y={virtual_val_s_0.1}] {Data/virtual_values_beta1_2_gamma_0_77.csv};
    
    \draw[black, dashed] (axis cs:0.3328333333, \pgfkeysvalueof{/pgfplots/ymin}) -- (axis cs:0.3328333333, 0);  
    \filldraw[black] (axis cs:0.3328333333,\pgfkeysvalueof{/pgfplots/ymin}) circle (2pt) node[anchor=south west]{\footnotesize $p^*=0.33$};
    
    \end{axis}
    \end{tikzpicture}
    }
    \hfill
    \subfigure[$s = 0.5$]{
    \begin{tikzpicture}
    \begin{axis}[
        width=0.45\textwidth,
        xlabel={$v$},
        ylabel={ironed virtual value},
        legend style={font=\tiny},
        legend pos=south east,
        table/col sep=comma,   %
        grid=both,
    ]
    \addplot [red, thick, unbounded coords=jump] table [x=v, y={iron_virtual_val_s_0.5}] {Data/virtual_values_beta1_2_gamma_0_77.csv};
    
    \addplot [blue, dashed, thick, unbounded coords=jump] table [x=v, y={virtual_val_s_0.5}] {Data/virtual_values_beta1_2_gamma_0_77.csv};
    
    \draw[black, dashed] (axis cs:0.5, \pgfkeysvalueof{/pgfplots/ymin}) -- (axis cs:0.5, 0);  
    \filldraw[black] (axis cs:0.5,\pgfkeysvalueof{/pgfplots/ymin}) circle (2pt) node[anchor=south west]{\footnotesize $p^*=0.5$};

    \end{axis}
    \end{tikzpicture}
    }

    \subfigure[$s = 0.8$]{
    \begin{tikzpicture}
    \begin{axis}[
        width=0.45\textwidth,
        xlabel={$v$},
        ylabel={ironed virtual value},
        legend style={font=\tiny},
        legend pos=south east,
        table/col sep=comma,   %
        grid=both,
    ]
    \addplot [red, thick, unbounded coords=jump] table [x=v, y={iron_virtual_val_s_0.8}] {Data/virtual_values_beta1_2_gamma_0_77.csv};
    
    \addplot [blue, dashed, thick, unbounded coords=jump] table [x=v, y={virtual_val_s_0.8}] {Data/virtual_values_beta1_2_gamma_0_77.csv};
    
    \draw[black, dashed] (axis cs:0.5587240437, \pgfkeysvalueof{/pgfplots/ymin}) -- (axis cs:0.5587240437, 0);  
    \filldraw[black] (axis cs:0.5587240437,\pgfkeysvalueof{/pgfplots/ymin}) circle (2pt) node[anchor=south west]{\footnotesize $p^*=0.56$};

    \end{axis}
    \end{tikzpicture}
    }
    \hfill
    \subfigure[$s = 0.95$]{
    \begin{tikzpicture}
    \begin{axis}[
        width=0.45\textwidth,
        xlabel={$v$},
        ylabel={ironed virtual value},
        legend style={font=\tiny},
        legend pos=south east,
        table/col sep=comma,   %
        grid=both,
    ]
    \addplot [red, thick, unbounded coords=jump] table [x=v, y={iron_virtual_val_s_0.95}] {Data/virtual_values_beta1_2_gamma_0_77.csv};
    
    \addplot [blue, dashed, thick, unbounded coords=jump] table [x=v, y={virtual_val_s_0.95}] {Data/virtual_values_beta1_2_gamma_0_77.csv};
    
    \draw[black, dashed] (axis cs:0.95, \pgfkeysvalueof{/pgfplots/ymin}) -- (axis cs:0.95, 0);  
    \filldraw[black] (axis cs:0.95,\pgfkeysvalueof{/pgfplots/ymin}) circle (2pt) node[anchor=south east]{\footnotesize $p^*=0.95$};

    \end{axis}
    \end{tikzpicture}
    }
    
    \caption{\textbf{Illustration of the four signal regimes in the single-buyer case.}
The figure displays the unironed (dashed blue) and ironed (solid red) virtual value functions corresponding to the four regimes characterized in \Cref{prop:optimal_price_corrected}, for a Beta$(1,2)$ prior and $\gamma = 0.77$. In each panel, the vertical dashed line indicates the smallest value at which the ironed virtual value crosses zero, which corresponds to the revenue-maximizing posted price.}
    \label{fig:virtual_single_buyer}
    \end{figure}

\Cref{fig:virtual_single_buyer} provides a visual proof of the structure described in \Cref{prop:optimal_price_corrected}. Each panel plots the unironed posterior virtual value (dashed) and its ironed counterpart (solid), and the optimal posted price corresponds to the smallest value at which the ironed virtual value crosses zero.

Panel (a) ($s = 0.1$) illustrates the \emph{ignore} regime. When the realized signal is very low, reacting to it is not profitable. Formally, the posterior virtual value remains negative throughout the region below $s$, so ironing does not create a new zero crossing there. Hence, the crossing happens after $s$. By \Cref{thm:main}, the ironed virtual value therefore coincides with the prior virtual value in the neighborhood where it crosses zero, implying that the optimal price equals the prior monopoly price $p^{\mathrm{ignore}}$.
Panel (b) ($s=0.5$) corresponds to the follow regime for intermediate signals ($L_\gamma \le s < M_\gamma$). In this region, the prior virtual value at $s$ is already positive. Consequently, the ironed virtual value crosses zero exactly at $v=s$, making $p^*(s)=s$ optimal.
Panel (c) ($s=0.8$) illustrates the cap regime. In this case, the unironed virtual value function before $s$ crosses $0$, at $p^{\mathrm{cap}} < s$. Importantly, we see that unironed virtual value function starts to ``bend'', but in this regime, the bending is mild enough for the ironed virtual value to remain above $0$ in that region.
Finally, panel (d) ($s=0.95$) corresponds to the follow-again regime for very large signals ($s \ge U_\gamma$). Unlike panel (b), hallucinations now distort the virtual value before $s$ so strongly that, after ironing, the ironed virtual value remains strictly negative for all $v < s$. Consequently, no price below $s$ is optimal, and the optimal posted price again equals the signal.

The key structural fact implying that the number of regimes is at most $4$ is that, under log-concavity of $f$, the unironed posterior virtual value on $[a,s)$ crosses $0$ at most twice. As a result, once the ironed virtual value function drops below zero before $s$ in the second follow regime, it will not move above $0$ again for higher values of $s$. This guarantees that, as $s$ varies, the optimal price can only transition through the four regimes observed in \Cref{fig:virtual_single_buyer}, yielding the simple piecewise structure in \Cref{prop:optimal_price_corrected}.

\begin{remark}[More pieces without log-concavity]
The four-piece ``ignore--follow--cap--follow'' structure in \Cref{prop:optimal_price_corrected} relies on the log-concavity assumptions on $f$. 
Without log-concavity, one can construct a regular %
prior for which $p^*(s)$ has strictly more than four pieces. \Cref{app:counterexample-many-regimes} provides such a counterexample based on a mixture of two Betas.
\end{remark}

\subsection{Comparison to the Value-with-Noise Model} \label{sec:noise}

We now illustrate how the structural insight from the noise-free setting $(\sigma = 0)$ characterized in \Cref{prop:optimal_price_corrected} informs the design of mechanisms in the noisy model of \Cref{sec:model} ($\sigma > 0)$.

\subsubsection{Numerical comparison of different models of error}

When $\sigma>0$ the induced posterior admits a density, and we compute the optimal posted price by solving the posterior revenue-maximization problem signal by signal.
We also consider a pure value-with-noise model obtained by setting $\gamma=0$, and in which $s=v+\varepsilon$ is a purely local perturbation of the value that carries no reliability risk. This model corresponds to the classical measurement error model in the literature~\citep{fuller1987measurement}.

\Cref{fig:comparison_hall_noise} reports the optimal posted price function $p^*(s)$ under these models for representative values of $\gamma$.

\begin{figure}[h!]
\centering
\subfigure[$\gamma = 0.75$]{
\begin{tikzpicture}
\begin{axis}[
    width=0.32\textwidth,
    height=0.28\textwidth,
    tick label style={font=\footnotesize},
    label style={font=\footnotesize},
    xtick={0,0.5,1},
    xlabel={$s$},
    ylabel={$p^*(s)$},
    table/col sep=comma,
    grid=both,
    legend to name=leg:comparison_hall_noise,
    legend columns=1,
    legend style={font=\small,draw=none,fill=none},
    legend cell align=left,
    ymin=0, ymax=1
]
\addplot [red,  unbounded coords=jump, line width = .5mm] table [x=s, y={p_hall_gamma_0.75}]       {Data/one_buyer_prices_all_models_beta_1_2.csv};
\addlegendentry{Hallucinations only ($\sigma = 0$)}

\addplot [gray, unbounded coords=jump, line width = .5mm] table [x=s, y={p_hall_noise_gamma_0.75}] {Data/one_buyer_prices_all_models_beta_1_2.csv};
\addlegendentry{Hallucinations + noise ($\sigma = 0.1$)}

\addplot [blue, dashed, unbounded coords=jump, line width = .5mm] table [x=s, y={p_signal_noise_gamma_0.75}] {Data/one_buyer_prices_all_models_beta_1_2.csv};
\addlegendentry{Noise only ($\gamma = 0,\ \sigma = 0.1$)}

\end{axis}
\end{tikzpicture}
}
\hfill
\subfigure[$\gamma = 0.77$]{
\begin{tikzpicture}
\begin{axis}[
    width=0.32\textwidth,
    height=0.28\textwidth,
    tick label style={font=\footnotesize},
    label style={font=\footnotesize},
    xtick={0,0.5,1},
    xlabel={$s$},
    ylabel={$p^*(s)$},
    table/col sep=comma,
    grid=both,
    ymin=0, ymax=1
]
\addplot [red,  unbounded coords=jump, line width = .5mm] table [x=s, y={p_hall_gamma_0.77}]       {Data/one_buyer_prices_all_models_beta_1_2.csv};

\addplot [gray, unbounded coords=jump, line width = .5mm] table [x=s, y={p_hall_noise_gamma_0.77}] {Data/one_buyer_prices_all_models_beta_1_2.csv};

\addplot [blue, dashed, unbounded coords=jump, line width = .5mm] table [x=s, y={p_signal_noise_gamma_0.77}] {Data/one_buyer_prices_all_models_beta_1_2.csv};

\end{axis}
\end{tikzpicture}
}
\hfill
\subfigure[$\gamma = 0.80$]{
\begin{tikzpicture}
\begin{axis}[
    width=0.32\textwidth,
    height=0.28\textwidth,
    tick label style={font=\footnotesize},
    label style={font=\footnotesize},
    xtick={0,0.5,1},
    xlabel={$s$},
    ylabel={$p^*(s)$},
    table/col sep=comma,
    grid=both,
    ymin=0, ymax=1
]
\addplot [red,  unbounded coords=jump, line width = .5mm] table [x=s, y={p_hall_gamma_0.8}]        {Data/one_buyer_prices_all_models_beta_1_2.csv};

\addplot [gray, unbounded coords=jump, line width = .5mm] table [x=s, y={p_hall_noise_gamma_0.8}]  {Data/one_buyer_prices_all_models_beta_1_2.csv};

\addplot [blue, dashed, unbounded coords=jump, line width = .5mm] table [x=s, y={p_signal_noise_gamma_0.8}] {Data/one_buyer_prices_all_models_beta_1_2.csv};

\end{axis}
\end{tikzpicture}
}
\par\medskip
\pgfplotslegendfromname{leg:comparison_hall_noise}
\caption{Optimal price function $p^*(s)$ under three signal-generation models, with a Beta$(1,2)$ prior. The panel-specific hallucination probability applies to the red and gray curves; the blue dashed curve is the same no-hallucination benchmark in every panel.}
\label{fig:comparison_hall_noise}
\end{figure}

Several insights emerge from \Cref{fig:comparison_hall_noise}. First, the error model has a first-order qualitative effect on the shape of the optimal price function. Under the pure value-with-noise model, errors are local. The optimal price therefore behaves like a shrinkage correction: it raises low signals and lowers high signals, reflecting regression to the mean. With hallucinations, with or without additional noise, the optimal rule is non-monotonic.

Second, the noise-free characterization in \Cref{prop:optimal_price_corrected} provides useful qualitative guidance for the noisy hallucination model. The noisy optimal price remains close to the same ignore, follow, and cap logic. Therefore the noise-free analysis provides a valuable sharp characterization that formalizes this logic. Moreover, the mechanism in \Cref{prop:optimal_price_corrected} is summarized by at most six parameters: the thresholds that delimit the ignore, follow, and cap regions, of which there are at most three, the prior monopoly price used in the ignore region, the cap price, and the signal itself in the follow regions. This makes the mechanism substantially simpler to describe and implement than the fully numerical optimum obtained when $\sigma > 0$.

In what follows, we formalize this intuition. We first show that applying the noise-free rule without adjustment can lead to a high deterioration of the revenue under arbitrarily small noise. We then construct a one-parameter correction that is near-optimal in the noisy environment.

\subsubsection{Simple near-optimal mechanism under noise}
The approximation above suggests using the noise-free mechanism $\mathcal{M}^* = p^*$ directly in the noisy setting, extended outside $[0,b]$ by $p^*(s)=p^{\mathrm{ignore}}$ for $s<0$ and $p^*(s)=b$ for $s>b$. For a posted-price rule $p$, we average its conditional revenue over signals and write
\begin{equation*}
\Rev_\sigma(p) \;:=\; \mathbb{E}_s\bigl[\Rev(p(s),F_{\gamma,s,\sigma})\bigr] =\; \mathbb{E}_{(v,s)}\bigl[\,p(s)\cdot\mathbbm{1}\{v \ge p(s)\}\,\bigr],
\end{equation*}

where $(v,s)$ are drawn jointly according to the noisy model of \Cref{sec:model} and $F_{\gamma,s,\sigma}$ is the posterior distribution of $v$ given $s$. The prior $F$ and hallucination probability $\gamma$ are held fixed in this notation, and $\Rev_0$ denotes the noise-free revenue. We identify a posted-price mechanism with its price rule when writing $\Rev_\sigma(\mathcal{M})$, and let $\mathcal{M}^*_\sigma$ denote a revenue-maximizing posted-price rule under the noisy model.
Applying $\mathcal{M}^*$ as is turns out to be brittle. On the follow regime $\mathcal{S}^{\mathrm{follow}} := [L_\gamma, M_\gamma) \cup (U_\gamma, b]$ the seller posts $p^*(s)=s$; under noise, whenever $\varepsilon>0$ the posted price exceeds the value and the buyer rejects, so revenue on the informative event drops sharply. Our next result formalizes this.

\begin{proposition} \label{prop:brittleness_noise}
Assume that the optimal noise-free mechanism $\mathcal{M}^*$ has a non-empty follow regime, i.e.\ the set $\mathcal{S}^{\mathrm{follow}} := [L_\gamma, M_\gamma) \cup (U_\gamma, b]$ has positive $F$-measure. Then there exists a constant $C>0$, depending only on $F$, such that for every sufficiently small $\sigma > 0$,
\begin{equation*}
\Rev_\sigma(\mathcal{M}^*_\sigma) - \Rev_\sigma(\mathcal{M}^*) \;\ge\; \frac{1-\gamma}{2} \int_{\mathcal{S}^{\mathrm{follow}}} v \, dF(v) \;-\; C\,\sigma\sqrt{2\log(b/\sigma)}.
\end{equation*}
In particular,
\begin{equation*}
\liminf_{\sigma \downarrow 0}\bigl[ \Rev_\sigma(\mathcal{M}^*_\sigma) - \Rev_\sigma(\mathcal{M}^*)\bigr]
\;\ge\; \frac{1-\gamma}{2} \int_{\mathcal{S}^{\mathrm{follow}}} v \, dF(v) \;>\; 0.
\end{equation*}
\end{proposition}

\Cref{prop:brittleness_noise} shows that the revenue loss from deploying the noise-free mechanism $\mathcal{M}^*$ without adjustment does not vanish in the small-noise limit. This loss arises because, on the follow regime, pricing at the signal $p^*(s)=s$ forfeits the sale exactly on the event $\{\varepsilon > 0\}$, which has probability $\tfrac12$ by the symmetry of $\varepsilon$. Despite this negative result, the remedy is simple. Because the buyer rejects when the noise pushes the signal above the value, it suffices to \emph{deflate} the posted price slightly in the follow regime, leaving the ignore and cap regimes unchanged.

\begin{definition}
\label{def:buffered}
For a buffer parameter $\tau \ge 0$, the \emph{noise-buffered mechanism} $\widetilde{\mathcal{M}}^\tau$ is the posted-price rule
\begin{equation*}
\widetilde p^\tau(s) \;:=\;
\begin{cases}
\max\bigl\{\, s - \tau,\;0\,\bigr\}, & \text{if } p^*(s) = s, \\
p^*(s), & \text{otherwise}.
\end{cases}
\end{equation*}
The buffer $\tau$ is common to all signals, and setting $\tau=0$ recovers the noise-free rule $\mathcal{M}^*$.
\end{definition}

To state the resulting guarantee, let $\mathcal{B} := \{0, L_\gamma, M_\gamma, U_\gamma, b\}$ denote the set of regime boundaries, $\kappa := \min\{|x-y| : x,y \in \mathcal{B},\, x \ne y\} \in (0,b]$ the minimum gap between consecutive boundaries, and $\bar f := \sup_{v \in [0,b]} f(v) < \infty$ the maximum of the density.

\begin{proposition}
\label{prop:sensitivity-noise}
Assume that $F$ has a twice continuously differentiable, log-concave density $f$ on $[0,b]$. For $\sigma \in (0, b/\sqrt{e}\,)$, set $\tau = \sqrt{2\log(b/\sigma)}\sigma$. If $\tau \le \kappa/2$, then
\begin{equation*}
\Rev_\sigma(\mathcal{M}^*_\sigma) \;-\; \Rev_\sigma\bigl(\widetilde{\mathcal{M}}^\tau\bigr)
\;\le\; \bigl(2 + 11\, b\,\bar f\bigr)\, \sigma\,\sqrt{2\log(b/\sigma)}.
\end{equation*}
\end{proposition}

\Cref{prop:sensitivity-noise} establishes that appropriately choosing $\tau$ yields a mechanism whose revenue is within order $\sigma\sqrt{\log(b/\sigma)}$ of the optimum for the noisy environment. The correction simply deflates the price by $\tau$ when the noise-free rule follows the signal, preserving the simplicity of that rule.
The parameter $\tau$ can also be tailored numerically to the prior and noise level, using the same family of mechanisms in \Cref{def:buffered}.

\Cref{fig:buffered_revenue_sigma} plots the expected revenue of four mechanisms as $\sigma$ varies: the unbuffered noise-free rule $\mathcal{M}^*$ applied without adjustment, the noise-buffered rule $\widetilde{\mathcal{M}}^\tau$ with the theoretical choice $\tau= \sqrt{2\log(b/\sigma)}\sigma$ of \Cref{prop:sensitivity-noise}, the same rule with $\tau$ tailored numerically to maximize revenue for each prior and noise level $\sigma$, and the optimal mechanism $\mathcal{M}^*_\sigma$ designed against the noisy model. The plot fixes $\gamma = 0.3$ and covers three priors: Uniform$[0,1]$, $\mathrm{Beta}(5,1)$, and a truncated $\mathrm{Lognormal}(0,1.3)$.

\begin{figure}[h!]
\centering
\pgfplotsset{
    every buffered revenue axis/.style={
        width=0.32\textwidth, height=0.30\textwidth,
        xlabel={$\sigma$},
        xmin=0.01, xmax=0.2,
        table/col sep=comma,
        grid=both,
        tick label style={font=\footnotesize, /pgf/number format/fixed},
        label style={font=\footnotesize},
        title style={font=\footnotesize},
    }
}
\begin{tikzpicture}
\begin{axis}[
    every buffered revenue axis,
    ylabel={expected revenue},
    title={Uniform on $[0,1]$},
    legend to name=leg:buffered_revenue_sigma,
    legend columns=4,
    legend style={font=\scriptsize,draw=none,fill=none,column sep=8pt},
    legend cell align=left,
]
\addplot[black!50,thick,dashed,mark=x,mark size=2pt,mark options={solid},mark repeat=6,mark phase=3] table[x=sigma, y=R_unbuffered] {figures/data/buffered_revenue_uniform.csv};
\addlegendentry{Unbuffered}
\addplot[RoyalBlue,thick,mark=square*,mark size=1.5pt,mark repeat=6,mark phase=2] table[x=sigma, y=R_theory] {figures/data/buffered_revenue_uniform.csv};
\addlegendentry{Theoretical buffer}
\addplot[BrickRed,thick,mark=*,mark size=1.5pt,mark repeat=6,mark phase=5] table[x=sigma, y=R_tailored] {figures/data/buffered_revenue_uniform.csv};
\addlegendentry{Tailored buffer}
\addplot[black,line width=1.3pt] table[x=sigma, y=R_noisy_opt] {figures/data/buffered_revenue_uniform.csv};
\addlegendentry{Optimal noisy}
\end{axis}
\end{tikzpicture}
\hfill
\begin{tikzpicture}
\begin{axis}[
    every buffered revenue axis,
    title={$\mathrm{Beta}(5,1)$ on $[0,1]$},
]
\addplot[black!50,thick,dashed,mark=x,mark size=2pt,mark options={solid},mark repeat=6,mark phase=3] table[x=sigma, y=R_unbuffered] {figures/data/buffered_revenue_beta51.csv};
\addplot[RoyalBlue,thick,mark=square*,mark size=1.5pt,mark repeat=6,mark phase=2] table[x=sigma, y=R_theory] {figures/data/buffered_revenue_beta51.csv};
\addplot[BrickRed,thick,mark=*,mark size=1.5pt,mark repeat=6,mark phase=5] table[x=sigma, y=R_tailored] {figures/data/buffered_revenue_beta51.csv};
\addplot[black,line width=1.3pt] table[x=sigma, y=R_noisy_opt] {figures/data/buffered_revenue_beta51.csv};
\end{axis}
\end{tikzpicture}
\hfill
\begin{tikzpicture}
\begin{axis}[
    every buffered revenue axis,
    title={Truncated $\mathrm{Lognormal}(0,1.3)$},
]
\addplot[black!50,thick,dashed,mark=x,mark size=2pt,mark options={solid},mark repeat=6,mark phase=3] table[x=sigma, y=R_unbuffered] {figures/data/buffered_revenue_lognormal13.csv};
\addplot[RoyalBlue,thick,mark=square*,mark size=1.5pt,mark repeat=6,mark phase=2] table[x=sigma, y=R_theory] {figures/data/buffered_revenue_lognormal13.csv};
\addplot[BrickRed,thick,mark=*,mark size=1.5pt,mark repeat=6,mark phase=5] table[x=sigma, y=R_tailored] {figures/data/buffered_revenue_lognormal13.csv};
\addplot[black,line width=1.3pt] table[x=sigma, y=R_noisy_opt] {figures/data/buffered_revenue_lognormal13.csv};
\end{axis}
\end{tikzpicture}

\medskip
\pgfplotslegendfromname{leg:buffered_revenue_sigma}
\caption{Expected revenue of the unbuffered noise-free rule $\mathcal{M}^*$, the noise-buffered mechanism $\widetilde{\mathcal{M}}^\tau$ with either the theoretical buffer $\tau= \sqrt{2\log(b/\sigma)}\sigma$ or a numerically tailored buffer $\tau$, and the noisy-optimal mechanism $\mathcal{M}^*_\sigma$, as functions of the noise level $\sigma$, at $\gamma = 0.3$, for three priors.}
\label{fig:buffered_revenue_sigma}
\end{figure}

Across all three priors, the unbuffered noise-free rule loses substantial revenue even for tiny $\sigma$, whereas the buffered rules recover most of this loss: with a tailored value of $\tau$, revenue is within $10$ to $15\%$ of the noisy-optimal mechanism, despite using a simple pricing scheme which we can express in closed form.

\subsection{Robustness to an Unknown Hallucination Probability}\label{sec:gamma-robust}

The mechanisms designed so far take the hallucination probability $\gamma$ as an input. In practice, $\gamma$ describes the reliability of a prediction technology, and the seller may only be able to estimate it coarsely. We conclude our analysis of the single-buyer problem by removing this input entirely and asking what the seller should post when $\gamma$ is unknown. We show that the regime structure of \Cref{prop:optimal_price_corrected} survives.

We return to the noise-free benchmark ($\sigma = 0$) and place no restriction on the reliability: the seller knows the prior $F$ but only knows that $\gamma \in [0,1]$. Conditional on signal $s$, the revenue from posting $p$ is
\begin{equation*}
    \Rev(p,F_{\gamma,s}) = (1-\gamma) \cdot p \cdot \mathbbm{1}\{p\le s\} + \gamma \cdot \Rev(p,F).
\end{equation*}
Let $p_F$ denote the monopoly price under the prior $F$, and let $p^*_{\gamma,s}$ denote the optimal posted price of \Cref{prop:optimal_price_corrected} when $\gamma$ is known. We evaluate a price $p$ posted at signal $s$ by its worst-case regret against this known-$\gamma$ benchmark,
\begin{equation*}
    \mathrm{WReg}_{s,F}(p)
    = \sup_{\gamma\in[0,1]} \bigl\{\Rev(p^*_{\gamma,s},F_{\gamma,s})-\Rev(p,F_{\gamma,s})\bigr\},
\end{equation*}
and the robust seller minimizes $\mathrm{WReg}_{s,F}$ signal by signal.

Under \Cref{ass:regular}, the revenue curve $p \mapsto \Rev(p,F)$ is unimodal on $[0,b]$. We assume its strict form: the monopoly price $p_F \in (0,b)$ is unique, and $\Rev(\cdot,F)$ is strictly increasing on $[0,p_F]$ and strictly decreasing on $[p_F,b]$.\footnote{We adopt the strict form for expositional clarity; \Cref{thm:gamma_robust:single_local} continues to hold when $\Rev(\cdot,F)$ is only weakly unimodal, with $\underline s$ and $p_s$ interpreted as the leftmost solutions of their defining equations and $p_F$ as the leftmost maximizer.}
\begin{proposition}
\label{thm:gamma_robust:single_local}
Assume $\sigma =0.$
Let $\underline s$ be the unique solution in $[0,p_F]$ to
\begin{equation}
\label{eq:gr_local_low_threshold}
    \underline s = \Rev(p_F,F) - \Rev(\underline s,F),
\end{equation}
and, for every $s>p_F$, let $p_s\in[p_F,s]$ be the unique solution to
\begin{equation}\label{eq:gr_local_high_balance}
    s - p_s = \Rev(p_F,F) - \Rev(p_s,F).
\end{equation}
Then $p \mapsto \mathrm{WReg}_{s,F}(p)$ admits a  minimizer given by,
\begin{equation*}
    \widetilde p(s) =
    \begin{cases}
        p_F, & s \le \underline s,\\
        s, & \underline s < s \le p_F,\\
        p_s, & s > p_F.
    \end{cases}
\end{equation*}
\end{proposition}

\Cref{thm:gamma_robust:single_local} characterizes the minimax-optimal posted price when $\gamma$ is unknown. Notably, it shows that the regime structure of \Cref{prop:optimal_price_corrected} survives: the robust price ignores low signals and posts the monopoly price, follows intermediate signals, and caps high signals at a price strictly between $p_F$ and $s$. The one qualitative change is that the cap is no longer flat: $p_s$ strictly increases with the signal. Furthermore, the follow-again regime at the very top disappears: when $\gamma$ is known, following a very high signal might be optimal because the hallucination downside is known to carry weight only $\gamma$; whereas when it is not known, the pure-hallucination scenario is still a possibility and no matter how high the signal is, the seller keeps hedging.

Under the assumption of \Cref{prop:optimal_price_corrected}, we can further relate the threshold $\underline s$ marking the ignore-to-follow transition back to the known-$\gamma$ thresholds of \Cref{prop:optimal_price_corrected}. Indeed, recall that when $\gamma$ is known, $L_\gamma$ denotes the threshold at which the optimal posted price of \Cref{prop:optimal_price_corrected} switches from ignoring the signal to following it. Our next result shows that $\underline{s}$ equals $L_\gamma$ when $\gamma = 0.5$.
\begin{corollary}
\label{cor:gamma_robust:half}
Assume that $F$ has a twice continuously differentiable, log-concave density $f$ on $[0,b]$. Then $\underline s = L_{1/2}$.
\end{corollary}

We next numerically compare the structure of the minimax robust mechanism to the one when $\gamma$ is known. \Cref{fig:gr_single_buyer_prices} illustrates the robust price under a uniform prior on $[0,1]$ and under a $\mathrm{Beta}(1,2)$ prior, against the known-$\gamma$ optimal prices $p^*_{\gamma,s}$ for $\gamma\in\{0.1,0.5,0.9\}$. Under the uniform prior the robust price admits a closed form,
\begin{equation*}
    \widetilde p(s) =
    \begin{cases}
        1/2, & s \le 1-\sqrt{3}/2,\\
        s, & 1-\sqrt{3}/2 < s \le 1/2,\\
        \sqrt{s-1/4}, & s>1/2.
    \end{cases}
\end{equation*}
Under the $\mathrm{Beta}(1,2)$ prior the same three-piece structure holds with $p_F=1/3$ but the cap branch no longer admits an elementary closed form and is computed numerically.

\begin{figure}[h!]
\centering
\pgfplotsset{
    every gr single buyer axis/.style={
        width=0.45\textwidth, height=0.38\textwidth,
        xlabel={$s$}, ylabel={posted price},
        xmin=0,xmax=1,ymin=0,ymax=1,
        grid=both,
        tick label style={font=\footnotesize},
        label style={font=\footnotesize},
        title style={font=\footnotesize},
    }
}
\begin{tikzpicture}
\begin{axis}[
    every gr single buyer axis,
    title={Uniform on $[0,1]$},
    legend style={at={(0.02,0.98)},anchor=north west,font=\scriptsize,draw=none,fill=none},
    legend cell align=left,
]
\addplot[black,line width=1.3pt,mark=star,mark size=2.5pt,mark repeat=20,mark phase=2] table[x=s, y=robust, col sep=comma] {figures/data/unknown_gamma_prices_uniform.csv};
\addlegendentry{Robust $\widetilde p(s)$}
\addplot[RoyalBlue,thick,mark=square*,mark size=1.5pt,mark repeat=20,mark phase=5]    table[x=s, y=g01, col sep=comma] {figures/data/unknown_gamma_prices_uniform.csv};
\addlegendentry{$\gamma=0.1$}
\addplot[ForestGreen,thick,mark=triangle*,mark size=2pt,mark repeat=20,mark phase=12]  table[x=s, y=g05, col sep=comma] {figures/data/unknown_gamma_prices_uniform.csv};
\addlegendentry{$\gamma=0.5$}
\addplot[BrickRed,thick,mark=*,mark size=1.5pt,mark repeat=20,mark phase=19]     table[x=s, y=g09, col sep=comma] {figures/data/unknown_gamma_prices_uniform.csv};
\addlegendentry{$\gamma=0.9$}
\end{axis}
\end{tikzpicture}
\hspace{0.04\textwidth}
\begin{tikzpicture}
\begin{axis}[
    every gr single buyer axis,
    title={$\mathrm{Beta}(1,2)$ on $[0,1]$},
]
\addplot[black,line width=1.3pt,mark=star,mark size=2.5pt,mark repeat=20,mark phase=2] table[x=s, y=robust, col sep=comma] {figures/data/unknown_gamma_prices_beta12.csv};
\addplot[RoyalBlue,thick,mark=square*,mark size=1.5pt,mark repeat=20,mark phase=5]    table[x=s, y=g01, col sep=comma] {figures/data/unknown_gamma_prices_beta12.csv};
\addplot[ForestGreen,thick,mark=triangle*,mark size=2pt,mark repeat=20,mark phase=12]  table[x=s, y=g05, col sep=comma] {figures/data/unknown_gamma_prices_beta12.csv};
\addplot[BrickRed,thick,mark=*,mark size=1.5pt,mark repeat=20,mark phase=19]     table[x=s, y=g09, col sep=comma] {figures/data/unknown_gamma_prices_beta12.csv};
\end{axis}
\end{tikzpicture}
\caption{Single-buyer robust posted price $\widetilde p(s)$ (thick black line with star markers) against the known-$\gamma$ optimal prices $p^*_{\gamma,s}$.}
\label{fig:gr_single_buyer_prices}
\end{figure}

Two features of \Cref{fig:gr_single_buyer_prices} are worth highlighting. First, in the cap region the robust price runs through the middle of the known-$\gamma$ family: it is more aggressive than the caps designed for unreliable predictions (high $\gamma$) and more conservative than following the signal. Second, as established in \Cref{cor:gamma_robust:half}, the robust ignore-to-follow threshold coincides with the known-$\gamma$ threshold at $\gamma = 1/2$: the robust seller starts following the signal exactly at the point where a seller facing even odds of a hallucination would.

\Cref{fig:gr_single_buyer_regret} plots the regret as a function of the true $\gamma$, for four candidate policies: three fixed-price policies each optimal at a specific design value $\gamma \in \{0.1, 0.5, 0.9\}$, and the minimax price $\widetilde p$ from \Cref{thm:gamma_robust:single_local}.

\begin{figure}[h!]
\centering
\pgfplotsset{
    every gr single buyer regret axis/.style={
        width=0.44\textwidth, height=0.37\textwidth,
        xlabel={true $\gamma$}, ylabel={regret},
        xmin=0,xmax=1,ymin=0,
        table/col sep=comma,
        grid=both,
        tick label style={font=\footnotesize, /pgf/number format/fixed},
        label style={font=\footnotesize},
        title style={font=\footnotesize},
    }
}
\begin{tikzpicture}
\begin{axis}[
    every gr single buyer regret axis,
    title={Uniform on $[0,1]$},
    legend style={at={(0.5,0.98)},anchor=north,font=\scriptsize,draw=none,fill=none},
    legend cell align=left,
    legend columns=2,
]
\addplot[RoyalBlue,thick,mark=square*,mark size=1.5pt,mark repeat=25,mark phase=5]    table[x=gamma, y=regret_g01] {figures/data/unknown_gamma_regret_uniform.csv};
\addlegendentry{$p_{\gamma=0.1}$}
\addplot[ForestGreen,thick,mark=triangle*,mark size=2pt,mark repeat=25,mark phase=12]  table[x=gamma, y=regret_g05] {figures/data/unknown_gamma_regret_uniform.csv};
\addlegendentry{$p_{\gamma=0.5}$}
\addplot[BrickRed,thick,mark=*,mark size=1.5pt,mark repeat=25,mark phase=19]     table[x=gamma, y=regret_g09] {figures/data/unknown_gamma_regret_uniform.csv};
\addlegendentry{$p_{\gamma=0.9}$}
\addplot[black,line width=1.3pt] table[x=gamma, y=regret_robust] {figures/data/unknown_gamma_regret_uniform.csv};
\addlegendentry{Robust $\widetilde p$}
\end{axis}
\end{tikzpicture}
\hspace{0.04\textwidth}
\begin{tikzpicture}
\begin{axis}[
    every gr single buyer regret axis,
    title={$\mathrm{Beta}(5,1)$ on $[0,1]$},
]
\addplot[RoyalBlue,thick,mark=square*,mark size=1.5pt,mark repeat=25,mark phase=5]    table[x=gamma, y=regret_g01] {figures/data/unknown_gamma_regret_beta51.csv};
\addplot[ForestGreen,thick,mark=triangle*,mark size=2pt,mark repeat=25,mark phase=12]  table[x=gamma, y=regret_g05] {figures/data/unknown_gamma_regret_beta51.csv};
\addplot[BrickRed,thick,mark=*,mark size=1.5pt,mark repeat=25,mark phase=19]     table[x=gamma, y=regret_g09] {figures/data/unknown_gamma_regret_beta51.csv};
\addplot[black,line width=1.3pt] table[x=gamma, y=regret_robust] {figures/data/unknown_gamma_regret_beta51.csv};
\end{axis}
\end{tikzpicture}
\caption{Regret as a function of the true $\gamma$ for different mechanisms.}
\label{fig:gr_single_buyer_regret}
\end{figure}

We observe that the price designed at $\gamma = 0.5$ delivers robust performance across the whole range of true $\gamma$: this provides a simple default mechanism when a full minimax computation is undesirable. Furthermore, each tailored price incurs low regret when the true $\gamma$ lies in a neighborhood of its design value, showing that the performance is not too sensitive to a mild misspecification in $\gamma$.

\section{Design of Simple Auctions}\label{sec:heuristic}

The classical Myerson revenue-optimal auction is quite simple when the buyers' valuations are drawn from the same distribution (it reduces to a second-price auction with a reserve price), but it becomes much more complex in the asymmetric case where buyers' valuations are drawn from different priors. In particular, the payment rule requires applying multiple virtual value formulas to determine the minimum bid that would have secured the item for the winning buyer. In our setting, even if the priors are symmetric, they become asymmetric once heterogeneous signals about the buyers are incorporated into the posterior through Bayesian updating. This raises the natural question of whether a simple auction can still perform well. In what follows, we design simple auctions that build on the general characterization in Theorem~\ref{thm:main} and we demonstrate that they generate revenue that is close to optimal.

Section~\ref{sec:heuristic_noiseless} constructs, in the noise-free model, a family of eager second-price auctions whose personalized reserves are derived from the single-buyer optimal prices of Theorem~\ref{thm:main}, and compares it analytically and numerically to natural benchmarks; Section~\ref{sec:heuristic_noise} then adapts this family to the noisy model by deflating the reserves of the follow regime with a buffer, and evaluates the resulting auctions as a function of the noise level $\sigma$.

\subsection{Simple Auctions in the Noiseless Regime}\label{sec:heuristic_noiseless}

In the case of asymmetric buyers, two classical mechanisms used in practice are the eager and lazy second-price auctions with personalized reserve prices (see \cite{paes2016field}). Our aim is to leverage the structure identified in Theorem~\ref{thm:main} to construct a second-price auction with personalized reserves. We focus on the eager mechanism because, for any choice of personalized reserve prices, the revenue it generates is at least as large as the revenue generated by the lazy mechanism (Theorem 4.5 in \cite{paes2016field}).

For $n$ buyers, the eager mechanism is defined as follows. Let $(r_i)_{i \in \{1,\ldots,n\}}$ be the set of personalized reserve prices. The seller first collects bids $b_1,\ldots,b_n$ from the buyers and then constructs the set of active buyers defined as,
\begin{equation*}
    \mathcal{A} = \{ i \in \{1,\ldots,n\} \text{ s.t. } b_i \geq r_i \}.
\end{equation*}
The item is allocated to the highest bidder in $\mathcal{A}$, subject to some tie-breaking rule. We denote the winner by $j$, and the payment is equal to $\max(r_j, \max_{i \in \mathcal{A} \setminus \{j\}} b_i)$.

When considering benchmarks for our proposed simpler auction, two simple auctions are natural candidates. The first is the second price auction with reserve (SPA) that ignores the signals. This is an optimal auction when $\gamma_i=1$ for all buyers, and when buyers are a priori symmetric as the signals are pure hallucinations in this case. The second is an eager auction where the reserve prices are set at the signals, i.e., $r_i = s_i$ for all $i$. We call this auction the signal eager auction. This is an optimal auction when $\gamma_i=0$ for all $i$ as it sells the item to the buyer with the highest valuation with probability 1 in equilibrium. We thus set the goal as being to construct a mechanism that is at least as good as the best of these two mechanisms. 

Theorem~\ref{thm:main} provides a complete characterization of the optimal mechanism and, in particular, provides a procedure to identify the optimal posted price for a buyer with signal $s_i$ when that buyer is considered in isolation. This suggests a natural personalized reserve structure. %
Let $p^*_i(s_i)$ denote the optimal price for buyer $i$ with signal $s_i$ as defined by Theorem~\ref{thm:main}. We define the monopoly-price eager auction by setting $r_i = p^*_i(s_i)$ for all $i$. 

Numerically, the monopoly-price eager auction tends to overperform the signal eager auction when $\gamma$'s are high as it doesn't treat the signals as being accurate like signal eager does (see \Cref{fig:beta}). Perhaps surprisingly, it often underperforms the signal eager policy when $\gamma$'s are low.

This motivates us to propose a family of eager auctions allowing us to improve over these approaches. Each of these auctions is parameterized by an integer $k \in \{0,\ldots,n\}$.
If buyer $i$ has a signal among the $k$ largest ones in $\{s_1,\ldots,s_n\}$, we use the reserve  $r_i = \max(s_i, p^*_i(s_i))$. Otherwise, we use $r_i = p^*_i(s_i)$. That is, we use the optimal monopoly-price as a reserve for lower signals, and the maximum between the optimal monopoly-price and the signal for higher signals. This maximum is meant to remove the cap from the higher reserve prices. We refer to this auction as the $k$-uncapped eager auction.
We note that the $0$-uncapped eager auction is equivalent to the monopoly price eager defined previously.

While the $0$-uncapped eager was not able to obtain a higher revenue than signal eager for all values of $\gamma$, we next establish that in the two-buyer case, the $1$-uncapped eager, leads to a consistent improvement across all $\gamma$'s. 

\begin{theorem}
\label{thm:signal_worse_than_cap}
Let $\gamma>0$ and let $F$ be a regular distribution on $[a,b]$ with continuous density $f$.
Let $\pi_{SE}$ (resp. $\pi_{1}$) be the signal (resp. $1$-uncapped) eager auction.
Then for every pair of signals $(s_1,s_2) \in [a,b]^2$,
\begin{equation*}
  \mathbb{E}_{v_1,v_2}\!\left[ p_1^{\pi_{SE}}(\bm v,\bm s)+p_2^{\pi_{SE}}(\bm v,\bm s)\,\middle|\,\bm s\right]
  \;\leq\;
  \mathbb{E}_{v_1,v_2}\!\left[ p_1^{\pi_1}(\bm v,\bm s)+p_2^{\pi_1}(\bm v,\bm s)\,\middle|\,\bm s\right].
\end{equation*}
\end{theorem}

\Cref{thm:signal_worse_than_cap} establishes the stronger claim, that our $1$-uncapped eager generates weakly more revenue than signal eager, uniformly across all potential signals observed. By taking an expectation, over signals we obtain as a corollary that the expected revenue of the $1$-uncapped eager is weakly higher than the signal eager auction.

\Cref{fig:beta} shows the fraction of the optimal revenue obtained by the different auctions considered in this section.

\begin{figure}[h!]
\centering
\subfigure[]{
\begin{tikzpicture}[scale = 0.68]
\begin{axis}[
    width=10cm,
    height=8cm,
    xlabel={$\gamma$},
    ylabel={Revenue ratio},
    title={},
    legend style={font=\scriptsize},
    legend pos=south west,
    ymin=0.75,
    ymax=1.02,
    table/col sep=comma,
    grid=both,
]

\addplot[red, thick, mark=o]
    table[x=gamma, y={ratio_eager_over_opt}]
    {Data/revenue_ratio_beta_alpha5.0_beta1.0.csv};
\addlegendentry{Signal Eager}

\addplot[violet, thick, mark=diamond*]
    table[x=gamma, y={ratio_naive_over_opt}]
    {Data/revenue_ratio_beta_alpha5.0_beta1.0.csv};
\addlegendentry{Monopoly-price Eager}

\addplot[blue, thick, mark=square*]
    table[x=gamma, y={ratio_topmax_m1_over_opt}]
    {Data/revenue_ratio_beta_alpha5.0_beta1.0.csv};
\addlegendentry{$1$-uncapped Eager}
\end{axis}
\end{tikzpicture}
\label{fig:beta}
}
\subfigure[]{
\begin{tikzpicture}[scale = 0.68]
\begin{axis}[
    width=10cm,
    height=8cm,
    xlabel={$\gamma$},
    ylabel={Revenue ratio},
    title={},
    legend style={font=\scriptsize},
    legend pos=south west,
    ymin=0.75,
    ymax=1.02,
    table/col sep=comma,
    grid=both,
]

\addplot[red, thick, mark=triangle*]
    table[x=gamma, y={ratio_hybrid_over_opt}]
    {Data/revenue_ratio_beta_alpha5.0_beta1.0.csv};
\addlegendentry{Hybrid}

\addplot[blue, thick, mark=square*]
    table[x=gamma, y={ratio_best_topmax_over_opt}]
    {Data/revenue_ratio_beta_alpha5.0_beta1.0.csv};
\addlegendentry{Best $k$-uncapped Eager}
\end{axis}
\end{tikzpicture}
\label{fig:combine_beta}
}
\caption{Fraction of revenue achieved by each auction compared to the optimal auction characterized in \Cref{thm:main} in the two-buyers case, as a function of $\gamma$. Prior distribution: Beta $(\alpha=5,\beta=1)$}
\end{figure}

\Cref{fig:beta} illustrates the improvement allowed by our $1$-uncapped eager auction over the signal eager auction across all values of $\gamma$. However, its performance is worse than the monopoly-price (equivalently the $0$-uncapped) eager auction for high values of $\gamma$. We therefore compare two combinations of auctions which depends on $\gamma$. 

The first combination selects the best among the SPA auction and the signal eager auction. This simple benchmark is a natural heuristic one would use without relying on the characterization derived in this work.
Indeed, this combination of auction, is appealing, as it benefits from the strong performance of SPA for high value of $\gamma$ (recall that it is optimal for $\gamma =1$) and uses signals when they become accurate (as signal eager is optimal when $\gamma = 0$). We will refer to this auction as the Hybrid auction.
Alternatively, one could use the best among the $k$-uncapped eager auctions we introduced. 

\Cref{fig:combine_beta} shows that our best $k$-uncapped eager auction uniformly dominates the Hybrid auction across all values of $\gamma$. While, the improvement could be considered as marginal in the case of the Beta(5,1) distribution illustrated in \Cref{fig:combine_beta}, we see in \Cref{fig:lognormal} that the improvement of our proposed auction can be considerable when the distribution of bid follows a log-normal distribution. 

\begin{figure}[h!]
\centering
\subfigure[Lognormal $(\mu=0,\sigma=1.3)$]{
\begin{tikzpicture}[scale = 0.7]
\begin{axis}[
    width=10cm,
    height=8cm,
    xlabel={$\gamma$},
    ylabel={Revenue ratio},
    title={},
    legend style={font=\scriptsize},
    legend pos=south west,
    ymin=0.75,
    ymax=1.02,
    table/col sep=comma,
    grid=both,
]

\addplot[blue, thick, mark=square*]
    table[x=gamma, y={ratio_best_topmax_over_opt}]
    {Data/revenue_ratio_lognormal_mu0_sigma1_3.csv};
\addlegendentry{Best $k$-uncapped Eager}

\addplot[red, thick, mark=triangle*]
    table[x=gamma, y={ratio_hybrid_over_opt}]
    {Data/revenue_ratio_lognormal_mu0_sigma1_3.csv};
\addlegendentry{Hybrid}
\end{axis}
\end{tikzpicture}
}
\hfill
\subfigure[Lognormal $(\mu=0,\sigma=1.5)$]{
\begin{tikzpicture}[scale = 0.7]
\begin{axis}[
    width=10cm,
    height=8cm,
    xlabel={$\gamma$},
    ylabel={Revenue ratio},
    title={},
    legend style={font=\scriptsize},
    legend pos=south west,
    ymin=0.75,
    ymax=1.02,
    table/col sep=comma,
    grid=both,
]

\addplot[blue, thick, mark=square*]
    table[x=gamma, y=ratio_best_topmax_over_opt]
    {Data/revenue_ratio_lognormal_mu0_sigma1_5.csv};
\addlegendentry{Best $k$-uncapped Eager}

\addplot[red, thick, mark=triangle*]
    table[x=gamma, y=ratio_hybrid_over_opt]
    {Data/revenue_ratio_lognormal_mu0_sigma1_5.csv};
\addlegendentry{Hybrid}
\end{axis}
\end{tikzpicture}
}
\par\medskip
\subfigure[Lognormal $(\mu=0,\sigma=1.8)$]{
\begin{tikzpicture}[scale = 0.7]
\begin{axis}[
    width=10cm,
    height=8cm,
    xlabel={$\gamma$},
    ylabel={Revenue ratio},
    title={},
    legend style={font=\scriptsize},
    legend pos=south west,
    ymin=0.75,
    ymax=1.02,
    table/col sep=comma,
    grid=both,
]

\addplot[blue, thick, mark=square*]
    table[x=gamma, y=ratio_best_topmax_over_opt]
    {Data/revenue_ratio_lognormal_mu0_sigma1_8.csv};
\addlegendentry{Best $k$-uncapped Eager}

\addplot[red, thick, mark=triangle*]
    table[x=gamma, y=ratio_hybrid_over_opt]
    {Data/revenue_ratio_lognormal_mu0_sigma1_8.csv};
\addlegendentry{Hybrid}
\end{axis}
\end{tikzpicture}
}
\caption{Fraction of revenue achieved by each auction compared to the optimal auction characterized in \Cref{thm:main} in the two-buyers case, as a function of $\gamma$, for different lognormal priors.}
\label{fig:lognormal}
\end{figure}
\Cref{fig:lognormal} highlights a qualitatively different picture than in the Beta experiments: 
the revenue gap between our proposed family and the Hybrid auction becomes substantial, and it grows with the dispersion parameter $\sigma$.
Across all three lognormal priors, the \emph{best $k$-uncapped eager} auction remains very close to optimal uniformly over $\gamma$ (always above roughly $0.98$ of the optimal revenue in our experiments), whereas the Hybrid auction exhibits a pronounced dip for intermediate-to-high hallucination rates.
In the heaviest-tailed case ($\sigma=1.8$), the Hybrid auction falls to about $0.82$ of optimal revenue, so selecting the best $k$-uncapped eager auction yields an improvement on the order of $15\%$--$20\%$ in revenue ratio in this region.

The economic driver of this gap is the interaction between hallucinations and heavy tails.
For moderately large $\gamma$, the seller faces a sharp trade-off:
exploiting the signal is valuable when it is accurate, as large signals are informative about extremely high valuations, but trusting the signal too literally is risky, since hallucinated signals can be arbitrarily large outliers.
The Hybrid auction can only switch between two extremes, ignoring the signal altogether (SPA) or fully relying on it (signal eager), and neither extreme is well suited to this regime.
By contrast, the $k$-uncapped eager family spans intermediate policies that cap how aggressively the mechanism responds to large signals, while still leveraging informative signals when they are accurate.
This ability to hedge against hallucinated extremes is a key feature of our algorithm which combines simplicity of implementation and robustness.

\subsection{Simple Auctions under Noise}\label{sec:heuristic_noise}

The auctions designed above use the noise-free price $p^*_i(s_i)$ of \Cref{thm:main} as a reserve, and we have evaluated them in the noise-free model ($\bm{\sigma} = \bm{0}$). We now return to the noisy model of \Cref{sec:model} ($\bm{\sigma} > \bm{0}$) and ask whether the $k$-uncapped eager family remains effective when informative signals contain measurement error. In the single-buyer case, \Cref{prop:brittleness_noise} shows that posting the noise-free price as is forfeits a non-vanishing fraction of the revenue on the follow regime, and \Cref{prop:sensitivity-noise} shows that deflating the price by a buffer in that regime is enough to restore near-optimality. We apply the same correction to the personalized reserves of the $k$-uncapped eager auctions.

Using the buffer parameter $\tau \geq 0$ of \Cref{def:buffered}, let
\begin{equation*}
\widetilde p^\tau_i(s_i) \;:=\;
\begin{cases}
\max\bigl\{\, s_i - \tau,\; a \,\bigr\}, & \text{if } p^*_i(s_i) = s_i, \\
p^*_i(s_i), & \text{otherwise},
\end{cases}
\end{equation*}
so that the follow regime is deflated by $\tau$ while the ignore and cap regimes are left unchanged. The \emph{buffered $k$-uncapped eager auction} is the eager auction with reserves $r_i = \max(s_i - \tau, \widetilde p^\tau_i(s_i))$ if buyer $i$ has a signal among the $k$ largest ones in $\{s_1,\ldots,s_n\}$, and $r_i = \widetilde p^\tau_i(s_i)$ otherwise.
Setting $\tau = 0$ recovers the $k$-uncapped eager auction defined above, and setting $\tau =  \sqrt{2\log(b/\sigma)}\sigma$ uses the theoretical buffer from \Cref{prop:sensitivity-noise}. In the experiments, all buyers have the same noise level $\sigma$, and we tailor $\tau$, common to all buyers and signals, separately for each prior and noise level.
We select $\tau$ from the grid $\{0, 0.25\sigma, \ldots, 5\sigma\}$, keeping the value that maximizes the fraction of the optimal revenue averaged over the plotted values of $\gamma$.

\Cref{fig:noisy_buffer} reports, as a function of $\sigma$, the fraction of the optimal revenue obtained by the Hybrid auction, by the best $k$-uncapped eager auction with the noise-free reserves $p^*_i(s_i)$, and by the best buffered $k$-uncapped eager auction with the tailored value of $\tau$, for the Beta$(5,1)$ and Lognormal$(0,1.8)$ priors used above. The optimal revenue is now that of the optimal auction for the noisy model, obtained by applying Myerson's program to the posterior $\bm{F}_{\bm{\gamma},\bm{s},\bm{\sigma}}$, which admits a density when $\sigma > 0$.
Since the Hybrid auction would otherwise be penalized by the brittleness of the signal eager auction, we let its signal eager component deflate its reserves as well, i.e., the Hybrid auction selects the best among SPA, the signal eager auction, and the signal eager auction with reserves $r_i = \max(s_i - \tau, a)$ evaluated at the theoretical choice $\tau= \sqrt{2\log(b/\sigma)}\sigma$. The lognormal prior is truncated at its $0.5\%$ and $99.5\%$ quantiles so that its support is bounded.

\begin{figure}[h!]
\centering
\subfigure[Beta$(5,1)$]{
\begin{tikzpicture}[scale = 0.64]
\begin{axis}[
    width=10cm,
    height=8cm,
    xlabel={$\sigma$},
    ylabel={Revenue ratio},
    title={},
    ymin=0.5,
    ymax=1.02,
    xmin=0,
    xmax=0.155,
    xtick={0,0.05,0.1,0.15},
    xticklabel style={/pgf/number format/fixed},
    scaled x ticks=false,
    table/col sep=comma,
    grid=both,
    legend style={font=\small, text=black},
    legend columns=3,
    legend to name=leg:noisy_buffer,
]
\addlegendimage{red, thick, mark=triangle*}
\addlegendentry{Hybrid}
\addlegendimage{violet, thick, mark=diamond*}
\addlegendentry{Best $k$-uncapped Eager}
\addlegendimage{blue, thick, mark=square*}
\addlegendentry{Best buffered $k$-uncapped Eager}
\addlegendimage{black, thick, solid}
\addlegendentry{$\gamma = 0.2$}
\addlegendimage{black, thick, dashed}
\addlegendentry{$\gamma = 0.5$}
\addlegendimage{black, thick, dotted}
\addlegendentry{$\gamma = 0.8$}
\addplot[red, thick, solid, mark=triangle*]
    table[x=sigma, y={hyb_gamma0.2}]
    {Data/noisy_tuned_buffer_beta_alpha5_beta1.csv};
\addplot[red, thick, dashed, mark=triangle*]
    table[x=sigma, y={hyb_gamma0.5}]
    {Data/noisy_tuned_buffer_beta_alpha5_beta1.csv};
\addplot[red, thick, dotted, mark=triangle*]
    table[x=sigma, y={hyb_gamma0.8}]
    {Data/noisy_tuned_buffer_beta_alpha5_beta1.csv};
\addplot[violet, thick, solid, mark=diamond*]
    table[x=sigma, y={naive_gamma0.2}]
    {Data/noisy_tuned_buffer_beta_alpha5_beta1.csv};
\addplot[violet, thick, dashed, mark=diamond*]
    table[x=sigma, y={naive_gamma0.5}]
    {Data/noisy_tuned_buffer_beta_alpha5_beta1.csv};
\addplot[violet, thick, dotted, mark=diamond*]
    table[x=sigma, y={naive_gamma0.8}]
    {Data/noisy_tuned_buffer_beta_alpha5_beta1.csv};
\addplot[blue, thick, solid, mark=square*]
    table[x=sigma, y={buf_gamma0.2}]
    {Data/noisy_tuned_buffer_beta_alpha5_beta1.csv};
\addplot[blue, thick, dashed, mark=square*]
    table[x=sigma, y={buf_gamma0.5}]
    {Data/noisy_tuned_buffer_beta_alpha5_beta1.csv};
\addplot[blue, thick, dotted, mark=square*]
    table[x=sigma, y={buf_gamma0.8}]
    {Data/noisy_tuned_buffer_beta_alpha5_beta1.csv};
\end{axis}
\end{tikzpicture}
}
\hspace{0.04\textwidth}
\subfigure[Lognormal$(0,1.8)$]{
\begin{tikzpicture}[scale = 0.64]
\begin{axis}[
    width=10cm,
    height=8cm,
    xlabel={$\sigma$},
    ylabel={Revenue ratio},
    title={},
    ymin=0.5,
    ymax=1.02,
    xmin=0,
    xmax=0.155,
    xtick={0,0.05,0.1,0.15},
    xticklabel style={/pgf/number format/fixed},
    scaled x ticks=false,
    table/col sep=comma,
    grid=both,
]
\addplot[red, thick, solid, mark=triangle*]
    table[x=sigma, y={hyb_gamma0.2}]
    {Data/noisy_tuned_buffer_lognormal_mu0_sigma1_8.csv};
\addplot[red, thick, dashed, mark=triangle*]
    table[x=sigma, y={hyb_gamma0.5}]
    {Data/noisy_tuned_buffer_lognormal_mu0_sigma1_8.csv};
\addplot[red, thick, dotted, mark=triangle*]
    table[x=sigma, y={hyb_gamma0.8}]
    {Data/noisy_tuned_buffer_lognormal_mu0_sigma1_8.csv};
\addplot[violet, thick, solid, mark=diamond*]
    table[x=sigma, y={naive_gamma0.2}]
    {Data/noisy_tuned_buffer_lognormal_mu0_sigma1_8.csv};
\addplot[violet, thick, dashed, mark=diamond*]
    table[x=sigma, y={naive_gamma0.5}]
    {Data/noisy_tuned_buffer_lognormal_mu0_sigma1_8.csv};
\addplot[violet, thick, dotted, mark=diamond*]
    table[x=sigma, y={naive_gamma0.8}]
    {Data/noisy_tuned_buffer_lognormal_mu0_sigma1_8.csv};
\addplot[blue, thick, solid, mark=square*]
    table[x=sigma, y={buf_gamma0.2}]
    {Data/noisy_tuned_buffer_lognormal_mu0_sigma1_8.csv};
\addplot[blue, thick, dashed, mark=square*]
    table[x=sigma, y={buf_gamma0.5}]
    {Data/noisy_tuned_buffer_lognormal_mu0_sigma1_8.csv};
\addplot[blue, thick, dotted, mark=square*]
    table[x=sigma, y={buf_gamma0.8}]
    {Data/noisy_tuned_buffer_lognormal_mu0_sigma1_8.csv};
\end{axis}
\end{tikzpicture}
}
\\[0.5em]
\ref{leg:noisy_buffer}
\caption{Fraction of revenue achieved by each auction compared to the optimal auction for the noisy model in the two-buyers case, as a function of the noise level $\sigma$. The buffered $k$-uncapped eager auction uses a tailored value of $\tau$; the Hybrid auction's buffered component uses $\tau= \sqrt{2\log(b/\sigma)}\sigma$.}
\label{fig:noisy_buffer}
\end{figure}

Three insights emerge from \Cref{fig:noisy_buffer}. First, the best $k$-uncapped eager auction with noise-free reserves degrades substantially as soon as $\sigma > 0$: its worst-case fraction of the optimal revenue over $\gamma$ drops to $0.82$--$0.84$ for the Beta$(5,1)$ prior and $0.56$--$0.70$ for the lognormal prior, and it does not recover as $\sigma$ decreases. This is the multi-buyer counterpart of \Cref{prop:brittleness_noise}: on the follow regime, a reserve set at the signal forfeits the sale whenever the noise pushes the signal above the value. Second, the buffer recovers most of this gap. The best buffered $k$-uncapped eager auction obtains at least $0.90$ of the optimal revenue at every $(\gamma,\sigma)$ in our experiments, and for the lognormal prior it obtains more than $0.97$ of the optimal revenue on average over $\gamma$ at every noise level. 

For the Hybrid format, the comparison mirrors the noise-free one. For the lognormal prior, the Hybrid auction deteriorates as $\gamma$ grows, obtaining $0.92$--$0.94$ of the optimal revenue at $\gamma = 0.5$ and only $0.83$--$0.88$ at $\gamma = 0.8$, whereas the best buffered $k$-uncapped eager auction remains above $0.977$ at every $(\gamma,\sigma)$, an improvement of up to $15$ percentage points. For the Beta$(5,1)$ prior, the two auctions are within one percentage point of each other except at $\gamma = 0.2$ and when $\sigma$ is large, in which case, the hallucination model is not the first order source of statistical error, and is mainly dominated by the noise around the true value $v$.

\section{Conclusion}

In this paper, we studied how Bayesian mechanism design can be adapted to address the challenges posed by hallucination-prone predictions generated by modern machine learning models. By introducing a novel Bayesian framework, we modeled these imperfect signals and rigorously characterized the structure of optimal mechanisms, extending classical results like those of \citet{myerson1981optimal} to settings where posterior distributions lack continuous densities. Our findings provide new insights into how sellers can navigate uncertainty and optimize revenue in environments shaped by unreliable predictions. The analysis starts from a noisy hallucination-prone signal model, uses the zero-noise benchmark to obtain sharp structural results, and then returns to the noisy setting to show how the benchmark mechanism should be buffered when informative signals contain measurement error.

Our framework helps bridge the gap between traditional auction theory and contemporary machine learning applications, where sellers increasingly rely on complex predictive systems that may be unreliable. In particular, our results provide two main insights.
First, in the single-buyer case, hallucination-prone predictions lead to an optimal pricing rule with a regime-based structure that is fundamentally different from the smooth shrinkage behavior induced by classical value-with-noise models. The seller alternates between ignoring, following, and capping the signal, depending on its magnitude. This shows that the structure of the optimal mechanism depends not only on the average accuracy of predictions, but critically on the nature of their errors.
Second, in multi-bidder environments, our results clarify how such prediction-based guidance can and cannot be translated into simple auction formats. Although the optimal signal-revealing mechanism provides bidder-specific pricing recommendations, naively embedding these prices as personalized reserves in a standard eager second-price auction can reduce revenue and underperform simpler benchmarks. In turn, we appropriately design the personalized reserve prices and obtain an eager mechanism with near-optimal performance.

\bibliographystyle{agsm}
\bibliography{ref}

\newpage

\appendix

\renewcommand{\theequation}{\thesection-\arabic{equation}}
\renewcommand{\theproposition}{\thesection-\arabic{proposition}}
\renewcommand{\thelemma}{\thesection-\arabic{lemma}}
\renewcommand{\thetheorem}{\thesection-\arabic{theorem}}
\renewcommand{\thedefinition}{\thesection-\arabic{definition}}
\pagenumbering{arabic}
\renewcommand{\thepage}{App-\arabic{page}}

\setcounter{equation}{0}
\setcounter{proposition}{0}
\setcounter{definition}{0}
\setcounter{lemma}{0}
\setcounter{theorem}{0}

\setstretch{1.4}
\part{Appendix} %
\parttoc %
\section{Key Technical Arguments in the Proof of \Cref{thm:main}}
\label{sec:technical_work}
In this section we present the key technical arguments needed to prove \Cref{thm:main}. We first describe the family of semi-infinite dimensional problems developed in \citet{monteiro2010optimal} to characterize the ironed virtual value for arbitrary distributions. We then solve this family of problems to obtain our closed-form solution.

\subsection{Ironing for Arbitrary Distributions}
\label{sec:gen_ironing}
Let $F$ be a regular distribution which admits a positive density $f$ on its support. 
For any $\gamma \in (0,1)$ and any $s$ in the support of $F$, recall the definition of the posterior distribution $F_{\gamma,s}$ defined in Eq.~\eqref{eq:cumulative-F}.
We note that the posterior distribution does not admit a density at $v = s$. In this setting, the standard ironing of virtual value function (as defined by Myerson) is no longer well defined, because the posterior distribution does not admit a density.
In what follows, we present the formalism developed in \citet{monteiro2010optimal} to characterize the optimal auction for general distributions. This formalism generalizes Myerson's characterization.

For every distribution $F$ (which does not need to have a density), we define for every $x \in [a,b]$ the function
\begin{equation*}
H_{F}(x) = \int_{a}^x t  dF(t) - \int_a^x (1-F(t))dt - a.
\end{equation*}
Fix $t \in [a,b]$. For every $x \in [a,t]$, we define the generalized convex hull of $H_F$ as,
\begin{subequations}
\label{eq:gen_virtual_value}
\begin{alignat}{2}
\Psi_{F}^t(x) = \; &\!\sup_{\alpha,\beta \in \mathbb{R}} &\;& \alpha + \beta \cdot F(x) \\
&\text{s.t.} &      &  \alpha + \beta \cdot F(y) \leq H_{F}(y) \quad \forall y \in [a,t]. 
\end{alignat}
\end{subequations}
Let $\partial \Psi_{F}^t(x)$ be the generalized sub-differential of $\Psi_{F}^t$ at $x$ defined as the set of $\beta \in \mathbb{R}$ such that
\begin{equation}
\label{eq:subgradient}
\Psi_{F}^t(z) \geq \Psi^t_{F}(x) + \beta \cdot (F(z) - F(x)) \quad \text{for every $z \in [a,t]$}.  
\end{equation}
Equivalently (see Section 2 of \citet{monteiro2010optimal}), one has that
\begin{equation}
    \label{eq:subgrad_are_solutions}
    \partial \Psi_{F}^t(x) = \{ \beta \in \mathbb{R} \text{ s.t. there exists $\alpha \in \mathbb{R}$ such that $(\alpha,\beta)$ is optimal for \eqref{eq:gen_virtual_value}} \}.
\end{equation}
Furthermore, let $\ell^t_{F}(x) = \inf \partial \Psi^t_{F}(x)$ and $s^t_{F}(x) = \sup \partial \Psi^t_{F}(x)$\footnote{Note that we will drop dependencies in $t$ when $t = b$, as $\Psi_F^b$ corresponds to the generalized convex hull of $H_F$ on the whole domain $[a,b]$.
}.

\begin{figure}[h]
    \centering
    \begin{tikzpicture}[every text node part/.style={align=center},transform shape,]
    \begin{axis}[
        width=10cm,
        height=8cm,
        xmin=-0.,xmax=2.0,
        ymin=-0.3,ymax=0.01,
        scaled y ticks={base 10:2},
        table/col sep=comma,
        xlabel={$v$},
        ylabel={Negative Revenue},
        grid=both,
        legend pos=south east
    ]

    \addplot [blue, dashed, line width=.7mm] table[x=x,y=H] {Data/ironing_example_mix_truncated_normals.csv};
    \addlegendentry{Before ironing ($H_F$)}

    \addplot [black, line width=.5mm] table[x=x,y=psi] {Data/ironing_example_mix_truncated_normals.csv} ;
    \addlegendentry{Ironed curve ($\Psi_F$)}

 \addplot [violet, line width=.3mm] table[x=x,y expr={-1.22+1.2*\thisrow{F}}] {Data/ironing_example_mix_truncated_normals.csv} node[pos = 0.7,below right] {\footnotesize $y = -1.22+1.2  F(v)$};

    \addplot [red, line width=.3mm] table[x=x,y expr={-0.03-0.2*\thisrow{F}}] {Data/ironing_example_mix_truncated_normals.csv} node[pos = 0.7,below left] {\footnotesize $y = -0.03-0.2  F(v)$};

    \addplot [teal, line width=.3mm] table[x=x,y expr={-0.044-0.08*\thisrow{F}}] {Data/ironing_example_mix_truncated_normals.csv} node[pos = 0.3,below] { \footnotesize $y = -0.044-0.08  F(v)$};

    \end{axis}
    \end{tikzpicture}
    \caption{The figure illustrates the ironing procedure defined by \citet{monteiro2010optimal}. The distribution $F$ used is mixture of two truncated normals on $[0,2]$ with parameters $(0.1,0.04)$ and $(1.9,1.8)$ and respective weights $0.8$ and $0.2$. Instead of the standard convexification in quantile space, \cite{monteiro2010optimal} perform a generalized convexification in the value space where affine functions of $F$ are used to iron the revenue curve.}
    \label{fig:monteiro}
    \end{figure}

\citet{monteiro2010optimal} show that the mapping $\ell_F$ generalizes the notion of ironed virtual value functions for distributions which do not necessarily have a positive density. Figure \ref{fig:monteiro} shows an example of how this kind of ironing works via generalized convexification in value space. In particular, they show that when $F$ admits a positive density on its support, $\ell_F$ is equal to the usual Myerson ironing operator $\mathrm{IRON}_{[a,b]}[F].$ Our next result extends this result to the truncated ironing operator.

\begin{proposition}
    \label{prop:Myerson_and_Monteiro}
    Let $F$ be a distribution with positive density on $[a,b]$. Then, for every $t \in [a,b]$, $\ell_F^t = \mathrm{IRON}_{[a,t]}[F]$. 
\end{proposition}

We note that while \citet{monteiro2010optimal} provide a structural result about the general ironed virtual value function, one still needs to solve in general infinitely many semi-infinite optimization problems to be able to implement the optimal auction. In what follows, we solve Problem \eqref{eq:gen_virtual_value} for our model.

\subsection{Outline of the proof of \Cref{thm:main}}
\label{sec:outline}
Fix a regular distribution $F$ with positive continuous density $f$ on its support $[a,b]$.
The generalized convex hull of $H_{F_{\gamma,s}}$ is defined as,
\begin{subequations}
\begin{alignat}{2}
\Psi_{F_{\gamma,s}}(x) = \; &\!\sup_{\alpha,\beta \in \mathbb{R}} &\;& \alpha + \beta \cdot F_{\gamma,s}(x) \nonumber \\
&\text{s.t.} &      &  \alpha + \beta \cdot F_{\gamma,s}(y) \leq H_{F_{\gamma,s}}(y) \quad \forall y \in [a,b]. \nonumber
\end{alignat}
\end{subequations}

By expressing $F_{\gamma,s}$ and $H_{F_{\gamma,s}}$ as a function of $F$, $H_F$, $\gamma$ and $s$ (see \Cref{lem:F_and_H}), we obtain the following equivalent expression for $\Psi_{F_{\gamma,s}}$. For every $x$ we have that,
\begin{subequations}\label{eq:F_gamma_after_s}
\begin{alignat}{2}
\Psi_{F_{\gamma,s}}(x) = \; &\!\sup_{\alpha,\beta \in \mathbb{R}} &\;& \alpha  + \beta \cdot \gamma \cdot F(x) + \mathbbm{1} \{ x \geq s \} \cdot \beta \cdot (1-\gamma) \\
&\text{s.t.} &      &  \alpha + \beta \cdot \gamma \cdot F(y) \leq \gamma \cdot H_{F}(y) - (1-\gamma) \cdot y \quad \forall y < s. \label{eq:constraint_pre_s} \\ 
&  &      &  \alpha + \beta \cdot (1-\gamma) + \beta \cdot \gamma \cdot F(y) \leq \gamma \cdot H_{F}(y) \quad \forall y \geq s. \label{eq:constraint_post_s}
\end{alignat}
\end{subequations}

To prove \Cref{thm:main}, we aim to relate $\ell_{F_{\gamma,s}}$ to $\ell_{\gamma F}^s$ on the interval $[a,s)$ and $\ell_{F_{\gamma,s}}$ to $\ell_{F}$ on the interval $[s,b]$. Then, by applying \Cref{prop:Myerson_and_Monteiro}, we obtain the desired expression.\\

\noindent \textbf{Key proof technique.} To establish this result, we first prove that the generalized virtual value functions we consider are well-behaved on every interval which does not include $s$. We prove more generally the following result on the continuity of the generalized virtual value function.
\begin{lemma}
    \label{lem:continuity}
    Let $I$ be an interval included in $[a,b]$. Assume that $G$ admits a density $g$ that is positive and continuous on $I$. Then, $\ell_{G}$ is continuous on $I$.
\end{lemma}
Given a distribution $G$, recall that $\ell_{G}$ is the lowest generalized sub-gradient of the function $\Psi_G$ which is itself the generalized convex hull of the function $H_{G}$. Therefore, \Cref{lem:continuity} extends the statement that ``the convex hull of a differentiable function of one variable is continuously differentiable'' to our generalized notions of convexity and differentials.

In turn, the key argument to prove that two distributions of interest $F$ and $G$ have the same virtual value function on some interval consists in first establishing the continuity of $\ell_F$ and $\ell_G$ by using \Cref{lem:continuity}. We then prove that $\ell_F$ is a generalized sub-gradient of $\Psi_G$ on the whole interval and conclude applying the following lemma.
\begin{lemma}
\label{lem:inclusion_to_eq}
Let $F$ and $G$ be two distributions on $[a,b]$, and let $I$ be an interval included in $[a,b]$. If $\ell_F(x) \in \partial \Psi_{G}(x)$ for all $x \in I$, and if $\ell_F$ and $\ell_G$ are continuous on $I$, then $\ell_F = \ell_G$ on $I$.
\end{lemma}
We next show how we relate the generalized virtual value functions of the distributions of interest on the intervals $[a,s)$ and $[s,b]$.\\

\noindent \textbf{Analysis on the interval $[s,b]$.}
We first prove that for some $T$ (defined in \Cref{prop:from_F_to_feasible_Fs}) we have that $\ell_{F_{\gamma,s}}= \ell_F $ over the interval $[T,b]$. As discussed previously, we establish this result by leveraging \Cref{lem:inclusion_to_eq}. Hence, it is sufficient to prove that $\ell_{F}(x) \in \partial \Psi_{F_{\gamma,s}}$ for every $x \in [T,b]$. We note that the definition of the generalized differential presented in \eqref{eq:subgrad_are_solutions} implies that $\ell_{F}(x) \in \partial \Psi_{F_{\gamma,s}}$ if and only if there exists an optimal solution for Problem \eqref{eq:F_gamma_after_s} where $\beta = \ell_F(x)$. In what follows, we construct such a solution.

Let $x \in [s,b]$ and remark that \eqref{eq:subgrad_are_solutions} implies that there exists $(\aF,\bF)$ such that $\bF = \ell_{F}(x)$ which is optimal for Problem \eqref{eq:gen_virtual_value}. We define our related candidate solution for Problem \eqref{eq:F_gamma_after_s} as,
\begin{equation}
\label{eq:candidate}
(\aFs,\bFs) = (\gamma \cdot \aF - (1-\gamma) \cdot \bF, \bF).
\end{equation}
A critical aspect of the construction in \eqref{eq:candidate} is that $\bFs = \bF = \ell_{F}(x)$. Therefore, proving optimality of $(\aFs,\bFs)$ for Problem \eqref{eq:F_gamma_after_s} implies that $\ell_F(x) \in \partial \Psi_{F_{\gamma,s}}$.

A straightforward algebraic manipulation allows us to show that for every $x \in [s,b]$ the couple $(\aFs,\bFs)$ satisfies the constraint \eqref{eq:constraint_post_s} for every $y \geq s$.
However, the constraints \eqref{eq:constraint_pre_s} are not necessarily satisfied for all $x \in [s,b]$. We define the threshold $T$ such that $(\aFs,\bFs)$ satisfies the constraint \eqref{eq:constraint_pre_s} for all $y < s$.  
To that end, we define the following auxiliary mapping. For every $y \leq s$, let $\mu_y$ be defined as,
\begin{equation}
\label{eq:mu_y}
\mu_y(x) = \aFs + \gamma \cdot \bFs \cdot F(y) - \gamma \cdot H_{F}(y) + (1-\gamma) \cdot y \quad \text{for every $x \in [s,b]$.}
\end{equation}
This definition, implies that $(\aFs,\bFs)$ satisfies the constraint \eqref{eq:constraint_pre_s} at a given $y$ if and only if, $\mu_{y}(x) \leq 0$.  Consequently, the feasibility of $(\aFs,\bFs)$ for Problem \eqref{eq:F_gamma_after_s} reduces to the analysis of the sign of $\mu_y$. Our next result provides structural properties about $\mu_y$.
\begin{lemma}\label{lem:prop_mu}
\,
\begin{enumerate}
\item[(i)] $\mu_y$ is non-increasing for every $y \leq s$.
\item[(ii)] If $y > y'$, then for every $x \in [s,b]$, $\mu_{y}(x) > \mu_{y'}(x)$.
\item[(iii)] $\mu_s(s) > 0$ and $\mu_s(b) \leq 0$.
\end{enumerate}
\end{lemma}
\Cref{lem:prop_mu} implies, by property $(ii)$, that for every $(\aFs,\bFs)$ the most stringent constraint \eqref{eq:constraint_pre_s}  is for $y =s$. Furthermore, property $(i)$ implies that if $(\aFs,\bFs)$ satisfies \eqref{eq:constraint_pre_s} for a given $y$ and a given $x$ then for all $x' \geq x$, $(\aFs[x'],\bFs[x'])$ also satisfies \eqref{eq:constraint_pre_s} at $y$. By using these results, we construct a threshold $T$ such that the $(\aFs,\bFs)$ is feasible for all $x \geq T$. More generally, we prove the optimality of $(\aFs,\bFs)$ for Problem \eqref{eq:F_gamma_after_s} and establish the following result.
\begin{lemma}
\label{prop:from_F_to_feasible_Fs}
There exists $T \in (s,b]$ such that $\mu_s(T) = 0$. Furthermore, for every $x \in [T,b]$, we have that $\ell_{F}(x) \in \partial \Psi_{F_{\gamma,s}}(x)$. 
\end{lemma}
Combining \Cref{lem:continuity}, \Cref{lem:inclusion_to_eq} and \Cref{prop:from_F_to_feasible_Fs} we conclude that $\ell_F = \ell_{F_{\gamma,s}}$ on $[T,b]$. We complete the proof on the interval $[s,b]$ by showing that $\ell_{F_{\gamma,s}}$ is constant on $[s,T]$.\\

\noindent \textbf{Analysis on the interval $[a,s)$.}
On this interval, we show that $\ell_{F_{\gamma,s}}=\ell_{\gamma F}^s$, where $\ell_{\gamma F}^s$ is defined as the smallest generalized sub-gradient of the function, defined for every $x \in [a,s)$ as
\begin{subequations}\label{eq:main-gammaF}
\begin{alignat}{2}
\Psi^s_{\gamma F}(x) = \; &\!\sup_{\alpha,\beta \in \mathbb{R}} &\;& \alpha +  \beta \cdot \gamma \cdot F(x) \\
&\text{s.t.} &      &  \alpha + \beta \cdot \gamma \cdot F(y) \leq \gamma \cdot H_{F}(y) - (1-\gamma) \cdot y \quad \forall y \in [a,s]. %
\end{alignat}
\end{subequations}
We note that for every $x \in [a,s)$, Problem \eqref{eq:main-gammaF} is a relaxation of Problem \eqref{eq:F_gamma_after_s} in which we removed the constraint \eqref{eq:constraint_post_s}. The main argument consists in proving that the relaxation is tight in the sense that the value of both problems is the same. 
In particular, we establish that for every $x \in [a,s)$, either $\ell_{F_{\gamma,s}}(x) = \ell_{F_{\gamma,s}}(s)$ or $\ell_{F_{\gamma,s}}(x) \in \partial \Psi_{\gamma F}^s(x)$. By \Cref{lem:inclusion_to_eq} we then conclude that $\ell_{F_{\gamma,s}}(x) \in \{\ell_{F_{\gamma,s}}(s), \ell_{\gamma F}^s(x)\}.$
Using a continuity argument, we conclude that $\ell_{F_{\gamma,s}}$ must equal $\ell_{\gamma F}^s$ on the whole interval $[a,s)$.

The complete proof of \Cref{thm:main} is presented in \Cref{sec:apx_main_proof}.

\section{Full-Surplus Extraction for Arbitrary Mechanisms}
\label{sec:apx_full}
In this section, we provide an example with a single buyer, in which the seller can extract the full surplus from the buyers by using mechanisms which are non-signal-revealing.

Recall the definition of an optimal signal-revealing direct mechanism in \eqref{eq:optimal_mechanism}. We now consider the more general set of direct mechanisms which do not reveal the signal. In this case, the IC and IR constraints become that for every $i, v_i$ and $v_i'$,
\begin{align}
\label{eq:gen_ICR}
    \mathbb{E}_{(\bm{v_{-i},\bm{s})} \sim \bm{J}_{\vert v_i}} \left[ v_i \cdot x_i(v_i,\bm{v}_{-i},\bm{s}) - p_i(v_i,\bm{v}_{-i},\bm{s}) \right]
    &\geq \mathbb{E}_{(\bm{v_{-i},\bm{s})} \sim \bm{J}_{\vert v_i}} \left[ v_i \cdot x_i(v'_i,\bm{v}_{-i},\bm{s}) - p_i(v'_i,\bm{v}_{-i},\bm{s}) \right], \\
    \mathbb{E}_{(\bm{v_{-i},\bm{s})} \sim \bm{J}_{\vert v_i}} \left[ v_i \cdot x_i(v_i,\bm{v}_{-i},\bm{s}) - p_i(v_i,\bm{v}_{-i},\bm{s}) \right],
    &\geq 0
\end{align}
where $\bm{J}_{\vert v_i}$ is the posterior distribution of $(\bm{v_{-i}},\bm{s})$ given the observed type $v_i$.

We note that these conditions differ from the IC and IR constraints established in 
\eqref{eq:optimal_mechanism} in that each buyer is now satisfying the IC and IR constraints in expectation over both the competitors types and all the signals and by only conditioning on their own type.

Our next result shows that there exists a mechanism which satisfies these non-signal-revealing IC and IR constraints that extracts the full surplus.

\begin{proposition}
\label{prop:full_surplus}
    Assume $\sigma = 0$. Consider one buyer with a prior distribution $F$ such that $v = 1$ with probability $\alpha$ and $v=2$ otherwise, for some $\alpha \in (0,1)$. Then for every $\epsilon \in (0,1)$, there exists a mechanism $(x_\epsilon,p_\epsilon)$, which satisfies \eqref{eq:gen_ICR} and such that, the expected revenue of the mechanism is equal to $\mathbb{E}_{v \sim F}[v] - \epsilon$.
\end{proposition}
\Cref{prop:full_surplus} implies that the optimal revenue is equal to the full expected surplus.

\begin{proof}[\textbf{Proof of \Cref{prop:full_surplus}}]
    We will first construct our candidate mechanism. Let, 
    \begin{align*}
        q_1 &:= \mathbb{P}\left( s = 1 \, \vert \, v = 1  \right) = 1 - \gamma + \gamma \cdot \alpha,\\
        q_2 &:= \mathbb{P}\left( s = 1 \, \vert \, v = 2  \right) = \gamma \cdot \alpha,
    \end{align*}
and note that $q_1 \neq q_2$ when $\gamma < 1$. Consequently, the following linear system
\[
\begin{cases}
q_1 \cdot c_1 + (1 - q_1) \cdot c_2 = 1, \\
q_2 \cdot c_1 + (1 - q_2) \cdot c_2 = 0
\end{cases}
\]
admits a unique solution $c_1 = \frac{1-q_2}{1-\gamma}$ and $c_2 = -\frac{q_2}{1-\gamma}$. Hence, the mapping $w:\{1,2\}\to\mathbb R$ defined by $w(1)=c_1,\;w(2)=c_2$ satisfies,
\begin{equation}\label{E:moment}
  \mathbb{E}\left[w(s)\, \vert \, v=1\right]=1\quad \mbox{and} \quad
    \mathbb{E}\left[w(s)\, \vert \, v=2\right]=0.
\end{equation}
Our candidate mechanism $(x_\epsilon,p_\epsilon)$ is defined as follows. For every $(v,s) \in \{1,2\}^2$, we have that $x_\epsilon(v,s) = 1$, that is to say the mechanism always allocates the good. The payment rule satisfies for every $s \in \{1,2\}$,
\begin{equation*}
    \begin{cases}
        p_\epsilon(1,s) = 1 - \epsilon + 2  (1-w(s)), \\
        p_\epsilon(2,s) = 2 - \epsilon +  2w(s).
    \end{cases}
\end{equation*}
We next show that this mechanism satisfies the constraints \eqref{eq:gen_ICR}. Let us define the interim utility given a type $v$ and a reported type $\hat{v}$ as, 
\begin{equation*}
   U(v;\hat v):= \mathbb{E}_{s}\left[v \cdot x_\epsilon(\hat v,s)-p_{\epsilon}(\hat v,s) \, \vert \, v \right].
\end{equation*}
By definition of our mechanism and from \eqref{E:moment}, we have that,
\[
  U(1;1)=\epsilon,\qquad 
  U(1;2)= -3 + \epsilon,
\]
\[
  U(2;2)=\epsilon,\qquad 
  U(2;1)= -1 + \epsilon.
\]
Hence, $U(1;1) > U(1;2)$ and $U(2;2) > U(2;1)$ which implies that the IC constraint holds. Furthermore, $U(1,1)$ and $U(2,2)$ are positive, which implies that the IR constraint also holds.

Finally, the expected revenue generated by this mechanism is, equal to 
\begin{equation*}
    \alpha \cdot (1-\epsilon) + (1-\alpha) \cdot (2 - \epsilon) = 2-\alpha - \epsilon = \mathbb{E}_{v \sim F}[v] - \epsilon.
\end{equation*}
\end{proof}

\section{Counter-example for \Cref{thm:main} with Non-regular Prior} \label{sec:apx_non_decomposition_example}
    Consider the distribution $F$ putting a $0.8$ weight on a truncated normal on $[0.5,0.52]$ with mean $0.51$ and std $0.05$, and a $0.2$ weight on the Uniform over $[0,1]$. We note that this distribution is not regular. In \Cref{fig:counter_example}, we compare the value of $\mathrm{IRON}_{[0,s]}[\gamma F]$ and the actual generalized ironed virtual value of  $F_{\gamma,s}$  computed using the method described in \Cref{sec:technical_work}, for $s = 0.53$ and $\gamma = 0.9$.
\begin{figure}[h!]
    \centering
    \begin{tikzpicture}[scale = 0.65]
    \begin{axis}[
        width=10cm,
        height=10cm,
        xmin=0.5,xmax=0.55,
        ymin=0.4,ymax=0.55,
        table/col sep=comma,
        xlabel={$v$},
        ylabel={virtual value},
        grid=both,
        legend pos=north west
    ]
    
    \addplot [black,  line width = 0.7 mm,unbounded coords=jump] table[x=x,y={virtual_value}] {Data/counter_example.csv};
    \addlegendentry{Ironed virtual value}

    \addplot [red,  line width = 0.7 mm,unbounded coords=jump] table[x=x,y={virtual_value_pre_s},restrict expr to domain={\thisrow{x}}{0:0.531}] {Data/counter_example.csv};
    \addlegendentry{$\mathrm{IRON}_{[0,s]}(\gamma F)$}
    
    \draw[blue, dashed, thick] (axis cs:0.53, 0.4) -- (axis cs:0.53, 0.5);
    \filldraw[blue] (axis cs:0.53,0.4) circle (2pt) node[anchor=south west]{\footnotesize $s=0.53$};

    \end{axis}
    \end{tikzpicture}
    \caption{\textbf{Numerical counter-example to the near-decomposition property without regularity.}}
    \label{fig:counter_example}
\end{figure}

\Cref{thm:main} claims that the generalized ironed virtual value of  $F_{\gamma,s}$ should be equal to $\mathrm{IRON}_{[0,s]}[\gamma F]$ for every $v < s$. However, \Cref{fig:counter_example} demonstrates that this statement does not hold in our example.
This figure shows that when $F$ is not regular, the ironing procedure cannot independently be executed on the intervals $[0,s]$ and $[s,1]$ as described in \Cref{thm:main}. Intuitively, when $F$ is not regular, $s$ may lie in a region that already required ironing under the prior distribution $F$. Consequently, when considering the posterior distribution $F_{\gamma,s}$ the values before and after $s$ must be taken into account to properly compute the ironed virtual value around $s$.

\section{Proofs of Results in \Cref{sec:technical_work}}

\noindent \textbf{Notation.} For every $x \in \mathbb{R}$ and every function $f : \mathbb{R} \to \mathbb{R}$, we denote the left-limit of $f$ at $x$ by $f(x-) = \lim_{y \uparrow x} f(y)$.

\subsection{Proof of \Cref{prop:Myerson_and_Monteiro}}
\begin{proof}[\textbf{Proof of \Cref{prop:Myerson_and_Monteiro}}]
The proof follows from section 4.1 in \cite{monteiro2010optimal} by reapplying the same argument to the truncated ironed virtual value function.
\end{proof}

\subsection{Proofs of Lemmas in \Cref{sec:outline}}

\begin{proof}[\textbf{Proof of \Cref{lem:continuity}}]

\textit{Step 1:} We first show that for every $z \in I$, we have that $| \partial \Psi_{G}(z)| =1.$

    We note that for every $z,z' \in I$ we have that,
    \begin{equation*}
     \frac{H_{G}(z) - H_{G}(z')}{G(z)-G(z')} = \frac{\int_{z'}^z t dG(t)}{G(z) - G(z')} - \frac{\int_{z'}^z  (1-G(t))dt}{G(z) - G(z')}.
    \end{equation*}
    Given that $G$ admits a positive density on $I$, we have that $\lim_{z' \to z} \frac{H_{G}(z) - H_{G}(z')}{G(z)-G(z')} = z - \frac{1-G(z)}{g(z)}$ exists and is finite.

    Next, assume for the sake of contradiction that there exists $z \in I$ such that $| \partial \Psi_{G}(z)| > 1$, i.e. $\ell_{G}(z) < s_{G}(z)$. 

We argue that $H_{G}(z) = \Psi_{G}(z)$. 
Note that $G$ and $H_{G}$ are continuous at $z$, and $(vi)$ in \Cref{prop:monteiro} implies that $\Psi_{G}$ is also continuous at $z$.
Hence if, $H_{G}(z) > \Psi_{G}(z)$ we also have that $H_{G}(z-) > \Psi_{G}(z-)$, which implies by \Cref{prop:monteiro} property $(iv)$ and $(v)$ that there exists $z^* > z$ such that for every $x \in (z,z^*)$, we have that $s_{G}(z-) = s_{G}(z) = \ell_{G}(x)$  and $\ell_{G}(x) = \ell_{G}(z)$. Hence, $s_{G}(z) = \ell_{G}(z)$ which contradicts our initial assumption. Therefore,  
\begin{equation}
\label{eq:H_equal_Psi}
 H_{G}(z) = \Psi_{G}(z).
\end{equation}

Furthermore, for every $x \in [a,b]$, let
\begin{equation*}
    \rho(x) = \begin{cases}
        \Psi_{G}(z) + \ell_{G}(z) \cdot(G(x)-G(z)) \quad \text{for $x \in [a,z)$,}\\
        \Psi_{G}(z) + s_{G}(z) \cdot(G(x)-G(z)) \quad \text{for $x \in [z,b)$}.
    \end{cases}
\end{equation*}
By definition of the generalized sub-gradients in \eqref{eq:subgradient}, we have that $\rho(x) \leq \Psi_{G}(x)$. Furthermore, the definition of $\Psi_{G}(x)$ implies that $\Psi_{G}(x) \leq H_{G}(x)$ for every $x \in [a,b]$. 
Therefore, by using \eqref{eq:H_equal_Psi} we have established that, for every $x \in [a,b)$,
\begin{equation*}
    H_{G}(x) \geq \rho(x) = \begin{cases}
         H_{G}(z) + \ell_{G}(z) \cdot(G(x)-G(z)) \quad \text{for $x \in [a,z)$,}\\
         H_{G}(z) + s_{G}(z) \cdot(G(x)-G(z)) \quad \text{for $x \in [z,b)$}.
    \end{cases}
\end{equation*}
This implies that, $\frac{H_{G}(x) - H_{G}(z)}{G(x)-G(z)} \geq s_{G}(z)$ for every $x \in[z,b)$ and $\frac{H_{G}(x) - H_{G}(z)}{G(x)-G(z)} \leq \ell_{G}(z)$ for every $x \in [a,z)$. Then, by taking a limit over $x$ towards $z$, we obtain that,
\begin{equation*}
    \ell_{G}(z) \geq \lim_{x \to z} \frac{H_{G}(x) - H_{G}(z)}{G(x)-G(z)} \geq s_{G}(z).
\end{equation*}
Given that $\ell_{G}(z) <s_{G}(z)$ this leads to a contradiction. Therefore, for every $z \in I$, we have that $| \partial \Psi_{G}(z)| = 1$.

\textit{Step 2:} We next establish that $\ell_{G}$ is continuous on $I$.

We note that $g$ is positive and continuous on $I$. This implies that $G$ is increasing and continuous on $I$ and its inverse function $G^{-1}$ is well-defined in $G(I)$, increasing and continuous. 

Consider $\hat{\Psi}$ defined for every $z' \in G(I)$ as $\hat{\Psi}(z') = \Psi_{G}(G^{-1}(z'))$. We next show that $\hat{\Psi}$ is differentiable on $G(I)$. Let $z \in I$ and recall that $| \partial \Psi_{G}(z)| = 1$. For every $x \in G(I)$, we have that
    \begin{equation*}
         \frac{\hat{\Psi}(G(z)) - \hat{\Psi}(x)}{G(z)-x} = \frac{\Psi_{G}(z) - \Psi_{G}(G^{-1}(x))}{G(z)-G\left( G^{-1}(x) \right)}.
    \end{equation*}
By taking a limit as $x$ tends to $G(z)$ and by noting that $G^{-1}$ is continuous at $G(z)$ we obtain that,
    \begin{equation*}
        \lim_{x \to G(z)} \frac{\hat{\Psi}(G(z)) - \hat{\Psi}(x)}{G(z)-x} = \lim_{x \to G(z)} \frac{\Psi_{G}(z) - \Psi_{G}(G^{-1}(x))}{G(z)-G\left( G^{-1}(x) \right)} = \ell_{G}(z),
    \end{equation*}
where the equality follows from the fact that $\partial \Psi_{G}(z) = \{ \ell_{G}(z) \}$.  Therefore, $\hat{\Psi}$ is differentiable at $G(z)$ and its derivative is $\ell_{G}(z)$. 

We have just established that $\hat{\Psi}$ is a differentiable function on $G(I)$ and its derivative is equal to $\ell_{G} \circ G^{-1}$. This implies that $\hat{\Psi}$ is convex on $G(I)$ as $\ell_{G}$ and $G$ are non-decreasing. Hence, $\hat{\Psi}$ is a uni-variate differentiable function  which is convex. Therefore, it is continuously differentiable on $G(I)$ \cite[Corollary 25.5.1]{rockafellar2007convex}. We conclude that $\ell_{G} \circ G^{-1}$ is continuous on $G(I)$ and the continuity of $G$ on $I$ allows us to conclude by composition that $\ell_{G}$ is continuous on $I$.
\end{proof}

\begin{proof}[\textbf{Proof of \Cref{lem:inclusion_to_eq}}]
We first note that for every $x \in I$, the assumption that $\ell_{F}(x) \in \partial \Psi_{G}(x)$ implies that $\ell_{F}(x) \geq \ell_{G}(x)$. We next show the reverse inequality.

If $|\partial \Psi_{G}(x)| = 1$, we have that $\ell_{F}(x) \in \partial \Psi_{G}(x) =  \{ \ell_{G}(x) \} $ which implies that $\ell_{G}(x) = \ell_{F}(x)$.

Next, assume that $|\partial \Psi_{G}(x)| > 1$.
The item $(i)$ in \Cref{prop:monteiro} implies that there exist at most countably many points that satisfy this. Hence, there exists a sequence $(x_n)_{n \in \mathbb{N}} \in I^{\mathbb{N}}$ such that for every $n \in \mathbb{N}$, we have that $|\partial \Psi_{G}(x_n)| = 1$, and $\lim_{n \to \infty} x_n = x$.
We then note that,
\begin{equation*}
\ell_{G}(x) \stackrel{(a)}{=} \lim_{n \to \infty} \ell_{G}(x_n) \stackrel{(b)}{=} \lim_{n \to \infty} \ell_{F}(x_n) \stackrel{(a)}{=} \ell_{F}(x), 
\end{equation*}
where the equalities $(a)$ hold because $\ell_{F}$ and $\ell_{G}$ are continuous on $I$ and $(b)$ holds because $|\partial \Psi_{G}(x_n)| = 1$ and $x_n \in I$ for all $n \in \mathbb{N}$. 

We conclude that $\ell_{G}(x) = \ell_{F}(x)$ for every $x \in I$.
\end{proof}

\begin{proof}[\textbf{Proof of \Cref{lem:prop_mu}}]

\,

\noindent \textit{(i)} Let $y \leq s$, we first prove that $\mu_y$ is not increasing.

By replacing  $\aFs,\bFs$ with their expressions as a function of $\aF,\bF$, we can rewrite for every $y \leq s$  and $x \geq s$ that,
\begin{equation*}
\mu_{y}(x) = \gamma \cdot \aF + \gamma \cdot \bF \cdot F(y) - (1-\gamma) \cdot \bF - \gamma \cdot H_{F}(y) + (1-\gamma) \cdot y.
\end{equation*}
Next, we will differentiate this expression with respect to $x$. 
First note that,
\begin{equation*}
\frac{d \aF}{dx} \stackrel{(a)}{=} \frac{dH_{F}(x)}{dx} - \frac{d \bF \cdot F(x)}{dx} \stackrel{(b)}{=} - F(x) \frac{d \bF}{dx},
\end{equation*}
where $(a)$ follows from \Cref{lem:charac_PsiF} and $(b)$ holds because $\frac{dH_{F}(x)}{dx} = \bF \cdot f(x)$. In fact, since $F$ is regular and admits a positive density on its support, we have $\bF \cdot f(x) = \ell_F(x) \cdot f(x) = x \cdot f(x) - (1 - F(x))$. The statement then follows by noting that $H_F(x) = a -x\cdot (1-F(x))$.

Hence, we obtain that
\begin{equation*}
\frac{d \mu_{y}(x)}{dx} =  \frac{d \bF}{dx} \cdot \Big ( \gamma F(y) - \gamma F(x) - (1-\gamma) \Big) \leq 0,
\end{equation*}
where the last inequality holds because $F(y) \leq F(x)$ as $y \leq s \leq x$ and $\frac{d \bF}{dx} \geq 0$ because $\bF = \ell_{F}(x)$ which is non-decreasing by item $(ii)$ in \Cref{prop:monteiro}.

\noindent\textit{(ii)} Let $s \geq y > y'$ and let $x \geq s$. We have that,
\begin{align*}
\mu_{y}(x) - \mu_{y'}(x) &= \gamma \cdot \bF \cdot (F(y) - F(y')) - \gamma \cdot ( H_{F}(y) - H_{F}(y')) + (1-\gamma) \cdot (y-y')\\
&\stackrel{(a)}{=}  \gamma \cdot (F(y) - F(y')) \cdot \left( \bF - \frac{H_{F}(y) - H_{F}(y')}{F(y) - F(y')}  \right) + (1-\gamma) \cdot (y-y'),
\end{align*}
where $(a)$ holds because $F(y) > F(y')$ as $F$ has a positive density. Note that $(1-\gamma) \cdot (y-y') > 0$. Hence, to conclude, it is sufficient to show that
\begin{equation}
\label{eq:accroissement}
\bF - \frac{H_{F}(y) - H_{F}(y')}{F(y) - F(y')} \geq 0.
\end{equation}
Let $z = F(y)$ and $z' = F(y')$.
We note that $F$ is regular. Hence, $H_{F}(F^{-1}(\cdot))$ is convex. This implies that the mapping
$q \mapsto \frac{H_{F}(F^{-1}(z)) - H_{F}(F^{-1}(q))}{z-q}$
is non-decreasing. This implies that, 
\begin{equation*}
\frac{H_{F}(F^{-1}(z)) - H_{F}(F^{-1}(z'))}{z-z'} \leq \lim_{q \to z} \frac{H_{F}(F^{-1}(z)) - H_{F}(F^{-1}(q))}{z-q} = \frac{dH_{F}(F^{-1}(z))}{dz} = \bF[y]
\end{equation*}
Hence, we have that,
\begin{equation*}
\frac{H_{F}(y) - H_{F}(y')}{F(y) - F(y')}  \leq \bF[y] \stackrel{(a)}{\leq} \bF,
\end{equation*}
where $(a)$ holds because $\bF[\cdot]$ is non-decreasing by regularity of $F$ and $y \leq x$. Hence, we have proved that \eqref{eq:accroissement} holds, which concludes the proof of $(ii)$.

\noindent\textit{(iii)}
We first note that
\begin{align*}
    \mu_s(s) &= \gamma \cdot \aF[s] + \gamma \cdot \bF[s] \cdot F(s) - (1-\gamma) \cdot \bF[s] - \gamma \cdot H_{F}(s) + (1-\gamma) \cdot s \stackrel{(a)}{=} (1-\gamma) \cdot (s- \bF[s]) \stackrel{(b)}{>} 0, 
\end{align*}
where  $(a)$ follows from \Cref{lem:charac_PsiF} and $(b)$ holds because $\bF[s] = \ell_F(s) =  s - \frac{1-F(s)}{f(s)} < s$.

Furthermore,
\begin{align*}
    \mu_s(b) &= \gamma \cdot \aF[b] + \gamma \cdot \bF[b] \cdot F(s) - (1-\gamma) \cdot \bF[b] - \gamma \cdot H_{F}(s) + (1-\gamma) \cdot s\\
    &= \gamma \cdot \left( H_{F}(b) -  H_{F}(s) + \bF[b] \cdot F(s) - \bF[b] \cdot F(b) \right) + (1-\gamma) \cdot (s - \bF[b]) \\
    &= \gamma \cdot (F(b) - F(s)) \cdot \left( \frac{H_{F}(b) -  H_{F}(s)}{F(b)-F(s)} - \bF[b]  \right) + (1-\gamma) \cdot (s - \bF[b])\\
    &\stackrel{(a)}{\leq} (1-\gamma) \cdot (s - \bF[b]) \stackrel{(b)}{=} (1-\gamma) (s - b) \leq 0,
\end{align*}
where $(a)$ follows from the convexity of $H_F(F^{-1}(\cdot))$  and $(b)$ holds because $\bF[b] = \ell_F(b) =  b - \frac{1-F(b)}{f(b)} = b$ as $F(b) = 1$.

\end{proof}

\begin{proof}[\textbf{Proof of \Cref{prop:from_F_to_feasible_Fs}}]
The existence of $T$ follows from the intermediate value theorem applied to the function $\mu_s$ which is continuous, non-increasing (by property $(i)$ in \Cref{lem:prop_mu}) and satisfies $\mu_s(s) > 0$ and $\mu_s(b) \leq 0$ (by property $(iii)$ in \Cref{lem:prop_mu}).

Let $x \geq T$ and let $\bF = \ell_F(x)$. By \eqref{eq:subgrad_are_solutions}, there exists $\aF$ such that $(\aF,\bF)$ is optimal for Problem~\eqref{eq:gen_virtual_value} at the point $x$. We next show that the candidate solution $(\aFs,\bFs)$ defined in \eqref{eq:candidate} is feasible for Problem~\eqref{eq:F_gamma_after_s}.

The feasibility of $(\aFs,\bFs)$ for Problem~\eqref{eq:F_gamma_after_s} follows from \Cref{lem:feasible_post_T}.
We next show that $(\aFs,\bFs)$ is an optimal solution at $x$.

We note that the constraint \eqref{eq:constraint_post_s} evaluated at $x$ implies that any feasible $(\alpha,\beta)$ should satisfy
\begin{equation*}
\alpha + \beta \cdot (1-\gamma) + \beta \cdot \gamma \cdot F(x) \leq \gamma \cdot H_{F}(x).
\end{equation*}
Given that the LHS is equal to the objective of the problem, the value $\Psi_{F_{\gamma,s}}(x)$ of the problem is lower or equal to $\gamma \cdot H_{F}(x)$. We next show that $(\aFs,\bFs)$ achieves that value. Indeed, we remark that
\begin{align*}
\aFs + \bFs \cdot (1-\gamma) + \bFs \cdot \gamma \cdot F(x) = \gamma \cdot \left( \aF + \bF \cdot F(x) \right) = \gamma \cdot H_{F}(x),
\end{align*}
where the last equality follows from \Cref{lem:charac_PsiF}. 
This shows that $(\aFs,\bFs)$ is optimal for Problem \eqref{eq:F_gamma_after_s} at $x$. By \eqref{eq:subgrad_are_solutions}, this implies that $\ell_{F}(x) = \bF = \bFs   \in \partial \Psi_{F_{\gamma,s}}(x)$. Hence, $\ell_{F}(x) \geq \ell_{F_{\gamma,s}}(x)$.
\end{proof}

\section{Proof of \Cref{thm:main}}
\subsection{Proof of \Cref{thm:main}} \label{sec:apx_main_proof}

\begin{proof}[\textbf{Proof of \Cref{thm:main}}]~

\noindent  \textit{Step 1:}  We first characterize $\ell_{F_{\gamma,s}}(x)$ for $x \in [s,b]$.

Let $T$ be as defined in \Cref{prop:from_F_to_feasible_Fs}. 
\Cref{prop:from_F_to_feasible_Fs} implies that for every $x \in [T,b]$, we have that $\ell_{F}(x) \in \partial \Psi_{F_{\gamma,s}}(x)$. Furthermore, \Cref{lem:continuity} implies that $\ell_F$ and $\ell_{F_{\gamma,s}}$ are continuous on $[T,b]$ because they both admit a positive density on $[T,b]$ (since $T > s$). We conclude from \Cref{lem:inclusion_to_eq} that $\ell_{F_{\gamma,s}}(x) = \ell_{F}(x)$ for every $x \in [T,b]$.

We next prove that $\ell_{F_{\gamma,s}}(x) = \ell_{F}(T)$ for every $x \in [s,T]$. 

Note that, $F_{\gamma,s}(s) - F_{\gamma,s}(s-) = 1- \gamma > 0$. Hence, property $(iii)$ in \Cref{prop:monteiro} implies that
\begin{equation}
\label{eq:l_s_monteiro}
\ell_{F_{\gamma,s}}(s) = \frac{\Psi_{F_{\gamma,s}}(s)-\Psi_{F_{\gamma,s}}(s-)}{F_{\gamma,s}(s)-F_{\gamma,s}(s-)} = \frac{\Psi_{F_{\gamma,s}}(s)-\Psi_{F_{\gamma,s}}(s-)}{1-\gamma}.
\end{equation}
We next show that
\begin{equation}
\label{eq:bounds_psi_s}
    \Psi_{F_{\gamma,s}}(s-) \leq \gamma \cdot H_{F}(s) - (1-\gamma) \cdot s \quad \mbox{and} \quad  \Psi_{F_{\gamma,s}}(s) \geq \gamma \cdot H_{F}(s) - (1-\gamma) \cdot s + (1-\gamma) \cdot \ell_{F}(T).
\end{equation}

On the one hand, we note that for every $x < s$, \eqref{eq:constraint_pre_s} evaluated at $y = x$ implies that $\Psi_{F_{\gamma,s}}(x) \leq \gamma \cdot H_{F}(x) - (1-\gamma) \cdot x$. By taking the  left-limit to $s$ on both sides of the inequality and using the continuity of $H_{F}$ we obtain that $\Psi_{F_{\gamma,s}}(s-) \leq \gamma \cdot H_{F}(s) - (1-\gamma) \cdot s.$

On the other hand, let  $(\aFs[T],\bFs[T])$ be as defined in \eqref{eq:candidate}. \Cref{lem:feasible_post_T} implies that this vector is a feasible solution for Problem \eqref{eq:F_gamma_after_s}.  Therefore, 
\begin{align*}
\Psi_{F_{\gamma,s}}(s) &\stackrel{(a)}{\geq} \aFs[T] + \gamma \cdot \bFs[T] \cdot F(s) + (1-\gamma) \cdot \bFs[T]\\
&\stackrel{(b)}{=} \gamma \cdot H_{F}(s) - (1-\gamma) \cdot s + (1-\gamma) \cdot \bFs[T] \stackrel{(c)}{=}  \gamma \cdot H_{F}(s) - (1-\gamma) \cdot s + (1-\gamma) \cdot \ell_{F}(T),
\end{align*}
where $(a)$ holds by feasibility of $(\aFs[T],\bFs[T])$, $(b)$ holds because $\mu_s(T) = 0$ which implies that $\aFs[T] + \gamma \cdot \bFs[T] \cdot F(s) = \gamma \cdot H_{F}(s) - (1-\gamma)\cdot s$ and $(c)$ follows from the fact that by definition of $\bFs[T]$, we have that $\bFs[T] = \bF[T] = \ell_{F}(T)$.

By replacing in \eqref{eq:l_s_monteiro} the bounds derived for $\Psi_{F_{\gamma,s}}(s-) $ and $\Psi_{F_{\gamma,s}}(s)$ in \eqref{eq:bounds_psi_s}, we obtain that, $\ell_{F_{\gamma,s}}(s) \geq \ell_{F}(T).$ We then have that,
\begin{equation*}
    \ell_{F}(T) \stackrel{(a)}{=} \ell_{F_{\gamma,s}}(T) \stackrel{(b)}{\geq} \ell_{F_{\gamma,s}}(s) \geq \ell_{F}(T),
\end{equation*}
where $(a)$ has been established at the beginning of the proof, and $(b)$ follows from the monotonicity of $\ell_{F_{\gamma,s}}$ and because $T > s$. Hence, $\ell_{F_{\gamma,s}}$ is constant on $[s,T]$ equal to $\ell_{F_{\gamma,s}}(T)$.

\textit{Step 2:} 
Consider the following threshold:
\begin{equation}
\label{eq:def_T1}
    T_1 = \inf \{x \leq s \text{ s.t. } \ell_{F_{\gamma,s}}(x) = \ell_{F_{\gamma,s}}(s) \}. 
\end{equation}
Let us prove that for every $x < T_1$, we have that $\ell_{F_{\gamma,s}}(x) = \ell_{\gamma F}(x)$, where $\ell_{\gamma F}(x) = \inf \partial \Psi_{\gamma F}(x)$ and, for every $x$, we define
\begin{subequations}\label{eq:gammaF}
\begin{alignat}{2}
\Psi_{\gamma F}(x) = \; &\!\sup_{\alpha,\beta \in \mathbb{R}} &\;& \alpha +  \beta \cdot \gamma \cdot F(x) \\
&\text{s.t.} &      &  \alpha + \beta \cdot \gamma \cdot F(y) \leq \gamma \cdot H_{F}(y) - (1-\gamma) \cdot y \quad \forall y \leq s 
\end{alignat}
\end{subequations}
We  show in  \Cref{lem:relaxing_Psi_F_s} that for every $x \in [a,T_1)$, we have that $\Psi_{\gamma F}(x) = \Psi_{F_{\gamma,s}}(x)$. Furthermore, the feasible set of problem \eqref{eq:F_gamma_after_s} is included in the one of problem \eqref{eq:gammaF}, and both problems share the same objective function. Therefore, any optimal solution of \eqref{eq:F_gamma_after_s} is optimal for \eqref{eq:gammaF} which implies that $\partial \Psi_{F_{\gamma,s}}(x) \subset \partial \Psi_{\gamma F}(x)$ for all $x \in [a,T_1)$. In particular, we have $\ell_{F_{\gamma,s}}(x) \in \partial \Psi_{\gamma F}(x)$ for all $x \in [a,T_1)$. 

Moreover, for every $x \in [a,T_1)$, $F_{\gamma,s}(x) = \gamma F(x)$, and the distribution $\gamma F$ has a positive density $\gamma f$. Hence, \Cref{lem:continuity} implies that $\ell_{F_{\gamma,s}}$ and $\ell_{\gamma F}$ are continuous on $[a,T_1)$. We conclude from \Cref{lem:inclusion_to_eq} that for every $x \in [a,T_1)$, we have that $\ell_{F_{\gamma,s}}(x) = \ell_{\gamma F}(x)$, and \Cref{prop:Myerson_and_Monteiro} implies that, $\ell_{F_{\gamma,s}}(x) = \mathrm{IRON}_{[a,s]}[\gamma F](x)$ for every $x \in [a,T_1)$.
 
\textit{Step 3:} To complete our characterization we show that the threshold defined in \eqref{eq:def_T1} satisfies $T_1 = s$. 

\Cref{lem:continuity}  implies that $\ell_{F_{\gamma,s}}$ is continuous on $[a,s)$, as $F_{\gamma,s}$ has a positive density on $[a,s)$.
Next, let us prove that $\lim_{x \uparrow T_1} \ell_{F_{\gamma,s}}(x) < \lim_{x \downarrow T_1} \ell_{F_{\gamma,s}}(x)$.

For every $x < T_1$, we showed in step 2 that $\ell_{F_{\gamma,s}}(x) = \mathrm{IRON}_{[a,s]}[\gamma F](x)$.  \Cref{lem:iron_lower_virtual} implies that for every $x \in [a,T_1)$, we have that $\mathrm{IRON}_{[a,s]}[\gamma F](x) \leq \sup_{v \in [a,s]} \varphi_{\gamma F}(v)$. Furthermore, we note that $\varphi_{\gamma F}(v) = v - \frac{1/\gamma-F(v)}{f(v)}$. Hence, for every $v \in [a,s)$,
\begin{equation*}
    \varphi_{\gamma F}(v) = v - \frac{1/\gamma-F(v)}{f(v)} = v - \frac{1-F(v)}{f(v)} - \frac{\frac{1}{\gamma}-1}{f(v)}\stackrel{(a)}{=} \ell_{F}(v) -  \frac{\frac{1}{\gamma}-1}{f(v)} < \ell_{F}(v),
\end{equation*}
where $(a)$ follows from the regularity of $F$. Therefore,
\begin{equation*}
\lim_{x \uparrow T_1} \ell_{F_{\gamma,s}}(x) \leq \sup_{v \in [0,s]} \varphi_{\gamma F}(v) <  \sup_{v \in [0,s]} \ell_{F}(v) = \ell_{F}(s),
\end{equation*}
where the last equality holds because $\ell_F$ is non-decreasing.

On the other hand, we have by definition of $T_1$ (see \eqref{eq:def_T1}) that $\lim_{x \downarrow T_1} \ell_{F_{\gamma,s}}(x) = \ell_{F_{\gamma,s}}(s) = \ell_{F}(T)$. Given that $s \leq T$ and $\ell_F$ is non-decreasing, we conclude that $\lim_{x \uparrow T_1} \ell_{F_{\gamma,s}}(x) < \ell_{F}(s) \leq \ell_{F}(T) = \lim_{x \downarrow T_1} \ell_{F_{\gamma,s}}(x)$. This implies that $\ell_{F_{\gamma,s}}$ is not continuous at $T_1$ and given that $\ell_{F_{\gamma,s}}$ is continuous on $[a,s)$ and $T_1 \leq s$, we conclude that $T_1 =s$.

Finally, we note that the structure of the revenue-maximizing auction follows from property $(vii)$ in \Cref{prop:monteiro}.
\end{proof}

\subsection{Auxiliary Results and Proofs}

\begin{lemma}
\label{lem:relaxing_Psi_F_s}
Let $T_1$ be as defined in \eqref{eq:def_T1}. Then, for every $x < T_1$, $\Psi_{\gamma F}(x) = \Psi_{F_{\gamma,s}}(x)$.   
\end{lemma}

\begin{proof}[\textbf{Proof of \Cref{lem:relaxing_Psi_F_s}}]
First, remark that $\Psi_{\gamma F}(x) \geq \Psi_{F_{\gamma,s}}(x)$, as \eqref{eq:gammaF} is a relaxation of \eqref{eq:F_gamma_after_s}, in which we removed the constraints \eqref{eq:constraint_post_s}.

Let us assume for the sake of contradiction that $\Psi_{\gamma F}(x) > \Psi_{F_{\gamma,s}}(x)$. 

Let $(\hat{\alpha}(x),\hat{\beta}(x))$ (resp. $(\alpha'_{\gamma F}(x),\beta'_{\gamma F}(x))$) be an optimal solution for Problem \eqref{eq:F_gamma_after_s} (resp.  \eqref{eq:gammaF}). We note that an optimal solution is achieved for each problem as we are optimizing linear functions and the value of the problem is finite.

\textit{Step 1:} We will show that there exists $\tilde{y} \geq s$ such that,
\begin{equation}
\label{eq:constraint_achieved}
    \hat{\alpha}(x) + \hat{\beta}(x) \cdot \gamma \cdot F(\tilde{y}) + (1-\gamma) \cdot  \hat{\beta}(x) = \gamma \cdot H_{F}(\tilde{y}).
\end{equation}
Assume for the sake of contradiction that there does not exist any $\tilde{y} \geq s$ such that \eqref{eq:constraint_achieved} holds and consider,
\begin{align*}
        S_1 =& \sup_{y \in [s,b]} \hat{\alpha}(x) + \hat{\beta}(x) \cdot \gamma \cdot F(y) + (1-\gamma) \cdot  \hat{\beta}(x) - \gamma \cdot H_{F}(y),\\
        S_2 =& \sup_{y \in [s,b]} \alpha'_{\gamma F}(x) + \beta'_{\gamma F}(x) \cdot \gamma \cdot F(y) + (1-\gamma) \cdot  \beta'_{\gamma F}(x) - \gamma \cdot H_{F}(y).
\end{align*}
Note that both $S_1$ and $S_2$ are finite and achieved as the functions are continuous on a compact segment. Furthermore, $S_1 < 0$, otherwise \eqref{eq:constraint_achieved} would hold.

For every $\lambda \in [0,1]$, let $(\alpha_\lambda(x),\beta_\lambda(x)) = \lambda \cdot (\hat{\alpha}(x),\hat{\beta}(x)) + (1-\lambda) \cdot (\alpha'_{\gamma F}(x),\beta'_{\gamma F}(x)).$ We will choose $\lambda \in (0,1)$ such that $(\alpha_\lambda(x),\beta_\lambda(x))$ is feasible for Problem \eqref{eq:F_gamma_after_s} and achieves an objective strictly greater than the one achieved by $(\hat{\alpha}(x),\hat{\beta}(x))$. 

Note that for every $\lambda \in [0,1]$ we have that the constraints \eqref{eq:constraint_pre_s} are satisfied for all $y < s$ as both $(\hat{\alpha}(x),\hat{\beta}(x))$ and $(\alpha'_{\gamma F}(x),\beta'_{\gamma F}(x))$ satisfy these constraints and any convex combination of feasible solution is still feasible for these constraints.
Furthermore, $(\alpha_\lambda(x),\beta_\lambda(x))$ satisfies the constraints \eqref{eq:constraint_post_s} for all $y \geq s$ if and only if,
\begin{equation*}
    \sup_{y \in [s,b]} \alpha_\lambda(x) + \beta_\lambda(x) \cdot \gamma \cdot F(y) + (1-\gamma) \cdot  \beta_\lambda(x) - \gamma \cdot H_{F}(y) \leq 0.
\end{equation*}
By construction of $(\alpha_\lambda(x),\beta_\lambda(x))$, we have that
\begin{equation*}
    \sup_{y \in [s,b]} \alpha_\lambda(x) + \beta_\lambda(x) \cdot \gamma \cdot F(y) + (1-\gamma) \cdot  \beta_\lambda(x) - \gamma \cdot H_{F}(y) \leq \lambda \cdot S_1 + (1-\lambda) \cdot S_2.
\end{equation*}
Given that $S_1 < 0$, there exists $\lambda > 0$ such that $\lambda \cdot S_1 + (1-\lambda) \cdot S_2 \leq 0.$ In what follows we fix such $\lambda$. We obtain that $(\alpha_\lambda(x),\beta_\lambda(x))$ is feasible for the constraints \eqref{eq:constraint_post_s} for all $y \geq s$ which implies that $(\alpha_\lambda(x),\beta_\lambda(x))$ is feasible for \eqref{eq:F_gamma_after_s}. 

Finally, we note that the objective obtained with the solution $(\alpha_\lambda(x),\beta_\lambda(x))$ satisfies,
\begin{align*}
    \alpha_\lambda(x) + \beta_\lambda(x) \cdot \gamma \cdot F(x) 
    = \lambda \Psi_{F_{\gamma,s}}(x) + (1-\lambda) \Psi_{\gamma F}(x)
    \stackrel{(a)}{>} \Psi_{F_{\gamma,s}}(x) = \hat{\alpha}(x) + \hat{\beta}(x) \cdot \gamma \cdot F(x),
\end{align*}
where $(a)$ holds because $\Psi_{F_{\gamma,s}}(x) < \Psi_{\gamma F}(x)$. This contradicts the optimality of $(\hat{\alpha}(x),\hat{\beta}(x))$ for Problem \eqref{eq:F_gamma_after_s}. 

Therefore, $(\hat{\alpha}(x),\hat{\beta}(x))$ must satisfy \eqref{eq:constraint_achieved} for some $\tilde{y} \geq s$. 

\textit{Step 2:} We next show that for every $z \in (x,s]$, we have that  $\ell_{F_{\gamma,s}}(z) = \ell_{F_{\gamma,s}}(s)$. 
Let $\tilde{y} \geq s$ be such that $(\hat{\alpha}(x),\hat{\beta}(x))$ satisfies \eqref{eq:constraint_achieved}. We next show that this implies that $\hat{\beta} (x) \in \partial \Psi_{F_{\gamma,s}}(\tilde{y})$. 
Indeed, $(\hat{\alpha}(x),\hat{\beta}(x))$ is feasible for Problem \eqref{eq:F_gamma_after_s} at $\tilde{y}$ and it is optimal, because $\Psi_{F_{\gamma,s}}(\tilde{y}) \leq \gamma H_{F}(\tilde{y})$ and \eqref{eq:constraint_achieved} implies that $(\hat{\alpha}(x),\hat{\beta}(x))$ achieves this value. Hence, $\beta^*(x) \in \partial \Psi_{F_{\gamma,s}}(\tilde{y})$. Consequently, we have that $\partial \Psi_{F_{\gamma,s}}(\tilde{y}) \cap \partial \Psi_{F_{\gamma,s}}(x) \neq \emptyset$.

We conclude that for every $x' \in (x,s]$ we have that
\begin{equation*}
    \ell_{F_{\gamma,s}}(s) \stackrel{(a)}{\leq}  \ell_{F_{\gamma,s}}(\tilde{y}) \stackrel{(b)}{=} \ell_{F_{\gamma,s}}(x') \stackrel{(a)}{\leq} \ell_{F_{\gamma,s}}(s),
\end{equation*}
where $(a)$ follows from the monotonicity of the $\ell_{F_{\gamma,s}}$, and the fact that $x \leq s \leq \tilde{y}.$ Furthermore $(b)$ holds by \Cref{lem:disjoint_subg}.

Therefore, for every $x' \in (x,s]$, we have that $\ell_{F_{\gamma,s}}(x') = \ell_{F_{\gamma,s}}(s)$.

In particular, there exists $z \in (x,T_1)$ such that $\ell_{F_{\gamma,s}}(z) = \ell_{F_{\gamma,s}}(s)$. This contradicts the definition of $T_1$. Therefore, $\Psi_{\gamma F}(x) = \Psi_{F_{\gamma,s}}(x)$ for every $x \in [0,T_1)$.

\end{proof}

\begin{lemma}
\label{lem:F_and_H}
For every distribution $F$, any $\gamma \in (0,1)$ and any $s$ in the support of $F$, we have for every $x$ that,
\begin{equation*}
F_{\gamma,s}(x) = \begin{cases}
\gamma \cdot F(x) \quad \text{if $x < s$}\\
\gamma \cdot F(x) + (1-\gamma) \quad \text{if $x \geq s$}
\end{cases} 
\quad 
\mbox{and} 
\quad 
H_{F_{\gamma,s}}(x) = \begin{cases}
\gamma \cdot H_{F}(x) - (1-\gamma) \cdot x  \quad \text{if $x < s$}\\
\gamma \cdot H_{F}(x) \quad \text{if $x \geq s$}.
\end{cases}
\end{equation*}
\end{lemma}

\begin{proof}[\textbf{Proof of \Cref{lem:F_and_H}}]
We characterized $F_{\gamma,s}$ in \eqref{eq:cumulative-F}. For $x < s$, we have that,
\begin{equation*}
    H_{F_{\gamma,s}}(x) = \int_a^x t d F_{\gamma,s}(t) - \int_a^x (1-F_{\gamma,s}(t))dt - a = \gamma \int_a^x t d F(t) - \int_a^x (1-\gamma F(t))dt - a = \gamma H_{F}(x) - (1-\gamma) \cdot x. 
\end{equation*}
Moreover, for $x \geq s$, we have that,
\begin{align*}
    H_{F_{\gamma,s}}(x) &= \int_a^x t d F_{\gamma,s}(t) - \int_a^x (1-F_{\gamma,s}(t))dt - a\\
    &= (1-\gamma) \cdot s + \gamma \int_a^x t d F(t) - \int_a^x (1-\gamma F(t)) dt  + (1-\gamma) \cdot (x-s) - a\\
    &= \gamma \cdot H_{F}(x).
\end{align*}

\end{proof}

\begin{lemma}
\label{lem:charac_PsiF}
For every $x$, let $(\aF,\bF)$ be an optimal solution of \eqref{eq:gen_virtual_value}, then
\begin{equation*}	
\aF + \bF \cdot F(x) = H_{F}(x).
\end{equation*}
\end{lemma}
\begin{proof}[Proof of \Cref{lem:charac_PsiF}]
We note that by optimality of $(\aF,\bF)$, we have that $\Psi(x) = \aF + \bF \cdot F(x).$ Furthermore as $F$ has a positive density it is strictly increasing, and its inverse function $F^{-1}$ is well-defined.
The function defined for every $z \in [0,1]$ as $z \mapsto \Psi(F^{-1}(z))$ then corresponds to the convex envelope of $z \mapsto H_{F}(F^{-1}(z))$. The latter is convex because $F$ is regular.
Therefore, $\Psi(F^{-1}(z)) = H_{F}(F^{-1}(z))$ for every $z \in [0,1]$. By evaluating this equality for $z = F(x)$, we obtain that $\aF + \bF \cdot F(x) = \Psi(x) = H_{F}(x)$.
\end{proof}

\begin{lemma}	
\label{lem:feasible_post_T}
Let $T$ such that, $\mu_s(T) = 0$. Then, for every $x \geq T$, the vector $(\aFs,\bFs)$ is feasible for Problem \eqref{eq:F_gamma_after_s}.
\end{lemma}
\begin{proof}[\textbf{Proof of \Cref{lem:feasible_post_T}}]

Let $x \geq T$ and let $\bF = \ell_F(x)$. By \eqref{eq:subgrad_are_solutions}, there exists $\aF$ such that $(\aF,\bF)$ is optimal for Problem~\eqref{eq:gen_virtual_value} at the point $x$. We next show that the candidate solution $(\aFs,\bFs)$ defined in \eqref{eq:candidate} is feasible for Problem~\eqref{eq:F_gamma_after_s}.

Let $y \geq s$. We note that,
\begin{align*}
\aFs + (1-\gamma) \cdot \bFs + \bFs \cdot \gamma \cdot F(y) \leq \gamma \cdot H_{F}(y) &\iff \gamma \cdot \left(  \aF + \bF \cdot F(y) \right) \leq \gamma H_{F}(y) \\
&\iff \aF + \bF \cdot F(y)  \leq H_{F}(y).
\end{align*}
The last inequality holds because $(\aF,\bF)$ is feasible for \eqref{eq:gen_virtual_value}. Hence, the first inequality holds. Which implies that $(\aFs,\bFs)$ satisfies the constraint \eqref{eq:constraint_post_s}.

Furthermore, fix $y \leq s$. By definition of $\mu_y$, we have that $(\aFs,\bFs)$ satisfies \eqref{eq:constraint_pre_s} if and only if, $\mu_y(x) \leq 0.$ This inequality holds because,
\begin{equation*}
\mu_y(x) \stackrel{(a)}{\leq} \mu_s(x) \stackrel{(b)}{\leq} \mu_s(T) = 0,
\end{equation*}
where $(a)$ follows from property $(ii)$ in \Cref{lem:prop_mu} and $(b)$ from property $(i)$ in \Cref{lem:prop_mu}.
Therefore, $(\aFs,\bFs)$ is feasible for Problem \eqref{eq:F_gamma_after_s}.

\end{proof}

\begin{lemma}
\label{lem:disjoint_subg}
Let $x < y$, assume that $\partial \Psi_{F}(x) \cap \partial \Psi_{F}(y) \neq \emptyset$, then for every $x' \in (x,y]$, we have that $\ell_{F}(x') = \ell_{F}(y)$.
\end{lemma}
\begin{proof}[\textbf{Proof of \Cref{lem:disjoint_subg}}]
Let $u \in \partial \Psi_{F}(x) \cap \partial \Psi_{F}(y).$ Let $x' \in (x,y)$ and assume for sake of contradiction that $\ell_{F}(x') < \ell_{F}(y)$. 
We have that, $s_{F}(x') \geq s_{F}(x) \geq u \geq \ell_{F}(y) > \ell_{F}(x') \geq \ell_{F}(x)$.
Hence, $(\ell_{F}(x), s_{F}(x)) \cap (\ell_{F}(x'), s_{F}(x')) \neq \emptyset$. This contradicts item $(ii)$ \Cref{prop:monteiro}. Therefore, for every $x' \in (x,y]$, we have that $\ell_{F}(x') = \ell_{F}(y)$.
\end{proof}

\begin{lemma}
    \label{lem:iron_lower_virtual}
    Let $F$ be a distribution with positive and continuous density and let $\varphi_F$ be the virtual function. Then, for every $x,t \in [a,b]$ such that $x \leq t$, we have that
    \begin{equation*}
    \mathrm{IRON}_{[a,t]}[F](x) \leq \sup_{v \in [a,t]} \varphi_{F}(v). 
    \end{equation*}
\end{lemma}

\begin{proof}[\textbf{Proof of \Cref{lem:iron_lower_virtual}}]
    Let $t \in [a,b]$.
    Note that by construction $\mathrm{IRON}_{[a,t]}[F]$ is a non-decreasing function. Therefore, it is sufficient to prove that, 
    \begin{equation}
    \label{eq:ineq_to_prove}
    \mathrm{IRON}_{[a,t]}[F](t) \leq \sup_{v \in [a,t]} \varphi_{F}(v). 
    \end{equation}
    Assume for the sake of contradiction that this inequality does not hold.

    Recall the definition of $J$ (see \eqref{eq:J}) and of $G_{F(t)}$, the convex hull of the restriction of $J$ on $[0,F(t)]$.
    By definition, we have that
    \begin{equation*}
        G_{F(t)}(F(t)) = \min_{ \substack{(\lambda,r_1,r_2) \in [0,1]\times[0,F(t)]^2\\ \text{s.t. } \lambda \cdot r_1 + (1-\lambda) \cdot r_2 = F(t)} } \lambda \cdot J(r_1) + (1-\lambda) \cdot J(r_2) = J(F(t)),
    \end{equation*}
    where the last equality holds because $F(t)$ is an extreme point of $[0,F(t)]$, hence $r_1$ must equal $r_2$. Similarly we can show that $G_{F(t)}(0) = J(0).$ Let $u = \sup \{ x < t \text{ s.t. } G_{F(t)}(F(x)) = J(F(x)) \}$ Note that $u$ exists and is finite as the set is non-empty (it includes $0$) and bounded. We reason by disjunction of cases on the value of $u$.

    \textit{Case 1: $u = t$.} \Cref{prop:Myerson_and_Monteiro} and \Cref{lem:continuity} imply that $\mathrm{IRON}_{[a,t]}[F]$ is continuous as $F$ has a positive and continuous density. By assumption we have that $\mathrm{IRON}_{[a,t]}[F](t) > \sup_{v \in [a,t]} \varphi_{F}(v)$. The continuity of $\mathrm{IRON}_{[a,t]}[F]$  implies that there exists $\epsilon$ such that for every $t' \in [t-\epsilon,t]$ we have that, $\mathrm{IRON}_{[a,t]}[F](t') > \sup_{v \in [a,t]} \varphi_{F}(v)$. Moreover, as $u = t$, for $\epsilon$ small enough we also have that $G_{F(t)}(F(t-\epsilon)) = J(F(t-\epsilon))$. Fix such $\epsilon$ and observe that,
    \begin{align*}
        \int_{F(t-\epsilon)}^{F(t)} \varphi_F(F^{-1}(r)) dr &\stackrel{(a)}{=} J(F(t)) -  J(F(t - \epsilon))\\
        &= G_{F(t)}(F(t)) - G_{F(t)} (F(t-\epsilon)) \stackrel{(b)}{=} \int_{F(t-\epsilon)}^{F(t)}  \mathrm{IRON}_{[a,t]}[F](F^{-1}(r)) dr,
    \end{align*}
    where $(a)$ follows from \eqref{eq:J} and $(b)$ holds because $F$ is a continuous and increasing (as it admits a positive density everywhere) and hence it is invertible with inverse $F^{-1}.$
    
    Hence, we have established that 
    \begin{equation*}
        \int_{F(t-\epsilon)}^{F(t)} (\varphi_F(F^{-1}(r)) -\mathrm{IRON}_{[a,t]}[F](F^{-1}(r)))  dr =0,
    \end{equation*}
    which contradicts the fact that $\mathrm{IRON}_{[a,t]}[F](t') > \sup_{v \in [a,t]} \varphi_{F}(v)$ for every $t' \in [t-\epsilon,t]$.

    \textit{Case 2: $u < t$.}
    In that case $G_{F(t)}$ has a constant differential on $(u,t]$. Hence, we obtain that
    \begin{align*}
       \int_{F(u)}^{F(t)} \varphi_F(F^{-1}(r)) dr &= J(F(t)) -  J(F(u))\\
       &=  G_{F(t)}(F(t)) - G_{F(t)}(F(u))\\
       &=  \int_{F(u)}^{F(t)}  \mathrm{IRON}_{[a,t]}[F](F^{-1}(r)) dr > (F(t)-F(u)) \sup_{v \in [a,t]} \varphi_{F}(v).
    \end{align*}
    This leads to a contradiction.

\end{proof}

\begin{proposition}[\cite{monteiro2010optimal}]
\label{prop:monteiro}
Let $F$ be a distribution. We have that:
\begin{enumerate}
\item[i.] $ \{ x  \, \text{s.t.} \,  |\partial \Psi_{F}(x)| > 1 \}$ is at most countable.   
\item[ii.] $\ell_F(\cdot)$ and $s_F(\cdot)$ are non-decreasing and $(\ell(x),s(x)) \cap (\ell(x'),s(x')) = \emptyset$ for all $x \neq x'$.
\item [iii.] If $F(x) - F(x-) >0$, we have that $\ell_{F}(x) = \frac{\Psi_{F}(x)-\Psi_{F}(x-)}{F(x)-F(x-)}$.
\item [iv.] If $\Psi_F(x-) < H_F(x-)$, then there exists an interval $[x,z^*)$ such that $s(z) = \ell(z) = s(x-)$ for every $z \in (x,z^*)$.
\item [v.] If $\Psi_F(x) < H_F(x)$, then there exists an interval $[x,z^*)$ such that $\ell(z) = \ell(x)$ for every $z \in (x,z^*)$.
\item [vi.] For every $x$, $\Psi_{F}$ is continuous at $x$ if $F$ is continuous at $x$.
\item [vii.] Let $(F_i)_{i \in \{1,\ldots,n\}}$ be the value distributions of the buyers, and $(\hat{v}_i)_{i \in \{1,\ldots,n\}}$ their reported values. The following auction is revenue-maximizing. Allocate to the buyer with the highest non-negative value of $\ell_{F_i}(\hat{v}_i)$, and make them pay  $\ell_{F_i}^{-1} \left( \max \{0, \max_{j \neq i} \ell_{F_j}(\hat{v}_j)\} \right)$
\end{enumerate}
\end{proposition}

\begin{proof}[\textbf{Proof of \Cref{prop:monteiro}}]
We next point to the results in \cite{monteiro2010optimal} implying each of the points in the proposition.
    $(i)$ follows from Remark 3, $(ii)$ follows from Proposition 1.f and Remark 3, $(iii)$ follows from Proposition 4, $(iv)$ and $(v)$ are established in Proposition 5, $vi$ follows from Proposition 3 and $(vii)$ is established in Theorem 5.
\end{proof}

\section{Additional Results and Proofs of Results in \Cref{sec:single-buyer}} \label{sec:apx_single_buyer}

\subsection{Proofs of Results in \Cref{sec:opt_1b}}

\Cref{prop:optimal_price_corrected} follows from the following result by remarking that if $a = 0$, then we have that, $\varphi_{\gamma F}(a) = \varphi_{\gamma F}(0) = -\frac{1}{\gamma f(0)} < 0.$
\begin{proposition}
\label{prop:optimal_price_corrected_general}
Let $n=1$ and assume $F$ has a log-concave density $f$ on $[a,b]$. Furthermore, assume that $f$ is twice continuously differentiable.
Define\footnote{We use the convention that the infimum (resp. supremum) of an empty subset of $[a,b]$ is $b$ (resp. $a$).} $C_\gamma = \sup \{v \in [a,b] : \varphi_{\gamma F}(v) < 0\}$ and, 
\begin{equation*}
p^{\mathrm{ignore}} = \inf\{v \in [a,b] : \varphi_F(v) \geq 0\}
\end{equation*}
and,
\begin{equation*}
    p^{\mathrm{cap}} =
\begin{cases}
\inf\{v \in [a,b] : \varphi_{\gamma F}(v) \geq 0\} \; \; &\text{if $\varphi_{\gamma F}(a) < 0$,}\\
C_\gamma &\text{if $\varphi_{\gamma F}(a) \geq 0$ and } a\cdot(1-\gamma F(a))  \leq C_\gamma \cdot(1-\gamma F(C_\gamma)),\\
a \quad &\text{otherwise}.
\end{cases}
\end{equation*}
Then, the following posted price is optimal:
\begin{equation*}
    p^*(s)=
\begin{cases}
p^{\mathrm{ignore}}, 
    & s < L_\gamma, \\
s, 
    & L_\gamma \le s < M_\gamma, \\
p^{\mathrm{cap}},
    & M_\gamma \le s \leq U_\gamma, \\
s,
    & s > U_\gamma,
\end{cases}
\end{equation*}
where
\begin{equation*}
   L_\gamma = \begin{cases}
       \inf\{s \in [a,b] : T_s \ge p^{\mathrm{ignore}}\} &\text{if $\varphi_{\gamma F}(a) < 0$,}\\
       \min \left( \inf \{v \in [a,b] \text{ s.t. }   v \cdot (1- \gamma F(v)) > a \cdot (1 - \gamma F(a))\}, M_\gamma \right) &\text{otherwise},
   \end{cases}
\end{equation*}
and,
\begin{align*}
    M_\gamma = p^{\mathrm{cap}}, \quad U_\gamma = \sup \{s \in [a,b] : s \cdot (1-\gamma F(s)) \leq M_\gamma \cdot (1-\gamma F(M_\gamma)) \},
\end{align*}
and $T_s$ is the threshold in \Cref{thm:main}.
\end{proposition}

\begin{proof}[\textbf{Proof of \Cref{prop:optimal_price_corrected_general}}]
For any mapping $u$ from $[a,b]$ to $\mathbb{R}$ we define for every $x \in [a,b]$, the notation $u(x-) = \lim_{\substack{t \to x\\t < x}} u(t)$ to denote the left-limit of $u$ at $x$ (when it exists).

Note that log-concavity implies that $F$ is regular with a continuous density. Therefore the characterization in \Cref{thm:main} holds.
By \Cref{thm:main}, with one buyer the optimal price $p^*(s)$ is the smallest $v$ such that $\bar{\varphi}_{F_{\gamma,s}}(v)\ge 0$. Alternatively, this characterization implies that $p^{*}(s)$ is a maximizer of the mapping $p \mapsto p \cdot (1- F_{\gamma,s}(p-))$.

Recall that the structure of $\bar{\varphi}_{F_{\gamma,s}}$ is
\[
\bar{\varphi}_{F_{\gamma,s}}(v)=
\begin{cases}
\mathrm{IRON}_{[a,s]}[\gamma F](v), & v < s,\\[2pt]
\varphi_F(T_s),                    & s \le v < T_s,\\[2pt]
\varphi_F(v),                      & v \ge T_s,
\end{cases}
\]
with $T_s$ non-decreasing in $s$ (see \Cref{lem:monotonic_T_s}).

\noindent \textit{Case 1: $\varphi_{\gamma F}(a) < 0.$}~\\
If $s < L_\gamma$, then $T_s < p^{\mathrm{ignore}}$. By monotonicity of  $\bar{\varphi}_{F_{\gamma,s}}$, we obtain that for every $v < T_s$,
$\bar{\varphi}_{F_{\gamma,s}}(v) \leq \bar{\varphi}_{F_{\gamma,s}}(T_s) = \varphi_{F}(T_s) < 0$, where the last inequality holds because $\varphi_{F}$ is non-decreasing and $T_s < p^{\mathrm{ignore}} = \inf\{v \in [a,b] : \varphi_F(v) \geq 0\}$.
Furthermore,  
$\bar{\varphi}_{F_{\gamma,s}}(v)=\varphi_F(v)$ for every $v\ge T_s$.
The smallest non-negative $v$ of $\varphi_F$ is $p^{\mathrm{ignore}}$, hence
$p^*(s)=p^{\mathrm{ignore}}$.

In what follows, we assume that $L_\gamma \leq s$. This implies that $T_s \geq p^{\mathrm{ignore}}$ and, $\bar{\varphi}_{F_{\gamma,s}}(s) = \varphi_F(T_s)\ge \varphi_F(p^{\mathrm{ignore}}) \geq 0$. Thus, we have $p^*(s) \leq s$.

Equivalently, $p^*(s) \in \argmax_{p \in [a,s]} p \cdot (1-F_{\gamma,s}(p-))$. Note that for every $p \in [a,s]$, 
\begin{equation*}
    p \cdot (1-F_{\gamma,s}(p-)) = p \cdot (1-\gamma \cdot F(p-)) = p \cdot (1-\gamma \cdot F(p)),
\end{equation*}
where the first equality follows from the definition of $F_{\gamma,s}$ and the second equality holds because $F$ has a continuous density.

By continuity of $F$, the mapping $u: p \mapsto p \cdot (1-\gamma F(p))$ is continuous, hence it achieves its maximum on the compact set $[a,s]$. We next analyze the behavior of $u$ on $[a,s]$. 

For every $x \in [a,b]$ we have,
\begin{equation*}
    u(x)= u(a) - \int_a^{\gamma F(x)} \varphi_{\gamma F} \left( (\gamma F)^{-1}(t) \right)dt.
\end{equation*}
Consequently, for every $x \in [a,b]$, $u'(x)$ has the same sign as $-\varphi_{\gamma F}(x)$. We next analyze the variations of $u$.

As $f$ is log-concave, \Cref{lem:two-zeros} implies $\varphi_{\gamma F}$ admits at most two distinct $0$ on $[a,b]$. 
By continuity of $\varphi_{\gamma F}$,  we thus have that $\{ v \in [a,b] \text{ s.t } \varphi_{\gamma F}(v) \geq 0\}$ is a closed interval. 
If it is empty, then $\varphi_{\gamma F}(v)<0$ for all $v\in[a,b]$, hence $u'(v)>0$ for all $v\in[a,b]$, and therefore $\max_{p\in[a,s]}u(p)=u(s)$ so $p^*(s)=s$ for every $s \geq L_\gamma$. In that case, we have by convention that $M_\gamma = p^{\mathrm{cap}} = b$ and, by definition, $U_\gamma = b$. Hence, we obtain the desired characterization where only the first two pieces are non-empty.

Next, assume that $\{ v \in [a,b] \text{ s.t } \varphi_{\gamma F}(v) \geq 0\}$ is not empty. 
Let $p^{\mathrm{cap}} = M_\gamma = \inf \{ v \in [a,b] \text{ s.t } \varphi_{\gamma F}(v) \geq 0\}$ and $m_{\gamma}$ be the supremum of such a set.
We have that $u'(x) \leq 0$ if and only if $x \in [M_\gamma,m_\gamma]$. Consequently, $u$ is non-decreasing on $[a,M_{\gamma}]$ achieves a local maximum at $M_\gamma$, it is then non-increasing on $[M_{\gamma},m_{\gamma}]$, achieves a local minimum at $m_{\gamma}$ and is again non-decreasing on $[m_{\gamma},b]$. 
Furthermore, let $K = \{ x \in [m_{\gamma},b] \text{ s.t. } u(x) \leq u(M_\gamma) \}$. Remark that $K = u^{-1}([u(M_\gamma),\infty)) \cap [m_\gamma,b]$, and $u$ is a continuous non-decreasing mapping on $[m_\gamma,b]$. Therefore, the intermediate value theorem implies that $K$ is empty or $K = [m_\gamma, U_\gamma]$, where $\sup K$. We summarize these findings in \Cref{fig:variation}.

\begin{figure}[h!]
    \begin{center}
\begin{tikzpicture}
  \tkzTabInit[lgt=2.2,espcl=4]%
    {$x$/1, $u'(x)$/1, $u$/2}%
    {$a$, $M_\gamma$, $m_\gamma$, $b$}
  \tkzTabLine{,+,z,-,z,+}
  \tkzTabVar{-/$u(a)$, +/$u(M_\gamma)$, -/$u(m_\gamma)$, +/$u(b)$}
  \tkzTabVal{3}{4}{0.5}{$U_\gamma$}{$u(M_\gamma)$}
\end{tikzpicture}
\end{center}
    \caption{Variation of $u$ in sub-case $(a)$}
    \label{fig:variation}
\end{figure}

\Cref{fig:variation} allows us to conclude the proof for this sub-case.
\begin{itemize}
    \item If $s < M_\gamma$, we have that $\max_{p \in [a,s]} u(p) = u(s)$ hence $p^*(s) = s$ is optimal.
    \item If $s \in [M_\gamma,U_\gamma),$ we have that $\max_{p \in [a,s]} u(p) = u(M_\gamma)$ hence $p^*(s) = M_
    \gamma$ is optimal.
    \item If $s \geq U_\gamma$ we have that $\max_{p \in [a,s]} u(p) = u(s)$ hence $p^*(s) = s$ is optimal.
\end{itemize}

\noindent \textit{Case 2: $\varphi_{\gamma F}(a) \geq 0.$}~\\
Recall that for every $v \in [a,b]$, we have that $\varphi_{F}(v) \geq \varphi_{\gamma F}(v)$, and $\varphi_{F}$ is non-decreasing which implies that, for every $v \in [a,b]$, $\varphi_{F}(v) \geq 0$. Hence,  $\bar{\varphi}_{F_{\gamma,s}}(s) = \varphi_F(T_s) \geq 0$. Thus, we have $p^*(s) \leq s$. As in case 1, it is sufficient to maximize the function $u$ on $[a,s]$.

First, if $\varphi_{\gamma F}(v) \geq 0 $ for all $v \in [a,b]$, then $u'(v) \leq 0$ for all $v \in [a,b]$ and $u$ is non-increasing. Hence, for every $s \in [a,b]$, $p^*(s) = a$ maximizes $u$.
In that case, remark that  $\varphi_{F}(a) \geq \varphi_{\gamma F}(a) \geq 0$, so $p^{\mathrm{ignore}}= a$.
Furthermore, the set $\{ v \in [a,b] \text{ s.t. } \varphi_{\gamma F}(v) < 0 \}$ is  empty and by convention we have that $L_\gamma = M_\gamma = p^{\mathrm{cap}} = a$ and, as $u$ is non-increasing, $U_\gamma = b$. Hence, we obtain the desired characterization where only the third piece is non-empty.

Next, assume the set $\{ v \in [a,b] \text{ s.t. } \varphi_{\gamma F}(v) < 0 \}$ is non-empty.
Recall that $C_\gamma = \sup \{ v \in [a,b] \text{ s.t. } \varphi_{\gamma F}(v) < 0 \}$ and $m_\gamma = \inf \{ v \in [a,b] \text{ s.t. } \varphi_{\gamma F}(v) < 0 \}$ in $[a,b]$. By an analysis similar to case (1), we establish that $u$ is non-increasing on $[a,m_{\gamma}]$, then non-decreasing on $[m_\gamma,C_{\gamma}]$ and again non-increasing on $[m_\gamma,b]$. These variations are summarized in \Cref{fig:variation_2}.
\begin{figure}[h!]
    \begin{center}
\begin{tikzpicture}
  \tkzTabInit[lgt=2.2,espcl=4]%
    {$x$/1, $u'(x)$/1, $u$/2}%
    {$a$, $m_\gamma$, $C_\gamma$, $b$}
  \tkzTabLine{,-,z,+,z,-}
  \tkzTabVar{+/$u(a)$, -/$u(m_\gamma)$, +/$u(C_\gamma)$, -/$u(b)$}
\end{tikzpicture}
\end{center}
    \caption{Variation of $u$ in sub-case $(b)$}
    \label{fig:variation_2}
\end{figure}

Note that in this case, $U_\gamma = b.$
From \Cref{fig:variation_2} we remark that  if $u(C_\gamma) < u(a)$, then for every $s \in [a,b]$, we have that, $\max_{p \in [a,s]} u(p) = u(a)$ hence $p^*(s) = a = p^{\mathrm{ignore}} = L_{\gamma}$. In that case $p^{\mathrm{cap}} = M_\gamma = a$. Hence, the characterization hold with the third piece being non-empty.

Furthermore, if $u(C_\gamma) \geq u(a)$, then
\begin{itemize}
    \item If $s < L_\gamma$, we have that $\max_{p \in [a,s]} u(p) = u(a)$, by definition of $L_\gamma$, hence $p^*(s) = p^{\mathrm{ignore}}$ is optimal.
    \item If $s \in [L_\gamma,M_\gamma)$, we have that $\max_{p \in [a,s]} u(p) = u(s)$, hence $p^*(s) = s$ is optimal.
    \item If $s \in [M_\gamma,U_\gamma),$ we have that $\max_{p \in [a,s]} u(p) = u(M_\gamma)$ hence $p^*(s) = M_
    \gamma$ is optimal.
\end{itemize}
This concludes the proof.

\end{proof}

\begin{lemma}
\label{lem:monotonic_T_s}
    For every $y \leq s$, let $\mu_y$ be the mapping defined in \eqref{eq:mu_y}. And for every $s \in [a,b)$, let $T_s$ be the smallest $T \in (s,b]$ such that $\mu_{s}(T_s) = 0.$ Then, the mapping $s \mapsto T_s$ is non-decreasing.
\end{lemma}

\begin{proof}[\textbf{Proof of \Cref{lem:monotonic_T_s}}]
    Let $s < s'$. Let us prove that $T_s \leq T_{s'}$. Property (i) in \Cref{lem:prop_mu} implies that $\mu_{s'}$ is non-increasing. Therefore, it is sufficient to prove that $\mu_{s'}(T_s) \geq \mu_{s'}(T_{s'}).$

    Furthermore we have that,
    \begin{equation*}
        \mu_{s'}(T_s) > \mu_{s}(T_s) = 0 = \mu_{s'}(T_s'),
    \end{equation*}
    where the first inequality follows from $(ii)$ in \Cref{lem:prop_mu}. This concludes the proof.
\end{proof}

\begin{lemma}[At most two zeros of $\varphi$ under log-concavity]\label{lem:two-zeros}
Let $F$ be a distribution on $(a,b)$ which admits a positive density function $f$. Furthermore, assume that $f$ is $C^2$ and log-concave.
Fix $\gamma \in(0,1)$,
then the equation $\varphi_{\gamma F}(v)=0$ has at most two distinct solutions in $(a,b)$.
\end{lemma}

\begin{proof}[\textbf{Proof of \Cref{lem:two-zeros}}]
Set, for all $v \in (a,b)$
\begin{equation*}
    A(v):=\frac{1}{\gamma}-F(v).
\end{equation*}
Since $\gamma\in(0,1)$ and $F(v)\le 1$, we have $A(v)>0$ on $(a,b)$. 
By definition, for all $v \in (a,b)$
\begin{equation}\label{eq:phi-A}
v-\varphi_{\gamma F}(v)=\frac{A(v)}{f(v)} 
\end{equation}
so in particular $v-\varphi_{\gamma F}(v)>0$ on $(a,b)$.

\noindent\textit{Step 1: A structural differential inequality.}
Write $u=\log f$. Differentiating \eqref{eq:phi-A} and using $A'=-f$, we obtain
\begin{equation*}
    f'(v) \cdot (v-\varphi_{\gamma F}(v)) + f(v) \cdot (1-\varphi_{\gamma F}'(v)) \;=\; -\,f(v).
\end{equation*}
Dividing by $f(v)>0$ and recalling $u'=f'/f$ gives
\begin{equation*}
u'(v) \cdot (v-\varphi_{\gamma F}(v)) + 1 - \varphi_{\gamma F}'(v) = -1, 
\end{equation*}
which implies that,
\begin{equation*}
    u'(v) = \frac{\varphi_{\gamma F}'(v)-2}{\,v-\varphi_{\gamma F}(v)\,}.
\end{equation*}
Differentiating once more, implies that
\begin{equation*}
    u''(v)
= \frac{(v-\varphi_{\gamma F}(v)) \cdot \varphi_{\gamma F}''(v) + (\varphi_{\gamma F}'(v)-2) \cdot (\varphi_{\gamma F}'(v)-1)}{(v-\varphi_{\gamma F}(v))^2}.
\end{equation*}

Since $u''\le 0$ by log-concavity and $v-\phi(v)>0$ by \eqref{eq:phi-A}, we obtain that for all $v \in (a,b)$,
\begin{equation}\label{eq:key}
(v-\varphi_{\gamma F}(v))\cdot \varphi_{\gamma F}''(v) + \bigl(\varphi_{\gamma F}'(v)-2\bigr) \cdot \bigl(\varphi_{\gamma F}'(v)-1\bigr) \;\leq \; 0
\end{equation}

\smallskip
\noindent\textit{Step 2: $\varphi_{\gamma F}'$ has at most one zero.}
Let $g:=\varphi_{\gamma F}'$. From \eqref{eq:key} and $v-\varphi_{\gamma F}>0$ we have, for every $v \in (a,b)$ that
\begin{equation}\label{eq:concavity-bands}
g(v)\le 1 \ \text{ or }\ g(v)\ge 2 \quad\Longrightarrow\quad \varphi_{\gamma F}''(v)=g'(v)\leq 0.
\end{equation}
Furthermore, at any $z\in (a,b)$ with $g(z)=0$ we have $g'(z) < 0$.

Let $z_1 \in (a,b)$ be such that $g(z_1) = 0$ and assume for the sake of contradiction that there exists $z' > z_1$ such that $g(z') = 0$.
Since $g(z_1) = 0$, we have that $g'(z_1) < 0$ and, by continuity of $g'$, there exists $\varepsilon>0$ such that $g'(v) < 0$ for every $v \in [z_1,z_1+\varepsilon]$, hence $g$ is decreasing on $[z_1,z_1+\varepsilon]$ and thus $g(v) < 0$ for every $v \in [z_1,z_1+\varepsilon]$.
Let $z_2 = \inf \{ v > z_1 \text{ s.t. } g(v) = 0\}$. By the previous argument we know that $z_2 > z_1 +\epsilon.$ and by continuity of $g$, we have that $g(z_2) = 0$.
Continuity of $g$ also implies that for $v \in [z_1,z_2]$, we have that $g(v) \leq 0$ which implies that $g'(v) \leq 0$ by \eqref{eq:concavity-bands}. Consequently,
\begin{equation*}
    g(z_2) - g(z_1+\epsilon) = \int_{z_1+\epsilon}^{z_2} g'(v) dv \leq 0.
\end{equation*}
This implies that $g(z_2) \leq g(z_1+\epsilon) < 0$, which is a contradiction.
Thus $g=\varphi_{\gamma F}'$ has at most one zero on $(a,b)$.

\smallskip
\noindent\textit{Step 3: Conclusion for $\varphi_{\gamma F}$.}
If $\varphi_{\gamma F}'$ has no zero, then $\varphi_{\gamma F}$ is monotone on $(a,b)$. As $\varphi_{\gamma F}$ is continuous, the intermediate value theorem implies that $\varphi_{\gamma F}(v)=0$ has at most one solution. 
If $\varphi_{\gamma F}'$ has exactly one zero $v^\star\in (a,b)$, then $\varphi_{\gamma F}$ is a continuous function strictly monotone on each of $(a,v^*)$ and $(v^*,b)$. By applying the intermediate value theorem again, we conclude that $\varphi_{\gamma F}$ has at most one zero in each of these intervals. Therefore, $\varphi_{\gamma F}(v) = 0$ has at most two distinct solutions on $(a,b)$.
\end{proof}

\subsection{Proofs of Results in \Cref{sec:noise}}
\label{sec:apx_sensitivity_noise}

This appendix proves \Cref{prop:brittleness_noise} and \Cref{prop:sensitivity-noise}. Both proofs compare mechanisms \emph{within the same noisy environment}, with revenues evaluated using $\Rev_\sigma$, and both rest on the same design logic: the ignore and cap regimes of $\mathcal{M}^*$ post locally constant prices and are therefore insensitive to small perturbations of the signal, whereas the follow regimes price at the signal itself and must be protected by the buffer $\tau$ of \Cref{def:buffered}.
Throughout both proofs, we use the theoretical choice $\tau=c_\sigma\sigma$, where $c_\sigma := \sqrt{2\log(b/\sigma)}$, as in \Cref{prop:sensitivity-noise}.

Throughout, $(v, h, w, \varepsilon)$ are drawn as in \Cref{sec:model}: $v \sim F$, $h \sim \mathrm{Bernoulli}(\gamma)$, $w \sim F$, and $\varepsilon \sim \mathcal{N}(0,\sigma^2)$, mutually independent, and the signal is $s_\sigma = v + \varepsilon$ on the informative event $\{h = 0\}$ and $s_\sigma = w$ on the hallucination event $\{h = 1\}$. The density $f$ is extended by zero outside $[0,b]$. We restrict all posted-price rules to take values in $[0,b]$, which is without loss: a price above $b$ sells with probability zero, exactly like the price $b$. Since $F$ is atomless, the revenue notation of \Cref{sec:opt_1b} reads $\Rev(t,F) = t\,(1 - F(t))$, and we write $\Rev^\star := \max_{t \in [0,b]} \Rev(t,F) = \Rev(p^{\mathrm{ignore}},F) \in [0,b]$ for the monopoly revenue under the prior. Since $|\tfrac{d}{dt}\Rev(t,F)| = |1 - F(t) - t f(t)| \le 1 + b\bar f$ on $[0,b]$, the revenue map is Lipschitz and:
\begin{equation}
\label{eq:apx_rev_lipschitz}
\bigl|\Rev(t,F) - \Rev(t',F)\bigr| \;\le\; L_{\Rev}\,|t - t'| \quad \text{for all } t,t' \in [0,b], \qquad L_{\Rev} := 1 + b\bar f .
\end{equation}
We write $\phi_\sigma$ for the $\mathcal{N}(0,\sigma^2)$ density and $\phi, \Phi$ for the standard normal density and distribution function. As in \Cref{prop:sensitivity-noise}, we fix $\sigma \in (0, b/\sqrt e\,)$ with $\tau \le \kappa/2$; both restrictions hold for every sufficiently small $\sigma$, and the first gives $c_\sigma \ge 1$.

Two elementary estimates are used by both proofs. First, let
\begin{equation*}
    A := \{|\varepsilon| \le \tau\}
\end{equation*}
denote the event that the realized noise is absorbed by the buffer. The choice of $c_\sigma$ gives $\phi(c_\sigma) = \sigma/(b\sqrt{2\pi})$. Combined with the Mills-ratio bound $1 - \Phi(x) \le \phi(x)/x$, valid for every $x > 0$ because $1 - \Phi(x) = \int_x^\infty \phi(t)\,dt \le \int_x^\infty (t/x)\,\phi(t)\,dt = \phi(x)/x$, this yields at $x = c_\sigma \ge 1$
\begin{equation}
\label{eq:apx_noise_tail}
\mathbb{P}(A^c) \;=\; 2\bigl(1 - \Phi(c_\sigma)\bigr) \;\le\; \frac{2\,\phi(c_\sigma)}{c_\sigma}
\;=\; \frac{2\,\sigma}{b\, c_\sigma \sqrt{2\pi}}\,,
\qquad\text{and in particular}\qquad
b\,\mathbb{P}(A^c) \;\le\; \tau ,
\end{equation}
where the last inequality uses $2/\sqrt{2\pi} \le 1 \le c_\sigma^2$. 

Second, we isolate the signals lying close to a regime boundary. For a point $u \in \mathbb{R}$ and a nonempty set $S \subseteq \mathbb{R}$, write $\mathrm{dist}(u, S) := \min_{y \in S} |u - y|$ for the distance from $u$ to $S$. Recall from \Cref{prop:optimal_price_corrected} that $p^*$ partitions $[0,b]$ into the regimes cut out by the boundary set $\mathcal{B} = \{0, L_\gamma, M_\gamma, U_\gamma, b\}$: the constant-price ignore regime $[0, L_\gamma)$ and cap regime $[M_\gamma, U_\gamma]$, on which $p^* \equiv p^{\mathrm{ignore}}$ and $p^* \equiv p^{\mathrm{cap}}$ respectively, and the follow regime $\mathcal{S}^{\mathrm{follow}} = [L_\gamma, M_\gamma) \cup (U_\gamma, b]$, on which $p^*(s) = s$. Let $\mathcal{B}^\circ$ collect the boundary points across which a perturbation of the signal can change the posted price:
\begin{equation}
\label{eq:apx_noise_boundary}
\mathcal{B}^\circ :=
\begin{cases}
\mathcal{B} \setminus \{0\}, & \text{if the ignore regime } [0, L_\gamma) \text{ is nonempty},\\[2pt]
\mathcal{B}, & \text{otherwise,}
\end{cases}
\qquad \text{so that } |\mathcal{B}^\circ| \le 4 .
\end{equation}
The endpoint $0$ may be dropped when the ignore regime is nonempty because both $p^*$ and $\widetilde p^\tau$ then equal the constant $p^{\mathrm{ignore}}$ on all of $(-\infty, L_\gamma)$ --- the extension below $0$ included --- so no perturbation across $0$ changes either price. Define the guard strip of half-width $\tau$ around $\mathcal{B}^\circ$ and its complement,
\begin{equation}
\label{eq:apx_noise_strip_def}
B_\sigma := \bigl\{u \in [0,b] : \mathrm{dist}(u, \mathcal{B}^\circ) \le \tau \bigr\},
\qquad
I_\sigma := [0,b] \setminus B_\sigma .
\end{equation}
As $B_\sigma$ is contained in a union of at most $|\mathcal{B}^\circ| \le 4$ intervals, each of length at most $2 \tau$, and $f \le \bar f$,
\begin{equation}
\label{eq:apx_noise_strip}
\mathbb{P}(v \in B_\sigma) \;=\; \int_{B_\sigma} f(v)\,dv \;\le\; \bar f \cdot 4 \cdot 2 \tau \;=\; 8\, \bar f\, \tau .
\end{equation}
The role of $I_\sigma$ is summarized by the following property, invoked repeatedly below.

\begin{lemma}
\label{lem:locked-regime}
Fix $\sigma$ as above, and let $u \in I_\sigma$ and $\varepsilon \in \mathbb{R}$ with $|\varepsilon| \le \tau$. Then $u + \varepsilon$ lies in the same regime as $u$, and the following hold almost surely:
\begin{enumerate}
\item[\textup{(i)}] if $u$ lies in the ignore or cap regime, then $p^*(u + \varepsilon) = \widetilde p^\tau(u + \varepsilon) = p^*(u)$;
\item[\textup{(ii)}] if $u \in \mathcal{S}^{\mathrm{follow}}$, then $u > \tau$, and $p^*(u + \varepsilon) = u + \varepsilon$ while $\widetilde p^\tau(u + \varepsilon) = \max\{u + \varepsilon - \tau,\, 0\}$.
\end{enumerate}
\end{lemma}

\begin{proof}
By \eqref{eq:apx_noise_strip_def}, $u \in I_\sigma$ gives $\mathrm{dist}(u, \mathcal{B}^\circ) > \tau$, so the interval $[u - \tau,\, u + \tau]$ meets no point of $\mathcal{B}^\circ$; hence $u + \varepsilon$ lies in the same regime as $u$. If that regime is the ignore or cap regime, then $p^*$ is constant on it (equal to $p^{\mathrm{ignore}}$ or $p^{\mathrm{cap}}$), and since $\widetilde p^\tau$ coincides with $p^*$ wherever $p^*(s) \ne s$ (\Cref{def:buffered}), both rules post $p^*(u)$ at $u + \varepsilon$, which is~(i). If $u \in \mathcal{S}^{\mathrm{follow}}$, the follow interval containing $u$ has left endpoint $\ell \in \{L_\gamma, U_\gamma\} \subseteq \mathcal{B}^\circ$ with $\ell \ge 0$, and $u - \ell \ge \mathrm{dist}(u, \mathcal{B}^\circ) > \tau$, so $u > \tau$; moreover $p^*(s) = s$ on the follow regime, so $p^*(u + \varepsilon) = u + \varepsilon$ and $\widetilde p^\tau(u + \varepsilon) = \max\{u + \varepsilon - \tau,\, 0\}$ by \Cref{def:buffered}, which is~(ii). The almost-sure qualifier discards the at most one point of a constant-price regime at which $p^*(s) = s$; the continuously distributed $u + \varepsilon$ avoids it with probability one.
\end{proof}

\begin{proof}[\textbf{Proof of \Cref{prop:brittleness_noise}}]
We prove the proposition in four steps. First, using that $\mathcal{M}^*_\sigma$ maximizes revenue in the noisy environment while the noise-buffered mechanism is feasible, we reduce the claim to a comparison of the two explicit mechanisms $\widetilde{\mathcal{M}}^\tau$ and $\mathcal{M}^*$ using $\Rev_\sigma$. Second, we show that on the hallucination event the two mechanisms earn nearly the same revenue, because their prices differ by at most the buffer. Third, on the informative event we quantify their contrasting exposure to noise: the buffered price always keeps the sale and concedes at most the buffer, while the unbuffered price overshoots the buyer's value half of the time and forfeits the sale. Fourth, we combine these results and conclude.

\textit{Step 1:}
The noisy optimum has no closed form. However, as $\widetilde{\mathcal{M}}^\tau$ is a feasible posted-price rule,
\begin{equation*}
\Rev_\sigma(\mathcal{M}^*_\sigma) - \Rev_\sigma(\mathcal{M}^*)
\;\ge\; \Rev_\sigma(\widetilde{\mathcal{M}}^\tau) - \Rev_\sigma(\mathcal{M}^*),
\end{equation*}
where the inequality follows from the optimality of  $\mathcal{M}^*_\sigma$.
Therefore, in what follows, it suffices to bound from below the RHS. 

Moreover, conditioning on the hallucination indicator, we can decompose the difference between both revenue as,
\begin{equation}
\label{eq:apx_noise_britt_decomp}
\Rev_\sigma(\widetilde{\mathcal{M}}^\tau) - \Rev_\sigma(\mathcal{M}^*)
\;=\; \gamma\;\mathbb{E}_w\bigl[\Rev(\widetilde p^\tau(w),F) - \Rev(p^*(w),F)\bigr]
\;+\; (1-\gamma)\;\mathbb{E}_{v,\varepsilon}\bigl[\,\widetilde r - r^*\,\bigr],
\end{equation}
where
\begin{equation*}
\widetilde r := \widetilde p^\tau(v+\varepsilon)\,\mathbbm{1}\{v \ge \widetilde p^\tau(v+\varepsilon)\},
\qquad
r^* := p^*(v+\varepsilon)\,\mathbbm{1}\{v \ge p^*(v+\varepsilon)\}
\end{equation*}
are the two realized revenues on the informative event. Steps 2 and 3 bound the two terms of \eqref{eq:apx_noise_britt_decomp}.

\textit{Step 2:} We first bound the term conditioning on the signal being a hallucination.
Against a hallucinated signal the two mechanisms face an identical buyer, and by \Cref{def:buffered} their prices differ by at most the buffer: $|\widetilde p^\tau(s) - p^*(s)| \le \tau$ for every $s$. Applying the Lipschitz bound \eqref{eq:apx_rev_lipschitz} to $\Rev(\cdot\,,F)$,
\begin{equation}
\label{eq:apx_noise_britt_hall}
\gamma\;\mathbb{E}_w\bigl[\Rev(\widetilde p^\tau(w),F) - \Rev(p^*(w),F)\bigr]
\;\ge\; -\,\gamma\, L_{\Rev}\; \tau .
\end{equation}

\textit{Step 3:} We bound the informative-branch term $\mathbb{E}_{v,\varepsilon}[\widetilde r - r^*]$ of \eqref{eq:apx_noise_britt_decomp} by conditioning on $v = u$ and comparing the two rules on the buffer event $A$. Fix $u \in I_\sigma$, so that \Cref{lem:locked-regime} applies.

If $u$ lies in the ignore or cap regime, then by \Cref{lem:locked-regime}(i) both rules post $p^*(u)$ at every signal $u + \varepsilon$ with $\varepsilon \in A$, so $\widetilde r\,\mathbbm{1}_A = r^*\,\mathbbm{1}_A$ almost surely.

If $u \in \mathcal{S}^{\mathrm{follow}}$, then by \Cref{lem:locked-regime}(ii) the unbuffered rule posts the corrupted signal, so $r^* = (u + \varepsilon)\,\mathbbm{1}\{u \ge u + \varepsilon\} = (u + \varepsilon)\,\mathbbm{1}\{\varepsilon \le 0\}$: the sale is retained exactly on $\{\varepsilon \le 0\}$, an event of probability $\tfrac12$ by the symmetry of $\varepsilon \sim \mathcal{N}(0,\sigma^2)$. Since $u > \tau$ by \Cref{lem:locked-regime}(ii), on $A \cap \{\varepsilon \le 0\}$ we have $0 \le u + \varepsilon \le u$, so $r^*\,\mathbbm{1}_A \le u\,\mathbbm{1}\{\varepsilon \le 0\}$. Hence, this implies that for every $u \in \mathcal{S}^{\mathrm{follow}} \cap I_\sigma$ we have
\begin{equation*}
\mathbb{E}\bigl[r^*\,\mathbbm{1}_A \,\big|\, u\bigr] \;\le\; \mathbb{E}\bigl[u\,\mathbbm{1}\{\varepsilon \le 0\}\,\big|\, u\bigr] \;=\; \frac{u}{2}\,.
\end{equation*}
Now, let $u \in \mathcal{S}^{\mathrm{follow}} \cap I_\sigma$, and let us derive a lower-bound on the buffered rule's revenue. By \Cref{lem:locked-regime}(ii), on $A$ the buffered rule posts $\widetilde p^\tau(u+\varepsilon) = \max\{u+\varepsilon - \tau,\, 0\} \le u$ (as $\varepsilon \le \tau$ and $u \ge 0$), so the sale is always made and $\widetilde r\,\mathbbm{1}_A = \widetilde p^\tau(u+\varepsilon)\,\mathbbm{1}_A \ge (u + \varepsilon - \tau)\,\mathbbm{1}_A$. Since $\mathbb{E}[\varepsilon\,\mathbbm{1}_A \mid u] = 0$ (the law of $\varepsilon$ and the event $A$ are both symmetric in $\varepsilon$),
\begin{equation*}
\mathbb{E}\bigl[\widetilde r\,\mathbbm{1}_A \,\big|\, u\bigr]
\;\ge\; (u - \tau)\,\mathbb{P}(A)
\;=\; (u - \tau) - (u - \tau)\,\mathbb{P}(A^c)
\;\ge\; u - \tau - b\,\mathbb{P}(A^c),
\end{equation*}
the last inequality using $0 \le u - \tau \le b$. Subtracting this from the bound $\mathbb{E}[r^*\,\mathbbm{1}_A \mid u] \le u/2$ established above, for every $u \in \mathcal{S}^{\mathrm{follow}} \cap I_\sigma$,
\begin{equation}
\label{eq:apx_noise_britt_follow}
\mathbb{E}\bigl[(\widetilde r - r^*)\,\mathbbm{1}_A \,\big|\, u\bigr]
\;\ge\; \frac{u}{2} \;-\; \tau \;-\; b\,\mathbb{P}(A^c).
\end{equation}

\textit{Step 4:} We assemble the informative-branch bound. Both realized revenues lie in $[0,b]$, so $|\widetilde r - r^*| \le b$ pointwise. Partition the expectation according to the three events $A^c$, $A \cap \{v \in B_\sigma\}$, and $A \cap \{v \in I_\sigma\}$:
\begin{equation}
\label{eq:apx_noise_britt_partition}
\mathbb{E}_{v,\varepsilon}\bigl[\widetilde r - r^*\bigr]
= \mathbb{E}\bigl[(\widetilde r - r^*)\mathbbm{1}_{A^c}\bigr]
+ \mathbb{E}\bigl[(\widetilde r - r^*)\mathbbm{1}_{A}\mathbbm{1}\{v \in B_\sigma\}\bigr]
+ \mathbb{E}\bigl[(\widetilde r - r^*)\mathbbm{1}_{A}\mathbbm{1}\{v \in I_\sigma\}\bigr],
\end{equation}
and bound the three terms in turn. First, as $|\widetilde r - r^*| \le b$, the first term is at least $-\,b\,\mathbb{P}(A^c)$. Likewise the second term is at least $-\,b\,\mathbb{P}(v \in B_\sigma) \ge -\,8\, b\bar f\, \tau$ by \eqref{eq:apx_noise_strip}. Third, by \Cref{lem:locked-regime}(i), $\widetilde r\,\mathbbm{1}_A = r^*\,\mathbbm{1}_A$ almost surely when $v \in I_\sigma$ lies in the ignore or cap regime, so only the follow regime contributes; applying \eqref{eq:apx_noise_britt_follow},
\begin{equation*}
\mathbb{E}\bigl[(\widetilde r - r^*)\mathbbm{1}_{A}\mathbbm{1}\{v \in I_\sigma\}\bigr]
= \int_{\mathcal{S}^{\mathrm{follow}} \cap I_\sigma} \mathbb{E}\bigl[(\widetilde r - r^*)\mathbbm{1}_A \,\big|\, u\bigr]\, dF(u)
\;\ge\; \int_{\mathcal{S}^{\mathrm{follow}} \cap I_\sigma} \Bigl(\frac{u}{2} - \tau - b\,\mathbb{P}(A^c)\Bigr) dF(u).
\end{equation*}
Then, by splitting the last integral we get
\begin{align*}
\int_{\mathcal{S}^{\mathrm{follow}} \cap I_\sigma} \frac{u}{2}\, dF(u)
&\;\ge\; \frac{1}{2}\int_{\mathcal{S}^{\mathrm{follow}}} u\, dF(u) \;-\; \frac{b}{2}\,\mathbb{P}(v \in B_\sigma), \\
\int_{\mathcal{S}^{\mathrm{follow}} \cap I_\sigma} \bigl(\tau + b\,\mathbb{P}(A^c)\bigr) dF(u)
&\;\le\; \tau + b\,\mathbb{P}(A^c),
\end{align*}
where the first line using $0 \le u \le b$ on $\mathcal{S}^{\mathrm{follow}} \cap B_\sigma$. 

Substituting the three bounds into \eqref{eq:apx_noise_britt_partition}, we obtain
\begin{equation}
\label{eq:apx_noise_britt_collect}
\mathbb{E}_{v,\varepsilon}\bigl[\widetilde r - r^*\bigr]
\;\ge\; \frac{1}{2}\int_{\mathcal{S}^{\mathrm{follow}}} u\, dF(u)
\;-\; \frac{b}{2}\,\mathbb{P}(v \in B_\sigma) \;-\; \tau \;-\; 2\,b\,\mathbb{P}(A^c) \;-\; 8\, b\bar f\, \tau .
\end{equation}
Finally, $b\,\mathbb{P}(A^c) \le \tau$ by \eqref{eq:apx_noise_tail} and $\tfrac{b}{2}\,\mathbb{P}(v \in B_\sigma) \le 4\, b\bar f\, \tau$ by \eqref{eq:apx_noise_strip}, hence
\begin{equation}
\label{eq:apx_noise_britt_info}
\mathbb{E}_{v,\varepsilon}\bigl[\widetilde r - r^*\bigr]
\;\ge\; \frac{1}{2}\int_{\mathcal{S}^{\mathrm{follow}}} u\, dF(u)
\;-\; \bigl(3 + 12\, b\bar f\bigr)\, \tau .
\end{equation}

It remains to combine the two branches. Substituting the hallucination bound \eqref{eq:apx_noise_britt_hall} and the informative bound \eqref{eq:apx_noise_britt_info} into the decomposition \eqref{eq:apx_noise_britt_decomp}, and recalling $\Rev_\sigma(\mathcal{M}^*_\sigma) - \Rev_\sigma(\mathcal{M}^*) \ge \Rev_\sigma(\widetilde{\mathcal{M}}^\tau) - \Rev_\sigma(\mathcal{M}^*)$ from Step 1,
\begin{equation*}
\Rev_\sigma(\mathcal{M}^*_\sigma) - \Rev_\sigma(\mathcal{M}^*)
\;\ge\; \frac{1-\gamma}{2}\int_{\mathcal{S}^{\mathrm{follow}}} u\, dF(u)
\;-\; \bigl[\,\gamma\, L_{\Rev} + (1-\gamma)\,(3 + 12\, b\bar f)\,\bigr]\, \tau .
\end{equation*}
The bracketed coefficient is a convex combination of $L_{\Rev} = 1 + b\bar f$ and $3 + 12\, b\bar f$; since $L_{\Rev} \le 3 + 12\, b\bar f$, it is at most $3 + 12\, b\bar f$. Substituting $\tau = \sigma\sqrt{2\log(b/\sigma)}$,
\begin{equation*}
\Rev_\sigma(\mathcal{M}^*_\sigma) - \Rev_\sigma(\mathcal{M}^*)
\;\ge\; \frac{1-\gamma}{2}\int_{\mathcal{S}^{\mathrm{follow}}} u\, dF(u)
\;-\; \bigl(3 + 12\, b \bar f\bigr)\,\sigma\sqrt{2\log(b/\sigma)} ,
\end{equation*}
which is the claimed bound with $C = 3 + 12\, b\bar f$. Since $\sigma\sqrt{2\log(b/\sigma)} \to 0$ as $\sigma \downarrow 0$, taking $\liminf_{\sigma \downarrow 0}$ yields the second display of the proposition.
\end{proof}

\begin{proof}[\textbf{Proof of \Cref{prop:sensitivity-noise}}]
Below, $[x]_+ := \max\{x,0\}$ and $[x]_- := \max\{-x,0\}$. 

\textit{Step 1:}
For $t \in [0,b]$ and $s \in \mathbb{R}$, define the revenue density with respect to the signal
\begin{equation}
\label{eq:apx_noise_persignal}
\Rev^{\mathrm{sig}}_\sigma(t; s) \;:=\; \gamma\, f(s)\, \Rev(t,F)
\;+\; (1-\gamma)\; t \int_t^b f(v)\, \phi_\sigma(s - v)\, dv ,
\end{equation}
and note that conditioning on $h$ and applying Fubini's theorem gives $\Rev_\sigma(p) = \int_{\mathbb{R}} \Rev^{\mathrm{sig}}_\sigma(p(s); s)\, ds$ for every measurable rule $p$, hence
\begin{equation*}
\sup_p \Rev_\sigma(p) \;\le\; \int_{\mathbb{R}} \sup_{t \in [0,b]} \Rev^{\mathrm{sig}}_\sigma(t; s)\, ds ,
\qquad \mbox{and} \qquad
\Rev_\sigma\bigl(\widetilde{\mathcal{M}}^\tau\bigr) \;=\; \int_{\mathbb{R}} \Rev^{\mathrm{sig}}_\sigma\bigl(\widetilde p^\tau(s); s\bigr)\, ds.
\end{equation*}
Therefore, we have that
\begin{equation}
\label{eq:apx_noise_gap_target}
\sup_p \Rev_\sigma(p) - \Rev_\sigma\bigl(\widetilde{\mathcal{M}}^\tau\bigr)
\;\le\; \int_{\mathbb{R}} \Bigl[\,\sup_{t \in [0,b]} \Rev^{\mathrm{sig}}_\sigma(t;s)
- \Rev^{\mathrm{sig}}_\sigma\bigl(\widetilde p^\tau(s); s\bigr)\Bigr]\, ds .
\end{equation}

\textit{Step 2:} We next relate the noisy per-signal problem and the noise-free ones.
For $v \in [0,b]$ and $t \in [0,b]$, let
\begin{equation*}
\mathrm{Reg}_v(t) \;:=\; \Rev\bigl(p^*(v), F_{\gamma,v}\bigr) - \Rev\bigl(t, F_{\gamma,v}\bigr)
\end{equation*}
be the noise-free per-signal regret of posting $t$ at signal $v$.

To rewrite the expression of $\Rev^{\mathrm{sig}}_\sigma(t; s)$ we first note that \begin{equation}
\label{eq:apx_noise_info_term}
(1-\gamma)\, t \int_t^b f(v)\, \phi_\sigma(s-v)\, dv
\;=\; (1-\gamma)\,\mathbb{E}_\varepsilon\bigl[f(s-\varepsilon)\, t\, \mathbbm{1}\{t \le s-\varepsilon\}\bigr],
\end{equation}
where the equality follows from substituting $v = s - \varepsilon$ and using $\phi_\sigma(-\varepsilon) = \phi_\sigma(\varepsilon)$ together with $f = 0$ off $[0,b]$.

For every $s \in \mathbb{R}$, let $\delta(s) = \mathbb{E}_\varepsilon\bigl[f(s) - f(s-\varepsilon)\bigr]$ so that $f(s) = \mathbb{E}_\varepsilon\bigl[f(s-\varepsilon)\bigr] + \delta(s).$
Substituting \eqref{eq:apx_noise_info_term} into \eqref{eq:apx_noise_persignal}, and then expanding $\Rev(\cdot\,,F_{\gamma,\,s-\varepsilon})$, we obtain that
\begin{align}
\Rev^{\mathrm{sig}}_\sigma(t; s)
&\;=\; \gamma\,\Rev(t,F)\,\bigl(\mathbb{E}_\varepsilon\bigl[f(s-\varepsilon)\bigr] + \delta(s)\bigr)
\;+\; (1-\gamma)\,\mathbb{E}_\varepsilon\bigl[f(s-\varepsilon)\, t\, \mathbbm{1}\{t \le s-\varepsilon\}\bigr] \nonumber \\
&\;=\; \mathbb{E}_\varepsilon\Bigl[f(s-\varepsilon)\,\bigl(\gamma\,\Rev(t,F) + (1-\gamma)\, t\, \mathbbm{1}\{t \le s-\varepsilon\}\bigr)\Bigr]
\;+\; \gamma\,\Rev(t,F)\,\delta(s) \nonumber \\
\label{eq:apx_noise_identity}
&\;=\; \mathbb{E}_\varepsilon\bigl[f(s-\varepsilon)\,\Rev\bigl(t, F_{\gamma,\, s-\varepsilon}\bigr)\bigr]
\;+\; \gamma\, \Rev(t,F)\, \delta(s) ,
\end{align}
for all $t \in [0,b]$ and $s \in \mathbb{R}$.

Consequently, recalling that $\Rev^\star = \max_{t \in [0,b]} \Rev(t,F) \in [0,b]$ denotes the monopoly revenue under the prior,
\begin{align}
\sup_{t \in [0,b]} \Rev^{\mathrm{sig}}_\sigma(t; s)
&\;\le\; \mathbb{E}_\varepsilon\Bigl[\sup_{t \in [0,b]} f(s-\varepsilon)\,\Rev\bigl(t, F_{\gamma,\, s-\varepsilon}\bigr)\Bigr]
\;+\; \gamma\,\Rev^\star\, [\delta(s)]_+ \nonumber \\
&\;=\; \mathbb{E}_\varepsilon\bigl[f(s-\varepsilon)\,\Rev\bigl(p^*(s-\varepsilon), F_{\gamma,\, s-\varepsilon}\bigr)\bigr]
\;+\; \gamma\,\Rev^\star\, [\delta(s)]_+,
\label{eq:apx_noise_sup}
\end{align}
where the first inequality follows from $0 \le \Rev(t,F) \le \Rev^\star$  and from the exchange of the supremum with $\mathbb{E}_\varepsilon$, whereas the equality uses \Cref{prop:optimal_price_corrected}.

Moreover, \eqref{eq:apx_noise_identity} at $t = \widetilde p^\tau(s)$ implies that
\begin{equation}
\label{eq:apx_noise_buffered}
\Rev^{\mathrm{sig}}_\sigma\bigl(\widetilde p^\tau(s); s\bigr)
\;\ge\; \mathbb{E}_\varepsilon\bigl[f(s-\varepsilon)\,\Rev\bigl(\widetilde p^\tau(s), F_{\gamma,\, s-\varepsilon}\bigr)\bigr]
\;-\; \gamma\,\Rev^\star\, [\delta(s)]_- .
\end{equation}
Subtracting \eqref{eq:apx_noise_buffered} from \eqref{eq:apx_noise_sup} and using that $[\delta(s)]_+ + [\delta(s)]_- = |\delta(s)|$ give,
\begin{equation}
\label{eq:apx_noise_persignal_gap}
\sup_{t \in [0,b]} \Rev^{\mathrm{sig}}_\sigma(t;s) - \Rev^{\mathrm{sig}}_\sigma\bigl(\widetilde p^\tau(s); s\bigr)
\;\le\; \mathbb{E}_\varepsilon\Bigl[f(s-\varepsilon)\;\mathrm{Reg}_{s-\varepsilon}\bigl(\widetilde p^\tau(s)\bigr)\Bigr]
\;+\; \gamma\,\Rev^\star\, |\delta(s)| .
\end{equation}
The first term of \eqref{eq:apx_noise_persignal_gap} is nonnegative and jointly measurable, so Tonelli's theorem and, at each fixed $\varepsilon$, the substitution $v = s - \varepsilon$ give
\begin{align}
\int_{\mathbb{R}} \mathbb{E}_\varepsilon\Bigl[f(s-\varepsilon)\,\mathrm{Reg}_{s-\varepsilon}\bigl(\widetilde p^\tau(s)\bigr)\Bigr]\, ds
&\;=\; \int_{\mathbb{R}}\!\int_{\mathbb{R}} \phi_\sigma(\varepsilon)\, f(v)\,
\mathrm{Reg}_v\bigl(\widetilde p^\tau(v+\varepsilon)\bigr)\, dv\, d \varepsilon \nonumber \\
\label{eq:apx_noise_change_var}
&\;=\; \mathbb{E}_{v,\varepsilon}\Bigl[\mathrm{Reg}_v\bigl(\widetilde p^\tau(v+\varepsilon)\bigr)\Bigr] .
\end{align}
By integrating \eqref{eq:apx_noise_persignal_gap} over $s$ and combining with \eqref{eq:apx_noise_gap_target} and \eqref{eq:apx_noise_change_var}, we conclude that,
\begin{equation}
\label{eq:apx_noise_master}
\sup_p \Rev_\sigma(p) - \Rev_\sigma\bigl(\widetilde{\mathcal{M}}^\tau\bigr)
\;\le\; \mathbb{E}_{v,\varepsilon}\Bigl[\mathrm{Reg}_v\bigl(\widetilde p^\tau(v+\varepsilon)\bigr)\Bigr]
\;+\; \gamma\,\Rev^\star \int_{\mathbb{R}} |\delta(s)|\, ds .
\end{equation}
The first term of \eqref{eq:apx_noise_master} is the expected \emph{noise-free} regret of the buffered rule when it is fed the corrupted signal $v + \varepsilon$. Our goal in the next step is to bound this term.

\textit{Step 3:} We next establish that,
\begin{equation}
\label{eq:apx_noise_T_final}
\mathbb{E}_{v,\varepsilon}\Bigl[\mathrm{Reg}_v\bigl(\widetilde p^\tau(v+\varepsilon)\bigr)\Bigr]
\leq \bigl(2 + 9\, b \bar f\bigr)\, \tau .
\end{equation}

Recall the buffer event $A = \{|\varepsilon| \le \tau\}$, the guard strip $B_\sigma$ and its complement $I_\sigma = [0,b] \setminus B_\sigma$ from \eqref{eq:apx_noise_strip_def}. Splitting on the events $A^c$, $A \cap \{v \in B_\sigma\}$ and $A \cap \{v \in I_\sigma\}$, and bounding $\mathrm{Reg}_v(\cdot) \le b$ on the first two we get,
\begin{equation}
\label{eq:apx_noise_reg_split}
\mathbb{E}_{v,\varepsilon}\Bigl[\mathrm{Reg}_v\bigl(\widetilde p^\tau(v+\varepsilon)\bigr)\Bigr]
\;\le\; b\,\mathbb{P}(A^c) \;+\; b\,\mathbb{P}(v \in B_\sigma)
\;+\; \mathbb{E}_{v,\varepsilon}\Bigl[\mathrm{Reg}_v\bigl(\widetilde p^\tau(v+\varepsilon)\bigr)\,
\mathbbm{1}_A\, \mathbbm{1}\{v \in I_\sigma\}\Bigr].
\end{equation}
Fix $v \in I_\sigma$ and $\varepsilon \in A$, and let $t = \widetilde p^\tau(v+\varepsilon)$. If $v$ lies in the ignore or cap regime, then $t = p^*(v)$ almost surely by \Cref{lem:locked-regime}(i), so $\mathrm{Reg}_v(t) = 0$. If $v \in \mathcal{S}^{\mathrm{follow}}$, then $p^*(v) = v$ by \Cref{prop:optimal_price_corrected} and $t = \max\{v+\varepsilon-\tau,\, 0\}$ by \Cref{lem:locked-regime}(ii); since $\varepsilon \le \tau$ on $A$, this gives $t \le v$. Hence, for every $v \in \mathcal{S}^{\mathrm{follow}} \cap I_\sigma$ and $\varepsilon \in A$,
\begin{equation*}
\mathrm{Reg}_v(t) \;=\; (1-\gamma)\,(v - t) + \gamma\bigl[\Rev(v,F) - \Rev(t,F)\bigr]
\;\stackrel{(a)}{\le}\; \bigl[(1-\gamma) + \gamma\, L_{\Rev}\bigr](v-t)
\; \stackrel{(b)}{\le}\; L_{\Rev}\,(\tau - \varepsilon),
\end{equation*}
where $(a)$ follows from \eqref{eq:apx_rev_lipschitz} and $(b)$ holds because $L_{\Rev} \geq 1$ and $v - t \leq \tau - \varepsilon$.
By taking an expectation over $\varepsilon$ at this fixed $v \in \mathcal{S}^{\mathrm{follow}} \cap I_\sigma$, we obtain that,
\begin{equation}
\label{eq:apx_noise_reg_follow_avg}
\mathbb{E}_\varepsilon\Bigl[\mathrm{Reg}_v\bigl(\widetilde p^\tau(v+\varepsilon)\bigr)\,\mathbbm{1}_A\Bigr]
\;\le\; L_{\Rev}\bigl(\tau\,\mathbb{P}(A) - \mathbb{E}_\varepsilon[\varepsilon\,\mathbbm{1}_A]\bigr)
\;\le\; L_{\Rev}\, \tau,
\end{equation}
where the second inequality uses $\mathbb{E}_\varepsilon[\varepsilon\,\mathbbm{1}_A] = 0$, by the symmetry in $\varepsilon$ of both the law of $\varepsilon$ and the event $A$, together with $\mathbb{P}(A) \le 1$.

 We conclude that \eqref{eq:apx_noise_reg_split} can be bounded as follows,
\begin{equation*}
\mathbb{E}_{v,\varepsilon}\Bigl[\mathrm{Reg}_v\bigl(\widetilde p^\tau(v+\varepsilon)\bigr)\Bigr]
\;\le\; \tau \;+\; 8\, b\bar f\, \tau \;+\; L_{\Rev}\, \tau
\;=\; \bigl(2 + 9\, b \bar f\bigr)\, \tau,
\end{equation*}
where the first term is bounded using \eqref{eq:apx_noise_tail}, the second follows from \eqref{eq:apx_noise_strip}, and the third is obtained by \eqref{eq:apx_noise_reg_follow_avg}. The equality holds because $L_{\Rev} = 1 + b\bar f$.

This proves \eqref{eq:apx_noise_T_final}.

\textit{Step 4:}
We now establish that,
\begin{equation}
\label{eq:apx_noise_mismatch_final}
\gamma\,\Rev^\star \int_{\mathbb{R}} |\delta(s)|\, ds \;\le\; 2\, b\bar f\, \sigma\sqrt{2/\pi}
\;\le\; 2\, b \bar f\, \tau,
\end{equation}
and conclude.

For $g : \mathbb{R} \to \mathbb{R}$ integrable, write $\|g\|_{L^1} := \int_{\mathbb{R}} |g(x)|\, dx$, and let
\begin{equation*}
\mathrm{TV}(g) \;:=\; \sup\Bigl\{ \sum_{k=1}^{K} \bigl|g(x_k) - g(x_{k-1})\bigr| \;:\; K \in \mathbb{N},\ x_0 < x_1 < \dots < x_K \Bigr\}
\end{equation*}
denote its total variation over $\mathbb{R}$. For every such $g$ with $\mathrm{TV}(g) < \infty$ and every $t \in \mathbb{R}$,
\begin{equation}
\label{eq:apx_noise_shift_tv}
\bigl\|g - g(\cdot - t)\bigr\|_{L^1} \;\le\; \mathrm{TV}(g)\,|t| .
\end{equation}
Indeed, the substitution $x \mapsto x + t$ leaves the left-hand side of \eqref{eq:apx_noise_shift_tv} unchanged when $t$ is replaced by $-t$, so we may take $t \ge 0$. Let $\mu_g$ be the Lebesgue--Stieltjes measure of the right-continuous version of $g$, so that $g(x) - g(x-t) = \mu_g\bigl((x-t,x]\bigr)$ for almost every $x$ and $|\mu_g|(\mathbb{R}) = \mathrm{TV}(g)$. Tonelli's theorem then gives
\begin{equation*}
\bigl\|g - g(\cdot - t)\bigr\|_{L^1}
\;=\; \int_{\mathbb{R}} \bigl|\mu_g\bigl((x-t,x]\bigr)\bigr|\, dx
\;\le\; \int_{\mathbb{R}} \int_{\mathbb{R}} \mathbbm{1}\{x - t < y \le x\}\, dx\; |\mu_g|(dy)
\;=\; \mathrm{TV}(g)\, t .
\end{equation*}
Moreover, $\varepsilon \sim \mathcal{N}(0,\sigma^2)$ gives
\begin{equation*}
\mathbb{E}|\varepsilon| \;=\; 2\int_0^{\infty} x\,\phi_\sigma(x)\, dx
\;=\; \frac{2}{\sigma\sqrt{2\pi}} \int_0^{\infty} x\, e^{-x^2/(2\sigma^2)}\, dx
\;=\; \frac{2\sigma}{\sqrt{2\pi}} \;=\; \sigma\sqrt{2/\pi} .
\end{equation*}
Hence, we have
\begin{equation}
\label{eq:apx_noise_mismatch}
\int_{\mathbb{R}} |\delta(s)|\, ds
\;\stackrel{(a)}{\le}\; \mathbb{E}_\varepsilon \bigl\|f - f(\cdot - \varepsilon)\bigr\|_{L^1} \stackrel{(b)}{\leq} \mathrm{TV}(f) \mathbb{E}|\varepsilon|
\;\stackrel{(c)}{\le}\; 2\,\bar f\; \mathbb{E}|\varepsilon|
\;=\; 2\,\bar f\, \sigma\sqrt{2/\pi}\, ,
\end{equation}
where $(a)$ follows from the triangle inequality and Tonelli's theorem, $(b)$ applies \eqref{eq:apx_noise_shift_tv} to $g = f$ at $t = \varepsilon$ and averages over $\varepsilon$,
and $(c)$ holds because $f$ is log-concave on $[0,b]$ hence unimodal with some mode $m$, which implies that
\begin{equation*}
\mathrm{TV}(f) \;=\; \underbrace{f(0)}_{\text{jump at } 0} + \underbrace{\bigl[f(m) - f(0)\bigr]}_{\text{rise to } m} + \underbrace{\bigl[f(m) - f(b)\bigr]}_{\text{fall to } b} + \underbrace{f(b)}_{\text{jump at } b} \;=\; 2 f(m) \;=\; 2\bar f .
\end{equation*}
Combined with $\Rev^\star \le b$, $\gamma < 1$ and $\sqrt{2/\pi} < 1 \le c_\sigma$, the bound \eqref{eq:apx_noise_mismatch} yields \eqref{eq:apx_noise_mismatch_final}.

Finally, since $\tau = \sigma\sqrt{2\log(b/\sigma)}$ and $\Rev_\sigma(\mathcal{M}^*_\sigma) \le \sup_p \Rev_\sigma(p)$, substituting \eqref{eq:apx_noise_T_final} and \eqref{eq:apx_noise_mismatch_final} into \eqref{eq:apx_noise_master} yields
\begin{align*}
\Rev_\sigma(\mathcal{M}^*_\sigma) - \Rev_\sigma\bigl(\widetilde{\mathcal{M}}^\tau\bigr)
&\;\le\; \bigl(2 + 11\, b\bar f\bigr)\, \sigma\sqrt{2\log(b/\sigma)}
\end{align*}
which concludes the proof.
\end{proof}

\subsection{Proofs of Results in \Cref{sec:gamma-robust}}

\begin{proof}[\textbf{Proof of \Cref{thm:gamma_robust:single_local}}]
For fixed $p$ and $s$, $\Rev(p^*_{\gamma,s},F_{\gamma,s})-\Rev(p,F_{\gamma,s})$ is convex in $\gamma$. Indeed, the second term is affine in $\gamma$, and $\Rev(p^*_{\gamma,s},F_{\gamma,s}) = \sup_p\Rev(p,F_{\gamma,s})$ is the pointwise supremum of a family of affine functions of $\gamma$ and hence convex. A convex function on $[0,1]$ attains its maximum at an endpoint. At $\gamma=0$ the optimal price is $p^*_{0,s}=s$, the benchmark satisfies $\Rev(p^*_{0,s},F_{0,s})=s$, and $\Rev(p,F_{0,s}) = p\mathbbm{1}\{p\le s\}$ for every price $p$. At $\gamma=1$ the benchmark is $p^*_{1,s}=p_F$, $\Rev(p^*_{1,s},F_{1,s})=\Rev(p_F,F)$, and the candidate collects $\Rev(p,F_{1,s}) = \Rev(p,F)$. Hence
\begin{equation*}
    \mathrm{WReg}_{s,F}(p) = \max\bigl\{\Delta_0(p,s),\Delta_1(p)\bigr\},
\end{equation*}
where
\begin{equation*}
    \Delta_0(p,s)= \begin{cases}
        s-p \quad \text{if $p\le s$}\\
        s \quad \text{otherwise}
    \end{cases}
    \quad \mbox{and} \quad \Delta_1(p)=\Rev(p_F,F)-\Rev(p,F).
\end{equation*}

\emph{Case $s\le p_F$.} On $[0,s]$ both $\Delta_0(\cdot,s)$ and $\Delta_1(\cdot)$ are weakly decreasing (the latter because $\Rev(\cdot,F)$ is strictly increasing on $[0,p_F]$ and $s\le p_F$), so the best price on $[0,s]$ is $p=s$, with $\max\bigl\{\Delta_0(s,s),\Delta_1(s)\bigr\} = \Rev(p_F,F)-\Rev(s,F)$. On $(s,b]$ we note that $\Delta_0(\cdot,s)$ is constant equal to $s$ and $\Delta_1$ is minimized at $p_F\in[s,b]$, giving $\Delta_1(p_F)=0$ and  $\max\bigl\{\Delta_0(p_F,s),\Delta_1(p_F)\bigr\} =  s$. The minimum of $\mathrm{WReg}_{s,F}$ is therefore $\min\{s,\Rev(p_F,F)-\Rev(s,F)\}$. The map $s\mapsto s+\Rev(s,F)-\Rev(p_F,F)$ is strictly increasing on $[0,p_F]$, negative at $s=0$ and positive at $s=p_F$, so by the intermediate value theorem there is a unique crossing $\underline s\in(0,p_F)$, and the optimal price is $p_F$ on $[0,\underline s]$ and $s$ on $(\underline s,p_F]$. This proves the first two cases.

\emph{Case $s>p_F$.} We first show that the minimizer must lie in $[p_F,s]$. For any $p>s$ we have $\Delta_0(p,s)=s$, hence $\mathrm{WReg}_{s,F}(p)\ge s$; this is strictly worse than posting $p_F$, which gives $\mathrm{WReg}_{s,F}(p_F)=\max\{s-p_F,\,0\}=s-p_F<s$. For any $p<p_F$ we have $\Delta_0(p,s)=s-p>s-p_F$ and $\Delta_1(p)\ge 0$, so $\mathrm{WReg}_{s,F}(p)\ge s-p>s-p_F$, again strictly worse than $p_F$. It therefore suffices to optimize over $p\in[p_F,s]$. On that interval $\Delta_0(p,s)=s-p$ is strictly decreasing while $\Delta_1(p)=\Rev(p_F,F)-\Rev(p,F)$ is strictly increasing, with $\Delta_0(p_F,s)=s-p_F>0=\Delta_1(p_F)$ at the left endpoint and $\Delta_0(s,s)=0\le\Delta_1(s)$ at the right. Hence $\mathrm{WReg}_{s,F}=\max\{\Delta_0,\Delta_1\}$ is minimized at the unique point $p_s\in[p_F,s]$ where the two curves cross, i.e.\ $s-p_s=\Rev(p_F,F)-\Rev(p_s,F)$, giving $\mathrm{WReg}_{s,F}(p_s)=s-p_s$.
\end{proof}

\begin{proof}[\textbf{Proof of \Cref{cor:gamma_robust:half}}]
Fix $\gamma \in (0,1)$. The seller's optimal revenue at signal $s \le p_F$ equals
\begin{equation*}
    \max\Bigl\{\max_{p \le s}\,\bigl[(1-\gamma)\,p + \gamma\,\Rev(p,F)\bigr],\;\; \gamma\,\Rev(p_F,F)\Bigr\},
\end{equation*}
where the second branch is the best revenue from prices above the signal, attained at $p_F$. On $[0,p_F]$ both $p$ and $\Rev(p,F)$ are non-decreasing, so the first branch equals $(1-\gamma)\,s + \gamma\,\Rev(s,F)$ and is strictly increasing in $s$, while the second branch is constant; the first branch equals $0$ at $s=0$ and exceeds the second at $s=p_F$, so the two branches cross at a unique point of $(0,p_F)$. By \Cref{prop:optimal_price_corrected}, the optimal price is $p^{\mathrm{ignore}} = p_F$ (the second branch) for $s < L_\gamma$ and $p = s$ (the first branch) for $L_\gamma \le s < M_\gamma$; hence $L_\gamma$ is that crossing point:
\begin{equation*}
    (1-\gamma)\,L_\gamma + \gamma\,\Rev(L_\gamma,F) = \gamma\,\Rev(p_F,F).
\end{equation*}
Setting $\gamma = 1/2$ and rearranging gives $L_{1/2} = \Rev(p_F,F) - \Rev(L_{1/2},F)$, which is the defining equation \eqref{eq:gr_local_low_threshold} of $\underline s$. Any solution of that equation satisfies $s \le \Rev(p_F,F) \le p_F$, and its solution in $[0,p_F]$ is unique; hence $L_{1/2} = \underline s$.
\end{proof}

\subsection{Regular Counterexample with more than Four Regimes}
\label{app:counterexample-many-regimes}

This appendix provides a regular distribution for which the single-buyer optimal posted price under hallucinations exhibits more than four regimes.

We consider the distribution $F$ which is a mixture of the $\mathrm{Beta}(4,6)$ and the $\mathrm{Beta}(4,1)$ distributions with respective weights $3/4$ and $1/4$. Figure~\ref{fig:counterexample_two_csv} illustrates the key properties of this counterexample. 
\Cref{fig:virtual_counter} plots the virtual value function associated with $F$, while \Cref{fig:price_counter} reports the optimal posted price as a function of the hallucinated signal $s$.

\begin{figure}[h!]
\centering
\subfigure[virtual value]{
\begin{tikzpicture}[scale=0.5]
\begin{axis}[
    width=\textwidth,
    height=0.62\textwidth,
    xlabel={$v$},
    ylabel={virtual value $\phi(v)$},
    grid=both,
    table/col sep=comma,
    unbounded coords=jump, %
    ymin = -25,
]
\addplot[thick, blue, restrict x to domain=0.05:1, line width = .5mm] table[x=v, y=virtual_val] {Data/counterexample_virtual_values.csv};
\end{axis}
\end{tikzpicture}
\label{fig:virtual_counter}
}
\hfill
\subfigure[Optimal price as a function of $s$.]{
\centering
\begin{tikzpicture}[scale=0.5]
\begin{axis}[
    width=\textwidth,
    height=0.62\textwidth,
    xlabel={$s$},
    ylabel={best price $p^\star(s)$},
    grid=both,
    table/col sep=comma,
    unbounded coords=jump,
    ymin=0, ymax=1, %
]
\addplot[thick, red, line width = .5mm] table[x=s, y=best_price] {Data/counterexample_best_prices.csv};
\end{axis}
\end{tikzpicture}
\label{fig:price_counter}
}
\caption{Counterexample: virtual values and optimal prices}
\label{fig:counterexample_two_csv}
\end{figure}
From \Cref{fig:virtual_counter}, we verify that the counterexample distribution is indeed regular: the virtual value is non-decreasing over the support.
Despite this regularity, \Cref{fig:price_counter} shows that the optimal posted price under hallucinations exhibits five distinct regimes.
In particular, a second capped pricing regime emerges for high values of the signal $s$, leading to more than four pricing regions.

\section{Proof of Results in \Cref{sec:heuristic}}

\begin{proof}[\textbf{Proof of \Cref{thm:signal_worse_than_cap}}]
Fix $\gamma>0$, a regular distribution $F$ with continuous density on $[a,b]$, 
and signals $(s_1,s_2)$ with $s_1>s_2$.
Throughout the proof we condition on $(s_1,s_2)$ and suppress this conditioning from the
notation. Let
\begin{equation*}
    F_i := F_{\gamma,s_i},\qquad 
v_i\sim F_i \ \text{independently across }i,\qquad
p_i^* := p_i^*(s_i),\qquad i\in\{1,2\}.
\end{equation*}
Using the notation of \Cref{sec:opt_1b}, the revenue from posting $r$ to buyer $i$ is $\Rev(r,F_i) = r\cdot\mathbb{P}_{v_i\sim F_i}(v_i\ge r)$.
By definition of $p_i^*$, for all $r\in[a,b]$,
\begin{equation}
\label{eq:monopoly-optimality}
\Rev(p_i^*,F_i)  \;\ge\;  \Rev(r,F_i).
\end{equation}

In the signal–eager mechanism $\pi_E$ we use personalized reserves
\[
r_1^E := s_1,\qquad r_2^E := s_2,
\]
whereas in the capped mechanism $\pi_C$ we use
\[
r_1^C := \max\{s_1,p_1^*\},\qquad r_2^C := p_2^*.
\]
Recall that in any eager second–price auction with reserves $(r_1,r_2)$, buyer $i$ is active iff $v_i \ge r_i$, the good is allocated to the active bidder with the highest value (ties broken in favour of buyer~1, which is without loss of generality since under our assumption that valuations admit a continuous density, tie events occur with probability $0$), and the winner pays the maximum between her reserve and the highest competing bid.

We will prove the stronger statement
\begin{equation*}
\mathbb{E}\bigl[p_1^{\pi_C} (\bm{v}, \bm{s}) \vert \bm{s} \bigr] \;\ge\; \mathbb{E}\bigl[p_1^{\pi_E} (\bm{v}, \bm{s}) \vert \bm{s} \bigr]
\quad\text{and}\quad
\mathbb{E}\bigl[p_2^{\pi_C} (\bm{v}, \bm{s}) \vert \bm{s} \bigr] \;\ge\; \mathbb{E}\bigl[p_2^{\pi_E} (\bm{v}, \bm{s}) \vert \bm{s} \bigr].
\end{equation*}
The lemma then follows by summing these two inequalities.

\paragraph{Step 1. Comparing revenue generated from Buyer 2.}
For any auction $\pi$, Fubini's theorem implies that,
\begin{equation}
\label{eq:Fubini_B2}
    \mathbb{E}\bigl[p_2^{\pi} (\bm{v}, \bm{s}) \vert \bm{s} \bigr] = \mathbb{E}_{v_1}\bigl[ \mathbb{E}_{v_2}\bigl[p_2^{\pi} (\bm{v}, \bm{s}) \vert \bm{s}; v_1 \bigr]  \vert \bm{s} \bigr]. 
\end{equation}
To derive our result, we will prove the following stronger property: for every $v_1 \in [a,b]$.
\begin{equation*}
        \mathbb{E}_{v_2}\bigl[p_2^{\pi_C} (\bm{v}, \bm{s}) \vert \bm{s}; v_1 \bigr]\geq \mathbb{E}_{v_2}\bigl[p_2^{\pi_E} (\bm{v}, \bm{s}) \vert \bm{s}; v_1 \bigr]. 
\end{equation*}
We will establish this by considering three cases.

\emph{Case (i): $v_1 < s_1.$} Given that $r_1^{E} = s_1$ and $r_1^{C} \geq s_1$, buyer $1$ is inactive in both auctions. The good is allocated to buyer $2$ whenever its value is above the reserve set by the auctioneer and, the payment rule satisfies,
\begin{equation*}
    p_2^{\pi_C} (\bm{v}, \bm{s}) = r_2^C \cdot \mathbbm{1} \{ v_2 \geq r_2^C \} \qquad \mbox{and} \qquad p_2^{\pi_E} (\bm{v}, \bm{s})  = s_2 \cdot \mathbbm{1} \{ v_2 \geq s_2 \}.
\end{equation*}
Consequently, we obtain that,
\begin{equation*}
    \mathbb{E}_{v_2}\bigl[p_2^{\pi_C} (\bm{v}, \bm{s}) \vert \bm{s}; v_1 \bigr] = \Rev(r_2^C,F_2) = \Rev(p_2^*,F_2) \geq \Rev(s_2,F_2) = \mathbb{E}_{v_2}\bigl[p_2^{\pi_E} (\bm{v}, \bm{s}) \vert \bm{s}; v_1 \bigr],
\end{equation*}
where the inequality follows from \eqref{eq:monopoly-optimality}.

\emph{Case (ii): $v_1 \in [s_1, r_1^C)$.} In this case, buyer $1$ is active in $\pi_{E}$, but inactive in $\pi^{C}$.
In $\pi_{C}$, the payment rule remains $p_2^{\pi_C} (\bm{v}, \bm{s}) = r_2^C \cdot \mathbbm{1} \{ v_2 \geq r_2^C \}.$ In $\pi_{E}$, the good is allocated to buyer $2$ if and only if $v_2 \geq s_2$ and $v_2 \geq v_1$ and the payment is,
\begin{equation*}
    p_2^{\pi_E} (\bm{v}, \bm{s})  = \max(s_2,v_1) \cdot \mathbbm{1} \{ v_2 \geq s_2; \, v_2 \geq v_1  \} = \max(s_2,v_1) \cdot \mathbbm{1} \{ v_2 \geq \max(s_2,v_1) \}. 
\end{equation*}
Consequently, we obtain that
\begin{equation*}
    \mathbb{E}_{v_2}\bigl[p_2^{\pi_C} (\bm{v}, \bm{s}) \vert \bm{s}; v_1 \bigr] = \Rev(r_2^C,F_2) = \Rev(p_2^*,F_2) \geq \Rev(\max(s_2,v_1),F_2) = \mathbb{E}_{v_2}\bigl[p_2^{\pi_E} (\bm{v}, \bm{s}) \vert \bm{s}; v_1 \bigr],
\end{equation*}
where the inequality follows from \eqref{eq:monopoly-optimality}.

\emph{Case (iii): $v_1 \geq r_1^{C}$.} In this case, buyer $1$ is active in both auctions.
In $\pi_{E}$, the payment rule is the same as case $(ii)$ and, $p_2^{\pi_E} (\bm{v}, \bm{s}) = \max(s_2,v_1) \cdot \mathbbm{1} \{ v_2 \geq \max(s_2,v_1) \}.$ In $\pi_{C}$, the good is allocated to buyer $2$ if and only if $v_2 \geq r_2^C$ and $v_2 \geq v_1$ and the payment is,
\begin{equation*}
    p_2^{\pi_C} (\bm{v}, \bm{s})  = \max(r_2^C,v_1) \cdot \mathbbm{1} \{ v_2 \geq r_2^C; \, v_2 \geq v_1  \} = \max(r_2^C,v_1) \cdot \mathbbm{1} \{ v_2 \geq \max(r_2^C,v_1) \}. 
\end{equation*}
Consequently, we obtain that
\begin{equation*}
\begin{aligned}
    \mathbb{E}_{v_2}\bigl[p_2^{\pi_C} (\bm{v}, \bm{s}) \vert \bm{s}; v_1 \bigr] &= \Rev(\max(r_2^C,v_1),F_2),\\
    \mathbb{E}_{v_2}\bigl[p_2^{\pi_E} (\bm{v}, \bm{s}) \vert \bm{s}; v_1 \bigr] &= \Rev(\max(s_2,v_1),F_2).
\end{aligned}
\end{equation*}

Furthermore, we note that $v_1 \geq r_1^{C}$ implies that $v_1 \geq s_1 \geq s_2$, as $r_1^{C} \geq s_1$. Hence,
\begin{equation*}
    \mathbb{E}_{v_2}\bigl[p_2^{\pi_E} (\bm{v}, \bm{s}) \vert \bm{s}; v_1 \bigr] = \Rev(v_1,F_2) \leq \Rev(\max(r_2^C,v_1),F_2) = \mathbb{E}_{v_2}\bigl[p_2^{\pi_C} (\bm{v}, \bm{s}) \vert \bm{s}; v_1 \bigr],
\end{equation*}
where the inequality holds because either $v_1 < r_2^C$ and we conclude by \eqref{eq:monopoly-optimality}, or $v_1 \geq r_2^C$ and both terms are equal.

Hence, by combining the three cases and applying \eqref{eq:Fubini_B2}, we obtain that,
\begin{equation}
\label{eq:ccl_buyer2}
        \mathbb{E}\bigl[p_2^{\pi_C} (\bm{v}, \bm{s}) \vert \bm{s} \bigr] = \mathbb{E}_{v_1}\bigl[ \mathbb{E}_{v_2}\bigl[p_2^{\pi_C} (\bm{v}, \bm{s}) \vert \bm{s}; v_1 \bigr]  \vert \bm{s} \bigr] \geq  \mathbb{E}_{v_1}\bigl[ \mathbb{E}_{v_2}\bigl[p_2^{\pi_E} (\bm{v}, \bm{s}) \vert \bm{s}; v_1 \bigr]  \vert \bm{s} \bigr] =    \mathbb{E}\bigl[p_2^{\pi_E} (\bm{v}, \bm{s}) \vert \bm{s} \bigr].
\end{equation}

\paragraph{Step 2. Comparing revenue generated from Buyer 1.}
We now establish that for every $v_2 \in [a,b]$.
\begin{equation}
\label{eq:to_prove_buyer1}
        \mathbb{E}_{v_1}\bigl[p_1^{\pi_C} (\bm{v}, \bm{s}) \vert \bm{s}; v_2 \bigr]\geq \mathbb{E}_{v_1}\bigl[p_1^{\pi_E} (\bm{v}, \bm{s}) \vert \bm{s}; v_2 \bigr]. 
\end{equation}

We first assume that $r_2^{C} \leq s_2.$ Let $v_2 \in [a,b].$

\emph{Case (i): $v_2 < r_2^{C}$.} Buyer $2$ is inactive in both auctions. The good is allocated to buyer $1$ whenever its value is above the reserve set by the auctioneer and, the payment rule satisfies,
\begin{equation*}
    p_1^{\pi_C} (\bm{v}, \bm{s}) = r_1^C \cdot \mathbbm{1} \{ v_1 \geq r_1^C \} \qquad \mbox{and} \qquad p_1^{\pi_E} (\bm{v}, \bm{s})  = s_1 \cdot \mathbbm{1} \{ v_1 \geq s_1 \}.
\end{equation*}
Consequently, we obtain that,
\begin{equation*}
    \mathbb{E}_{v_1}\bigl[p_1^{\pi_C} (\bm{v}, \bm{s}) \vert \bm{s}; v_2 \bigr] = \Rev(r_1^C,F_1) = \Rev(\max(s_1,p_1^*),F_1) \geq \Rev(s_1,F_1) = \mathbb{E}_{v_1}\bigl[p_1^{\pi_E} (\bm{v}, \bm{s}) \vert \bm{s}; v_2 \bigr],
\end{equation*}
where the inequality follows from \eqref{eq:monopoly-optimality}, or because both terms are equal.

\emph{Case (ii): $v_2 \in [r_2^{C}, s_2)$.}  In this case, buyer $2$ is active in $\pi_{C}$, but inactive in $\pi_{E}$.
In $\pi_{E}$, the payment rule remains $p_1^{\pi_E} (\bm{v}, \bm{s}) =s_1 \cdot \mathbbm{1} \{ v_1 \geq s_1 \}.$ In $\pi_{C}$, the good is allocated to buyer $1$ if and only if $v_1 \geq r_1^{C}$ and $v_1 \geq v_2.$ and the payment is,
\begin{equation*}
    p_1^{\pi_C} (\bm{v}, \bm{s})  =  \max(r_1^{C},v_2) \cdot \mathbbm{1} \{ v_1 \geq \max(r_1^{C},v_2) \}. 
\end{equation*}
Remark that $v_2 < s_2 \leq s_1 \leq r_1^{C}$. Hence, $\max(r_1^{C},v_2) = r_1^C.$
Consequently, by using the same argument as in case (i) we conclude that
\begin{equation*}
    \mathbb{E}_{v_1}\bigl[p_1^{\pi_C} (\bm{v}, \bm{s}) \vert \bm{s}; v_2 \bigr] = \Rev(r_1^{C},F_1) \geq \Rev(s_1,F_1)  = \mathbb{E}_{v_1}\bigl[p_1^{\pi_E} (\bm{v}, \bm{s}) \vert \bm{s}; v_2 \bigr].
\end{equation*}

\emph{Case (iii): $v_2 \geq s_2$.} In this case, buyer $2$ is active in both auctions.
In $\pi_{C}$, the payment rule is the same as case $(ii)$ and, $p_1^{\pi_C} (\bm{v}, \bm{s}) = \max(r_1^{C},v_2) \cdot \mathbbm{1} \{ v_1 \geq \max(r_1^{C},v_2) \}.$ In $\pi_{E}$, the good is allocated to buyer $1$ if and only if $v_1 \geq s_1$ and $v_1 \geq v_2$ and the payment is
\begin{equation*}
    p_1^{\pi_E} (\bm{v}, \bm{s})  =  \max(s_1,v_2) \cdot \mathbbm{1} \{ v_1 \geq \max(s_1,v_2) \}. 
\end{equation*}
Consequently, we obtain that
\begin{equation*}
    \mathbb{E}_{v_1}\bigl[p_1^{\pi_C} (\bm{v}, \bm{s}) \vert \bm{s}; v_2 \bigr] = \Rev(\max(r_1^C,v_2),F_1) \geq \Rev(\max(s_1,v_2),F_1) = \mathbb{E}_{v_1}\bigl[p_1^{\pi_E} (\bm{v}, \bm{s}) \vert \bm{s}; v_2 \bigr],
\end{equation*}
where the inequality holds because
\begin{multline*}
    \Rev(\max(r_1^C,v_2),F_1) - \Rev(\max(s_1,v_2),F_1)\\
    =\begin{cases}
        \Rev(r_1^C,F_1) - \Rev(s_1,F_1) \stackrel{(a)}{\geq} 0, & \text{if $v_2 < s_1$,}\\
        \Rev(r_1^C,F_1) - \Rev(v_2,F_1) \stackrel{(b)}{\geq} 0, & \text{if $v_2 \in [s_1, r_1^C)$,}\\
        \Rev(v_2,F_1) - \Rev(v_2,F_1) = 0, & \text{if $v_2 \geq r_1^C$,}
    \end{cases}
\end{multline*}
where $(a)$ follows from the same argument as in case (i) and $(b)$ holds because $v_2 \in [s_1, r_1^C)$ is non-empty if and only if $r_1^C > s_1$, in which case, $\Rev(r_1^C,F_1) - \Rev(v_2,F_1) = \Rev(p_1^*,F_1) - \Rev(v_2,F_1) \geq 0$, by~\eqref{eq:monopoly-optimality}.

By combining the three cases we conclude similarly to step 1, that \eqref{eq:to_prove_buyer1} holds.

Next, let us treat the case, $r_2^C > s_2$. Let $v_2 \in [a,b]$.

\emph{Case (i'): $v_2 < s_2.$} Similarly to $(i)$, we can show that     
\begin{equation*}
    \mathbb{E}_{v_1}\bigl[p_1^{\pi_C} (\bm{v}, \bm{s}) \vert \bm{s}; v_2 \bigr] = \Rev(r_1^C,F_1) = \Rev(\max(s_1,p_1^*),F_1) \geq \Rev(s_1,F_1) = \mathbb{E}_{v_1}\bigl[p_1^{\pi_E} (\bm{v}, \bm{s}) \vert \bm{s}; v_2 \bigr].
\end{equation*}

\emph{Case (ii'): $v_2 \in [s_2,r_2^{C}).$} In this case, buyer $2$ is active in $\pi_{E}$, but inactive in $\pi_{C}$.
In $\pi_{C}$, the payment rule remains $p_1^{\pi_C} (\bm{v}, \bm{s}) =r_1^{C} \cdot \mathbbm{1} \{ v_1 \geq r_1^{C} \}.$ In $\pi_{E}$, the good is allocated to buyer $1$ if and only if $v_1 \geq s_1$ and $v_1 \geq v_2.$ and the payment is,
\begin{equation*}
    p_1^{\pi_E} (\bm{v}, \bm{s})  =  \max(s_1,v_2) \cdot \mathbbm{1} \{ v_1 \geq \max(s_1,v_2) \}. 
\end{equation*}
Hence,
\begin{equation*}
    \mathbb{E}_{v_1}\bigl[p_1^{\pi_C} (\bm{v}, \bm{s}) \vert \bm{s}; v_2 \bigr] = \Rev(r_1^C,F_1) \qquad \mbox{and} \qquad    \mathbb{E}_{v_1}\bigl[p_1^{\pi_E} (\bm{v}, \bm{s}) \vert \bm{s}; v_2 \bigr] = \Rev(\max(s_1,v_2),F_1).
\end{equation*}
If $s_1 \leq p_1^*$, we have that $r_1^C = p_1^*$ and \eqref{eq:monopoly-optimality}, implied that $\Rev(r_1^C,F_1)  \geq \Rev(\max(s_1,v_2),F_1).$ Furthermore, if $s_1 \geq p_1^*$ and $v_2 \leq s_1$, we have that $\Rev(r_1^C,F_1) = \Rev(s_1,F_1) = \Rev(\max(s_1,v_2),F_1)$. Hence, the last setting to argue is the case in which $p_1^* < s_1 < v_2 \leq r_2^C.$ We next show that such a setting is not possible.
By property $(ii)$ in \Cref{lem:order-price-signal}, there exists  $p^{\mathrm{ignore}} \in [a,b]$ such that
the relation $p_1^* < s_1$ implies that $p_1^* \geq p^{\mathrm{ignore}}$. Furthermore, property $(i)$ in \Cref{lem:order-price-signal} implies that as, $s_2 < r_2^C = p_2^*$, we must have that $p_2^* = p^{\mathrm{ignore}}$.
Combining these relations we obtain that $p^{\mathrm{ignore}} \leq p_1^* < p_2^* =p^{\mathrm{ignore}}$, which is impossible.

Therefore, we have $\mathbb{E}_{v_1}\bigl[p_1^{\pi_C} (\bm{v}, \bm{s}) \vert \bm{s}; v_2 \bigr] \geq \mathbb{E}_{v_1}\bigl[p_1^{\pi_E} (\bm{v}, \bm{s}) \vert \bm{s}; v_2 \bigr].$

\emph{Case (iii'): $v_2 \geq r_2^{C}.$}
In this case, buyer $2$ is active in both auctions.
In $\pi_{E}$, the payment rule is the same as case $(ii')$ and, $p_1^{\pi_E} (\bm{v}, \bm{s}) = \max(s_1,v_2) \cdot \mathbbm{1} \{ v_1 \geq \max(s_1,v_2) \}.$ In $\pi_{C}$, the good is allocated to buyer $1$ if and only if $v_1 \geq r_1^{C}$ and $v_1 \geq v_2$ and the payment is
\begin{equation*}
    p_1^{\pi_C} (\bm{v}, \bm{s})  =  \max(r_1^{C},v_2) \cdot \mathbbm{1} \{ v_1 \geq \max(r_1^{C},v_2) \}. 
\end{equation*}
And we conclude similarly to case $(iii)$ above.

All in all, we have established in all cases that 
\begin{equation}
\label{eq:ccl_buyer1}
        \mathbb{E}\bigl[p_1^{\pi_C} (\bm{v}, \bm{s}) \vert \bm{s} \bigr] = \mathbb{E}_{v_2}\bigl[ \mathbb{E}_{v_1}\bigl[p_1^{\pi_C} (\bm{v}, \bm{s}) \vert \bm{s}; v_2 \bigr]  \vert \bm{s} \bigr] \geq  \mathbb{E}_{v_2}\bigl[ \mathbb{E}_{v_1}\bigl[p_1^{\pi_E} (\bm{v}, \bm{s}) \vert \bm{s}; v_2 \bigr]  \vert \bm{s} \bigr] =    \mathbb{E}\bigl[p_1^{\pi_E} (\bm{v}, \bm{s}) \vert \bm{s} \bigr].
\end{equation}

Finally, by summing \eqref{eq:ccl_buyer2} and \eqref{eq:ccl_buyer1}, we obtain the desired result.

\end{proof}

\begin{lemma}
\label{lem:order-price-signal}
Assume $n=1$ and that $F$ is a regular distribution which admits a continuous, strictly positive density on $[a,b]$.
Let
\begin{equation*}
p^{\mathrm{ignore}} = \inf\{v\in[a,b]: \varphi_F(v)=0\},
\end{equation*}
and let $p^*(s)$ denote the revenue-maximizing posted price given signal $s$ (equivalently, the
smallest $v$ such that $\bar{\varphi}_{F_{\gamma,s}}(v)\ge 0$).
Then:
\begin{enumerate}
    \item[(i)] If $p^*(s) > s$, then $p^*(s)=p^{\mathrm{ignore}}$.
    \item[(ii)] If $p^*(s) < s$, then $p^*(s) \geq p^{\mathrm{ignore}}$.
\end{enumerate}
\end{lemma}

\begin{proof}[\textbf{Proof of \Cref{lem:order-price-signal}}]~\\
\noindent \textit{Step 1: Proof of (i).}
Assume that $p^*(s) > s$. By definition, 
\begin{align*}
    p^*(s) 
    &= \inf \{v \in [a,b] \, : \bar{\varphi}_{F_{\gamma,s}}(v)\ge 0 \}\\
    &\stackrel{(a)}{=}  \inf \{v \in [s,b] \, : \bar{\varphi}_{F_{\gamma,s}}(v)\ge 0 \} \\  &\stackrel{(b)}{=}  \inf \{v \in [s,b] \, : \varphi_{F}(v)\ge 0 \}\\
    &\stackrel{(c)}{=}  \inf \{v \in [s,b] \, : \varphi_{F}(v) = 0 \}\\
    &\stackrel{(d)}{=}  \inf \{v \in [a,b] \, : \varphi_{F}(v) = 0 \} = p^{\mathrm{ignore}},
\end{align*}
where $(a)$ holds because we assume that $p^*(s) > s$, $(b)$ follows from the characterization developed in \Cref{thm:main}, $(c)$ holds by continuity of the cdf and density of $F$ and $(d)$ holds because $\varphi_{F}(s) < 0$ (as $p^*(s) > s$) and $\varphi_{F}$ is non-decreasing as $F$ is regular.

\medskip
\noindent\textit{Step 2: Proof of (ii).}
Assume that $p^*(s) < s$. Assume for the sake of contradiction that $p^*(s) < p^{\mathrm{ignore}}$, and remark that there exists $\epsilon > 0$ such that $p_\epsilon := p^*(s) + \epsilon < \min(p^{\mathrm{ignore}}, s)$. Fix such $\epsilon$.
We note that,
\begin{align*}
    p^*(s) \cdot (1 -F_{\gamma,s}(p^*(s)-)) &\stackrel{(a)}{=} p^*(s) \cdot (1 -\gamma F(p^*(s)-)) \\
    &= \gamma \cdot p^*(s) \cdot (1 - F(p^*(s)-)) + (1-\gamma) \cdot p^*(s) \\
    &\stackrel{(b)}{<}\gamma \cdot p_\epsilon \cdot (1 - F(p_\epsilon-)) + (1-\gamma) \cdot p_\epsilon\\
    &\stackrel{(c)}{=} p_\epsilon \cdot (1 -F_{\gamma,s}(p_\epsilon-)),
\end{align*}
where $(a)$ and $(c)$ follow from the definition of $F_{\gamma,s}$ on the interval $[a,s)$. To establish $(b)$ remark that, by regularity of $F$, the mapping $p \mapsto p \cdot (1-F(p-))$ is unimodal and non-decreasing before $p^{\mathrm{ignore}}$. Given that $p^*(s) < p_\epsilon < p^{\mathrm{ignore}}$ we have that $p^*(s) \cdot (1-F(p^*(s)-)) \leq p_\epsilon \cdot (1 - F(p_\epsilon-))$. The inequality in $(b)$ is strict because of the second term.

This contradicts optimality of $p^*(s)$. Hence, $p^*(s) \geq p^{\mathrm{ignore}}$.
\end{proof}

\section{A Density-Dependent Hallucination Probability}
\label{sec:apx_ddh}

Throughout the paper, we have taken the hallucination probability $\gamma$ to be a constant that does not depend on the buyer's value; we refer to the model of \Cref{sec:model} as the \emph{constant-$\gamma$} model in this appendix. A natural refinement reflects the fact that a prediction technology is typically less reliable at values it encounters less often: contexts that are poorly represented in the training data are those where the output is most likely to be uninformative. In this appendix, we let the hallucination probability vary with the value through the prior density and show that our qualitative conclusions remain valid.

Formally, the model we study in this appendix replaces the constant $\gamma$ with the value-dependent hallucination probability
\begin{equation}\label{eq:ddh_gamma_v}
\mathbb{P}(h = 1 \mid v) \;=\; \gamma(v) \;=\; \gamma \cdot \frac{\underline f}{f(v)},
\qquad \underline f = \inf_{v \in [a,b]} f(v),
\end{equation}
where $\gamma \in (0,1)$ and, conditional on the hallucination event $\{h=1\}$, the signal is an independent draw $w \sim F$ exactly as before. By construction $\gamma(v) \le \gamma$ for every $v$, with equality at the rarest values (where $f$ attains its minimum $\underline f$), and the hallucination probability decreases proportionally to the density at more common values. This model requires $\underline f > 0$, i.e., a density bounded away from zero on its support.

Conditional on $v$, the signal is distributed as $(1-\gamma(v))\cdot \delta_v + \gamma(v) \cdot F$, so the joint law of $(v,s)$ in the zero-noise benchmark is
\begin{equation}\label{eq:ddh_joint}
f(v)\,\bigl[(1-\gamma(v))\,\delta_{v}(s) + \gamma(v) f(s)\bigr]
\;=\;
\bigl(f(v) - \gamma \underline f\bigr)\,\delta_{v}(s) \;+\; \gamma \underline f\, f(s),
\end{equation}
using the identity $\gamma(v) f(v) = \gamma \underline f$.

Integrating \eqref{eq:ddh_joint} over $v$ gives the marginal density of the signal, $m(s) = \bigl(1 + (b-a)  \gamma \underline f \bigr) f(s) - \gamma \underline f$, and dividing \eqref{eq:ddh_joint} by $m(s)$ gives the posterior:
\begin{equation}\label{eq:ddh_posterior}
f_{\gamma(\cdot),s}(v) \;=\; \widetilde\gamma(s) \cdot \frac{1}{b-a} \;+\; \bigl(1 - \widetilde\gamma(s)\bigr)\cdot \delta_s(v),
\qquad
\widetilde\gamma(s) \;:=\; \frac{\gamma \underline f\,(b-a)\, f(s)}{m(s)} \;\in\; (0,1).
\end{equation}
The posterior \eqref{eq:ddh_posterior} has exactly the constant-$\gamma$ form \eqref{eq:density-f}, with the prior replaced by the uniform distribution on $[a,b]$ and the hallucination probability by the signal-dependent $\widetilde\gamma(s)$. Since \Cref{thm:main} depends on the model only through posteriors of this form, its characterization of the optimal auction applies verbatim, with $F_i$ replaced by the uniform distribution on $[a_i,b_i]$ and $\gamma_i$ by $\widetilde\gamma_i(s_i)$.

\subsection{Optimal Mechanism for One Buyer}
\label{sec:apx_ddh_price}

We now derive the analogue of \Cref{prop:optimal_price_corrected}. We use the revenue notation of \Cref{sec:opt_1b} and let $F_{\gamma(\cdot),s}$ denote the cumulative distribution of the posterior \eqref{eq:ddh_posterior}.

\begin{proposition}\label{prop:ddh_optimal_price}
Let $n = 1$ and $\sigma = 0$, let $F$ admit a strictly positive density $f$ on $[a,b]$ with $a \le b/2$ and $\underline f > 0$, let the hallucination probability be given by \eqref{eq:ddh_gamma_v} for some $\gamma \in (0,1)$, and define
\begin{equation}\label{eq:ddh_pcap}
p^{\mathrm{cap}}(s) \;=\; \frac{b}{2} \;+\; \frac{1}{2\gamma \underline f} \;-\; \frac{1}{2 f(s)}.
\end{equation}
Then for every $s \in [a,b]$ the following posted price is optimal:
\begin{equation*}
p^*(s) \;=\;
\begin{cases}
b/2, & \text{if } s < b/2 \text{ and } \Rev\bigl(b/2,\, F_{\gamma(\cdot),s}\bigr) > \Rev\bigl(\min\{s,\, p^{\mathrm{cap}}(s)\},\, F_{\gamma(\cdot),s}\bigr),\\[2pt]
\min(s,p^{\mathrm{cap}}(s)) & \text{otherwise.}
\end{cases}
\end{equation*}
\end{proposition}

\Cref{prop:ddh_optimal_price} shows that the optimal price retains a structure similar to the three-piece property of \Cref{prop:optimal_price_corrected}: the seller ignores, follows, or caps. The main difference is that $p^{\mathrm{cap}}(s)$ is not monotone in $s$, so the optimal price may alternate between follow and cap regimes.

\begin{proof}[\textbf{Proof of \Cref{prop:ddh_optimal_price}}]
Fix $s \in [a,b]$ and write $c = \widetilde\gamma(s)/(b-a)$ so that
\begin{equation*}
F_{\gamma(\cdot),s}(v) \;=\; c \cdot (v - a) \;+\; \bigl(1 - \widetilde\gamma(s)\bigr)\cdot \mathbbm{1}\{v \ge s\}.
\end{equation*}
Furthermore, note that $1/c = (b-a) + 1/(\gamma \underline f) - 1/f(s)$, which gives $(a + 1/c)/2 = p^{\mathrm{cap}}(s)$. Since $\gamma \underline f \le f(s)$, we have $1/c \ge b-a$ and therefore $p^{\mathrm{cap}}(s) \ge (a + b - a)/2 = b/2$.

Consider first prices $p \in [a,s]$. We have that $1 - F_{\gamma(\cdot),s}(p^-) = 1 - c(p-a)$ and $\Rev(p, F_{\gamma(\cdot),s}) = p\,(1 - c(p-a))$, a strictly concave quadratic in $p$ maximized at $(a + 1/c)/2 = p^{\mathrm{cap}}(s)$. Since $p^{\mathrm{cap}}(s) \ge b/2 \ge a$ and by concavity, we conclude that the best price in $[a,s]$ is $\min\{s, p^{\mathrm{cap}}(s)\}$.

Now assume $p \in (s,b]$. There, $1 - F_{\gamma(\cdot),s}(p^-) = c\,(b - p)$ and $\Rev(p, F_{\gamma(\cdot),s}) = p\, c\, (b-p)$, a strictly concave quadratic maximized at $b/2$. If $s \ge b/2$, this quadratic is strictly decreasing on $(s,b]$, so every price $p > s$ earns strictly less than $s\,c\,(b-s)$; since the atom has positive mass ($\widetilde\gamma(s) < 1$ because $\gamma < 1$), we have $s\,c\,(b-s) < s\,(1 - c(s-a)) = \Rev(s, F_{\gamma(\cdot),s})$, and every price above $s$ is dominated by $s$ itself. If $s < b/2$, the best price in $(s,b]$ is $b/2$.

Combining the two cases, an optimal price is whichever of $\min\{s, p^{\mathrm{cap}}(s)\}$ and, when $s < b/2$, the price $b/2$ yields the larger posterior revenue, which is the statement.
\end{proof}

\subsection{Numerical Results}
\label{sec:apx_ddh_numerics}

We now illustrate the mechanism for two priors on $[0.1,0.9]$: a truncated $\mathrm{Beta}(2,5)$, and a truncated equal-weight mixture of $\mathrm{Beta}(6,2)$ and $\mathrm{Beta}(2,6)$ with a bimodal density. \Cref{fig:ddh_price_rules} compares the optimal posted price $p^*(s)$ of \Cref{prop:ddh_optimal_price} with the optimal constant-$\gamma$ prices of \Cref{prop:optimal_price_corrected}, at the same peak $\gamma$ and at the matched average $\gamma_{\mathrm{eff}} = \mathbb{E}\bigl[\gamma(v)\bigr] = \gamma\,\underline f\,(b-a)$, which equals $0.0024\,\gamma$ for the truncated $\mathrm{Beta}(2,5)$ and $0.62\,\gamma$ for the bimodal prior.

\begin{figure}[h!]
\centering
\pgfplotsset{
    every ddh price axis/.style={
        width=0.42\textwidth, height=0.36\textwidth,
        xlabel={$s$},
        table/col sep=comma,
        grid=both,
        tick label style={font=\footnotesize, /pgf/number format/fixed},
        label style={font=\footnotesize},
        title style={font=\footnotesize},
    }
}
\begin{tikzpicture}
\begin{axis}[
    every ddh price axis,
    ylabel={posted price},
    title={$\mathrm{Beta}(2,5)$ tr., $\gamma=0.6$},
    xmin=0.1, xmax=0.9, ymin=0.1, ymax=0.9,
    legend to name=leg:ddh_price_rules,
    legend columns=3,
    legend style={font=\scriptsize,draw=none,fill=none,column sep=8pt},
    legend cell align=left,
]
\addplot[black,line width=1.3pt] table[x=s, y=p_dd_0.6] {figures/data/ddh_prices_beta25.csv};
\addlegendentry{Density-dependent $\gamma(v)$}
\addplot[RoyalBlue,thick,mark=square*,mark size=1.2pt,mark repeat=60,mark phase=20] table[x=s, y=p_const_0.6] {figures/data/ddh_prices_beta25.csv};
\addlegendentry{Constant $\gamma$ (same peak)}
\addplot[BrickRed,thick,dashed,mark=*,mark size=1.2pt,mark options={solid},mark repeat=60,mark phase=45] table[x=s, y=p_ceff_0.6] {figures/data/ddh_prices_beta25.csv};
\addlegendentry{Constant $\gamma_{\mathrm{eff}}$ (matched average)}
\end{axis}
\end{tikzpicture}
\hfill
\begin{tikzpicture}
\begin{axis}[every ddh price axis, title={$\mathrm{Beta}(2,5)$ tr., $\gamma=0.95$}, xmin=0.1, xmax=0.9, ymin=0.1, ymax=0.9]
\addplot[black,line width=1.3pt] table[x=s, y=p_dd_0.95] {figures/data/ddh_prices_beta25.csv};
\addplot[RoyalBlue,thick,mark=square*,mark size=1.2pt,mark repeat=60,mark phase=20] table[x=s, y=p_const_0.95] {figures/data/ddh_prices_beta25.csv};
\addplot[BrickRed,thick,dashed,mark=*,mark size=1.2pt,mark options={solid},mark repeat=60,mark phase=45] table[x=s, y=p_ceff_0.95] {figures/data/ddh_prices_beta25.csv};
\end{axis}
\end{tikzpicture}

\begin{tikzpicture}
\begin{axis}[every ddh price axis, ylabel={posted price}, title={Bimodal, $\gamma=0.6$}, xmin=0.1, xmax=0.9, ymin=0.1, ymax=0.9]
\addplot[black,line width=1.3pt] table[x=s, y=p_dd_0.6] {figures/data/ddh_prices_bimodal.csv};
\addplot[RoyalBlue,thick,mark=square*,mark size=1.2pt,mark repeat=60,mark phase=20] table[x=s, y=p_const_0.6] {figures/data/ddh_prices_bimodal.csv};
\addplot[BrickRed,thick,dashed,mark=*,mark size=1.2pt,mark options={solid},mark repeat=60,mark phase=45] table[x=s, y=p_ceff_0.6] {figures/data/ddh_prices_bimodal.csv};
\end{axis}
\end{tikzpicture}
\hfill
\begin{tikzpicture}
\begin{axis}[every ddh price axis, title={Bimodal, $\gamma=0.95$}, xmin=0.1, xmax=0.9, ymin=0.1, ymax=0.9]
\addplot[black,line width=1.3pt] table[x=s, y=p_dd_0.95] {figures/data/ddh_prices_bimodal.csv};
\addplot[RoyalBlue,thick,mark=square*,mark size=1.2pt,mark repeat=60,mark phase=20] table[x=s, y=p_const_0.95] {figures/data/ddh_prices_bimodal.csv};
\addplot[BrickRed,thick,dashed,mark=*,mark size=1.2pt,mark options={solid},mark repeat=60,mark phase=45] table[x=s, y=p_ceff_0.95] {figures/data/ddh_prices_bimodal.csv};
\end{axis}
\end{tikzpicture}

\medskip
\pgfplotslegendfromname{leg:ddh_price_rules}
\caption{\textbf{Optimal posted price under the density-dependent hallucination probability.} }
\label{fig:ddh_price_rules}
\end{figure}

For the truncated $\mathrm{Beta}(2,5)$ prior, the density-dependent rule essentially follows the signal everywhere at every $\gamma$, while the constant-$\gamma$ rule at the peak develops ignore and cap regions as $\gamma$ grows. For the bimodal prior at $\gamma = 0.95$, the density-dependent rule exhibits \emph{five} regimes---ignore on $[0.10,0.13]$, cap on the windows $[0.49,0.62]$ and $[0.83,0.90]$, follow in between---and within each cap window the posted price is non-monotonic as a function of the signal. Even in this case, however, the constant-$\gamma_{\mathrm{eff}}$ price remains relatively close to the density-dependent rule.

\end{document}